\documentclass[sigconf,natbib]{acmart}
\pdfoutput=1
\PassOptionsToPackage{numbers, compress}{natbib}

\usepackage{amsmath,amsfonts}
\usepackage{algorithm}
\usepackage{algpseudocode}
\usepackage{wrapfig, blindtext}
\usepackage[export]{adjustbox}
\usepackage{graphicx}
\usepackage{textcomp}

\usepackage[utf8]{inputenc} % allow utf-8 input
\usepackage[T1]{fontenc}    % use 8-bit T1 fonts
\usepackage{hyperref}       % hyperlinks
\usepackage{url}            % simple URL typesetting
\usepackage{booktabs}       % professional-quality tables
\usepackage{amsfonts}       % blackboard math symbols
\usepackage{nicefrac}       % compact symbols for 1/2, etc.
\usepackage{microtype}      % microtypography
\usepackage{xcolor}         % colors

\usepackage{diagbox}
\usepackage{mathtools}
\usepackage{amsmath}
\usepackage{graphicx}
\usepackage{pifont}
\usepackage{multirow}
\usepackage{multicol}
\usepackage{threeparttable}
\usepackage{comment}
\usepackage{subcaption}
\usepackage{wrapfig, blindtext}
\usepackage[export]{adjustbox}

\usepackage{contour}
\usepackage{ulem}[normalem]
\usepackage{dsfont}
\usepackage{balance}

\newcommand{\ul}[1]{%
  \uline{\phantom{#1}}%
  \llap{\contour{white}{#1}}%
}

\usepackage{enumitem}
\usepackage{xcolor}
\usepackage{tikz}
\usepackage{pgfplots}
\usepackage[multiple]{footmisc}
\usepackage{multirow}
\usepackage{bm}
\usepackage{booktabs}
\usepackage{hyperref}
\usepackage{array}
\usepackage{siunitx}

\usepackage{tabularx}
\usepackage{geometry}
\usepackage{lipsum}

\DeclarePairedDelimiter\floor{\lfloor}{\rfloor}

\pgfplotsset{% https://tex.stackexchange.com/a/75811/194703
    name nodes near coords/.style={
        every node near coord/.append style={
            name=#1-\coordindex,
            alias=#1-last,
        },
    },
    name nodes near coords/.default=coordnode
}

\pgfplotsset{compat=1.5.1}
\newenvironment{customlegend}[1][]{%
  \begingroup
  \csname pgfplots@init@cleared@structures\endcsname
  \pgfplotsset{#1}%
}{%
  \csname pgfplots@createlegend\endcsname
  \endgroup
}%
\def\addlegendimage{\csname pgfplots@addlegendimage\endcsname}
\usetikzlibrary{patterns}
\usetikzlibrary{pgfplots.groupplots}
\usetikzlibrary{
  pgfplots.colorbrewer,
}

\definecolor{aa}{rgb}{0.2,0.7,0.310}
\definecolor{cc}{rgb}{1.0,0.49,0.0}
\definecolor{bb}{rgb}{0.514,0.325,0.831}

\def\E{\mathcal{E}}
\def\D{\mathcal{D}}
\def\U{\mathcal{U}}
\def\I{\mathcal{I}}
\def\W{\mathcal{W}}
\def\C{\mathcal{C}}

\def\P{\mathbf{P}}
\def\Y{\mathbf{Y}}
\def\T{{\scriptscriptstyle\mathsf{T}}}

\newcommand{\PD}{PD}

\def\diffpd{\operatorname{DPD}}%{\operatorname{PD}^{\text{diff}}}
\def\pd{\operatorname{PD}}

\def\dh{\operatorname{DH}}
\newcommand{\ours}{\textsc{\textsf{FADE}}}
\newcommand{\oursabs}{\textsc{\textsf{FADE-Abs}}}

\newcommand{\pretrain}{\textsc{\textsf{Pretrain}}}
\newcommand{\pretrainfair}{\textsc{\textsf{Pretrain-Fair}}}
\newcommand{\full}{\textsc{\textsf{Retrain}}}
\newcommand{\fullfair}{\textsc{\textsf{Retrain-Fair}}}
\newcommand{\fine}{\textsc{\textsf{Finetune}}}
\newcommand{\finefair}{\ours}

\newcommand{\adver}{\textsc{\textsf{Adver}}}

\newcommand{\rerank}{\textsc{\textsf{Rerank}}}

\newtheorem{prob}{Problem}

\newtheorem{defn}{Definition}

\copyrightyear{2024}
\acmYear{2024}
\setcopyright{acmlicensed}
\acmConference[WWW '24] {Proceedings of the ACM Web Conference 2024}{May 13--17, 2024}{Singapore, Singapore.}
\acmBooktitle{Proceedings of the ACM Web Conference 2024 (WWW '24), May 13--17, 2024, Singapore}
\acmISBN{979-8-4007-0171-9/24/05}
\acmDOI{10.1145/3589334.3645536}
\newcommand\TODO[1][]{{\color{orange}[TODO\ifthenelse{\equal{#1}{}}{}{: #1}]}}
\newcommand\SEC\section
\newcommand\SSEC\subsection
\newcommand\SSSEC\subsubsection

\newcommand\RM[1]{\mathrm{#1}}

\newcommand\EQ[1]{\begin{equation}#1\end{equation}}

\newcommand\BM[1]{\boldsymbol{#1}}
\newcommand\BB[1]{\mathbb{#1}}
\newcommand\CAL[1]{\mathcal{#1}}

\newcommand\OP[1]{\operatorname{#1}}

\newcommand\AL[1]{\begin{align}#1\end{align}}
\newcommand\AM[1]{\begin{align*}#1\end{align*}}

\newcommand\TLD[1]{\widetilde{#1}}
\newcommand\Prb{\mathbb{P}}
\newcommand\Exp{\mathbb{E}}

\newtheorem{DEF}{\textbf{Definition}}
\newtheorem{ASS}{\textbf{Assumption}}

\begin{document}

\title{Ensuring User-side Fairness in Dynamic Recommender Systems}


% \author{Hyunsik Yoo, Zhichen Zeng, Jian Kang, Zhining Liu, David Zhou, Fei Wang, Eunice Chan, and Hanghang Tong}

% \email{{hy40, zhichenz, jian, zhining, david, fei, eunice, htong}@illinois.edu}
% % \orcid{0000-0003-1828-2109}
% \affiliation{%
%   \institution{University of Illinois at Urbana-Champaign}
%   % \city{Urbana}
%   % \state{Illinois}
%   % \country{USA}
% }

\author{Hyunsik Yoo}
% \orcid{0000-0001-5253-5646}
\affiliation{
  \institution{University of Illinois Urbana-Champaign}
  \city{}
  \state{}
  \country{}
}
\email{hy40@illinois.edu}

\author{Zhichen Zeng}
\affiliation{
  \institution{University of Illinois Urbana-Champaign}
    \city{}
  \state{}
  \country{}
}
\email{zhichenz@illinois.edu}

\author{Jian Kang}
\affiliation{
  \institution{University of Rochester}
   \city{}
  \state{}
  \country{}
}
\email{jian.kang@rochester.edu}

\author{Ruizhong Qiu}
\affiliation{
  \institution{University of Illinois Urbana-Champaign}
   \city{}
  \state{}
  \country{}
}
\email{rq5@illinois.edu}

\author{David Zhou}
\affiliation{
  \institution{University of Illinois Urbana-Champaign}
    \city{}
  \state{}
  \country{}
}
\email{david23@illinois.edu}

\author{Zhining Liu}
\affiliation{
  \institution{University of Illinois Urbana-Champaign}
    \city{}
  \state{}
  \country{}
}
\email{liu326@illinois.edu}

\author{Fei Wang}
\affiliation{
  \institution{Amazon.com, Inc.}
    \city{}
  \state{}
  \country{}
}
\email{feiww@amazon.com}

\author{Charlie Xu}
\affiliation{
  \institution{Amazon.com, Inc.}
    \city{}
  \state{}
  \country{}
}
\email{caizhx@amazon.com}

\author{Eunice Chan}
\affiliation{
  \institution{University of Illinois Urbana-Champaign}
    \city{}
  \state{}
  \country{}
}
\email{ecchan2@illinois.edu}

\author{Hanghang Tong}
\affiliation{
  \institution{University of Illinois Urbana-Champaign}
    \city{}
  \state{}
  \country{}
}
\email{htong@illinois.edu}

\renewcommand{\shortauthors}{Hyunsik Yoo et al.}

\begin{abstract}
User-side group fairness is crucial for modern recommender systems, alleviating performance disparities among user groups defined by sensitive attributes like gender, race, or age. In the ever-evolving landscape of user-item interactions, continual adaptation to newly collected data is crucial for recommender systems to stay aligned with the latest user preferences. However, we observe that such continual adaptation often worsen performance disparities. This necessitates a thorough investigation into user-side fairness in dynamic recommender systems. This problem is challenging due to distribution shifts, frequent model updates, and non-differentiability of ranking metrics. To our knowledge, this paper presents the first principled study on ensuring user-side fairness in dynamic recommender systems. We start with theoretical analyses on fine-tuning v.s.\ retraining, showing that the best practice is incremental fine-tuning with restart. Guided by our theoretical analyses, we propose \ul{FA}ir \ul{D}ynamic r\ul{E}commender (\ours), an end-to-end fine-tuning framework to dynamically ensure user-side fairness over time. To overcome the non-differentiability of recommendation metrics in the fairness loss, we further introduce Differentiable Hit (DH) as an improvement over the recent NeuralNDCG method, not only alleviating its gradient vanishing issue but also achieving higher efficiency. Besides that, we also address the instability issue of the fairness loss by leveraging the competing nature between the recommendation loss and the fairness loss. Through extensive experiments on real-world datasets, we demonstrate that \ours\ effectively and efficiently reduces performance disparities with little sacrifice in the overall recommendation performance.\footnote{The code is available at: 
\href{https://github.com/hsyoo32/fade}{https://github.com/hsyoo32/fade}.}
% The code is available at:
% \href{https://doi.org/10.5281/zenodo.10669096}{https://doi.org/10.5281/zenodo.10669096}.
\end{abstract}




%
% Keywords. The author(s) should pick words that accurately describe
% the work being presented. Separate the keywords with commas.

\maketitle

\vspace{-2mm}
\section{Introduction}
\label{sec:introduction}

Recommender systems are essential for delivering high-quality personalized recommendations in a two-sided market (i.e., user-side and item-side)~\cite{wu2022graph, zou2019reinforcement,dong2021individual}. 
In this market, users provide feedback on recommended items, and the system refines the recommendations to better reflect their preferences. However, these recommender systems can perform poorly for users from certain demographic groups even while delivering high-quality recommendations on average~\cite{wang2023survey, chen2020bias}. For example, a job recommender system might recommend more irrelevant job opportunities to female engineers in STEM (Science, Technology, Engineering, and Mathematics), which can significantly impact their career growth~\cite{islam2021debiasing,lahoti2019ifair}. Thus, it is important to alleviate the performance disparity between different user groups in recommender systems~\cite{li2021user}.

Although there is a parallel line of research on item-side fairness, those methods do not apply to user-side fairness due to the fundamental distinction between user- and item-side fairness. In essence, user-side fairness is concerned with ensuring equitable recommendation quality for different users, while item-side fairness focuses on providing equal exposure opportunities for items within recommendations, often addressing the so-called popularity bias of items through debiasing techniques. For example, several works for item-side fairness~\cite{zhu2021popularity, steck2019collaborative, steck2011item, zhu2021popularityopportunity} calibrates predicted ratings with item popularity, which does not apply to user-side fairness.

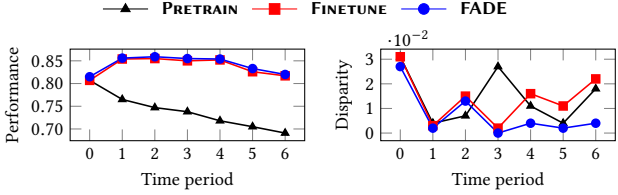
\begin{figure}[t]
\footnotesize
\centering
\begin{tikzpicture}
\begin{customlegend}[legend columns=3,legend style={align=left,draw=none,column sep=1ex},
        legend entries={ \textbf{\pretrain}, \textbf{\fine}, \textbf{\ours}
        % \textbf{New Data NOT Used}, \textbf{New Data Used w/o Fairness},\textbf{New Data Used w/ Fairness (\ours)} 
        }]
        %\addlegendimage{draw=violet,mark=triangle,mark size=1.5pt,line width=0.6pt}
        \addlegendimage{draw=black,color=black,mark=triangle*}   
        \addlegendimage{draw=red,color=red,mark=square*} 
        \addlegendimage{draw=blue,color=blue,mark=*} 
        \end{customlegend}
\end{tikzpicture}
\\
\begin{tikzpicture}
\begin{axis}[
height=2.8cm, width=4.7cm, %width=5cm,
xtick={1, 2, 3, 4, 5, 6,7,8,9,10},
xticklabels={0,1,2,3,4,5,6,7,8,9},
ylabel=Performance,%ylabel=NDCG@20,
xlabel=Time period,
y tick label style={/pgf/number format/.cd,fixed,fixed zerofill,precision=2,/tikz/.cd},]

\addplot[color=black,mark=triangle*,mark size=1.5pt,line width=0.6pt]
coordinates {(1,0.807) (2,0.765) (3,0.747) (4,0.738) (5,0.718) (6,0.705) (7,0.691) };
\addplot[color=red,mark=square*,mark size=1.5pt,line width=0.6pt]
coordinates {(1,0.807) (2,0.854) (3,0.855) (4,0.850) (5,0.852) (6,0.826) (7,0.817) };
\addplot[color=blue,mark=*,mark size=1.5pt,line width=0.6pt]
coordinates {(1,0.815) (2,0.856) (3,0.859) (4,0.855) (5,0.854) (6,0.833) (7,0.820) };

\end{axis}
\end{tikzpicture}
\hspace{2mm}
\begin{tikzpicture}
% \footnotesize
\begin{axis}[
height=2.8cm, width=4.7cm,%width=5cm,
xtick={1, 2, 3, 4, 5, 6,7,8,9,10},
xticklabels={0,1,2,3,4,5,6,7,8,9},
ylabel=Disparity,xlabel=Time period,
y tick label style={/pgf/number format/.cd,fixed,fixed zerofill,precision=0,/tikz/.cd},]
\addplot[color=black,mark=triangle*,mark size=1.5pt,line width=0.6pt]
coordinates {(1,0.031) (2,0.004) (3,0.007) (4,0.027) (5,0.011) (6,0.004) (7,0.018) };
\addplot[color=red,mark=square*,mark size=1.5pt,line width=0.6pt]
coordinates {(1,0.031) (2,0.003) (3,0.015) (4,0.002) (5,0.016) (6,0.011) (7,0.022) };
\addplot[color=blue,mark=*,mark size=1.5pt,line width=0.6pt]
coordinates {(1,0.027) (2,0.002) (3,0.013) (4,0.000) (5,0.004) (6,0.002) (7,0.004) };

\end{axis}
\end{tikzpicture}\\
{\scriptsize\textbf{(a)} Recommendation performance over time\quad\;\;\textbf{(b)} Performance disparity between user groups}
% \\
\caption{
% \RZ{The caption is not concise enough. How about this? \ch{``Even though incremental fine-tuning with new data improves recommendation performance, the performance disparity gradually expands over time if without fairness regularization. (See Appendix X for detail.)''} We should emphasize the fairness instead of treating recommendation performance and fairness equally.}
%The changes in recommendation performance (left) and performance disparity between two user groups (right) over time. 
Even though incremental fine-tuning with new data (red curve) upholds recommendation performance compared to pretrain (black curve),
%improves \hh{but the red curve is overall flat and even slightly decreases over time. shall we say sth like 'upholds/maintains recommendation performances compared with pretrain (black curve)'?}recommendation performance, 
the disparity gradually expands over time without fairness regularization. (See \S\ref{sec:experiment} for detail.)
% All models are based on Matrix Factorization as the base recommender and evaluated in Task-R on the Movielenz dataset. (See \S~\ref{{sec:experiment} for detailed settings). 
% (a) The performance gain of models using new data compared to the one NOT using new data tends to increase over time, indicating the importance of incremental adaptation to new data for improving the recommendation quality. 
% (b) The performance disparity remains relatively high over time when the model lacks fairness regularization, highlighting the necessity of addressing the issue of dynamic user-side fairness.
% See \S~\ref{{sec:experiment} for detailed experimental settings.
}\label{fig:observation}
% \vspace{-2mm}
\end{figure}

Furthermore, due to the evolving nature of user-item interactions, real-world recommender systems continually adapt to new data over time to improve recommendation quality~\cite{zhang2020retrain, kim2022meta}. However, as shown in Fig.~\ref{fig:observation}, neglecting fairness during dynamic adaptation leads to performance disparity between user groups persisting or even expanding over time. This highlights the importance of maintaining user-side fairness in dynamic recommendation.

Despite its critical importance, to the best of our knowledge, user-side fairness \cite{li2021user,fu2020fairness} has not been explored in the context of dynamic recommendation, which is in stark contrast to the extensive research effort on item-side fairness in dynamic recommendation \cite{zhu2021popularity, ge2021towards, morik2020controlling}. As item-side methods are inapplicable to user-side fairness, a thorough study of user-side fairness in dynamic recommendation will substantially expand the frontiers of fair dynamic recommendation and establish a prospective foundation for future research on two-sided fairness \cite{wu2021tfrom, do2021two} in dynamic recommendation.

This paper presents the first principled study of \textit{user-side fairness} in \textit{dynamic/continual} recommender systems. We identify and address the following challenges: 
\textbf{(C1)~Distribution shifts.} Constant emergence of new users/items and evolving user preferences lead to distribution shifts. Distribution shifts not only affects recommendation performance but also worsens performance disparity among user groups over time.
\textbf{(C2)~Frequent model updates.} Due to distribution shifts in dynamic recommendation, recommender systems need frequent updates to cater to current user preferences. This imposes \textit{efficiency} requirements on the model updating method. However, existing postprocessing methods involve time-intensive re-ranking \cite{fu2020fairness,li2021user}, which are inefficient for frequent model updates. 
\textbf{(C3)~Non-differentiability of ranking metrics.} 
The sorting operation in ranking metrics is non-differentiable. This raises a critical challenge in end-to-end training because we cannot directly use the non-differentiable performance disparity as the fairness loss. Even if one resorts to postprocessing methods like re-ranking \cite{fu2020fairness,li2021user} which does not involve end-to-end training, they critically suffer from the existing performance disparity in candidate item lists.

To address the challenges, we propose \ul{FA}ir \ul{D}ynamic r\ul{E}commender (\ours), an \textit{end-to-end} framework employing an incremental \textit{fine-tuning} strategy to dynamically alleviate performance disparity between user groups. 
Specifically, our key contributions are: 
\vspace{-4mm}%\vspace{-2mm}
\begin{itemize}[leftmargin=*]
    \item 
    \textbf{Problem.} We observe that the user-side performance disparity tends to persist or worsen over time, despite improvements in recommendation performance. 
    To our knowledge, we are the first to study user-side fairness in dynamic/continual recommendation. 
    \item\textbf{Theory.}
    To ground the design of our method, we theoretically analyze \textit{fine-tuning} v.s.\ \textit{retraining} in terms of generalization error (recommendation \& fairness) under distribution shifts. Our Theorems~\ref{thm:ft} \& \ref{thm:rt} show that the best practice is incremental fine-tuning with restart. 
    \item \textbf{Algorithm.}
    Based on theoretical analyses, we propose \ours, a novel dynamic recommender based on incremental fine-tuning that balances both recommendation quality and user-side fairness. To overcome the non-differentiability of recommendation metrics in the fairness loss, we further introduce \textit{Differentiable Hit} (DH) as an improvement over the recent NeuralNDCG method \cite{pobrotyn2021neuralndcg}, not only alleviating its gradient vanishing but also achieving higher efficiency. Besides that, we also address the instability of the fairness loss by leveraging the competing nature between the recommendation loss and the fairness loss (Proposition~\ref{prop:fair}).
    \item \textbf{Experiments.} 
    Empirical experiments on real-world datasets demonstrate that \ours\ effectively reduces performance disparity (average decrease of 48.91\%) without significantly compromising overall performance over time (average drop of 2.44\%). 
    %It also operates efficiently, achieving an 270x faster running time on average compared to the retraining approach.
\end{itemize}
% \vspace{-2mm}

\vspace{-2mm}
\section{Problem Definition}
\label{sec:preliminary}

In this section, we first present the key notations in the paper. Then we provide preliminaries on the settings of dynamic recommendation and user-side fairness. Finally, we formally define the problem of user-side fairness in dynamic recommender systems.

\noindent \textbf{Notations.} 
Table~\ref{tab:notations} provides a list of our symbols. 
% Throughout the paper, we use bold upper-case letters for matrices (e.g., $\mathbf{Y}$), bold lower-case letters for vectors (e.g., $\mathbf{r}$) and calligraphic letters for sets (e.g., $\U$). We use standard conventions %similar to NumPy~\cite{van2011numpy} in Python
% for indexing. For example, $\mathbf{Y}[i, j]$ is the entry at the $i$-th row and the $j$-th column in matrix $\mathbf{Y}$.
We use $\D_t = \{\U_t, \I_t, \E_t, \mathbf{Y}_{t}\}$ to denote the dataset collected at time period $t~\forall t\in\{1, \ldots, T\}$,
\footnote{%Depending on the needs of the system or implementation, 
Depending on the system, the time period could be either a specific time frame (e.g., daily or weekly) or until a specific number of interactions has been collected.}
where the subscript $_t$ indicates the time period $t$, $\U_t$ is the user set, $\I_t$ is the item set, $\E_t$ is the user-item interaction set, and $\mathbf{Y}_{t}$ is the user-item interaction matrix. 
We consider binary user-item interaction in this work, i.e., $\mathbf{Y}_t[u, i] = 1$ if user $u$ has interacted with item $i$ within the $t$-th time period, and 0 otherwise. 
The initial user set, item set, user-item interaction set, and the user-item interaction matrix \textit{before} the first time period (i.e., pretrain data) is denoted as $\U_0$, $\I_0$, $\E_0$, and $\mathbf{Y}_0$, respectively. 
Lastly, we denote the subscript $_{:t}$ as time period from the beginning up to $t$. 
% For example, $\U_{:t}$ denotes a set of items accumulated up to time period $t$ from the beginning (i.e., a set of entire users in the system).

\noindent \textbf{Dynamic/online recommendation.} We assume that an initial recommendation model has been pre-trained with $\D_0 = \{\U_0, \I_0, \E_0, \mathbf{Y}_0\}$ in an offline manner, and then the model is trained solely with the newly collected data
%only has access to the newly collected data 
$\D_t$ at the current time period $t,~\forall t\in\{1, \ldots, T\}$. Once the model has been trained/fine-tuned on $\D_t$, a top-$K$ recommendation list $[i_1, \dots, i_K]$ for each user $u$, ranked by the predicted scores $\mathbf{\widehat Y}_t[u, i], \forall i$, is generated.

\noindent \textbf{User-side fairness.} Given a binary sensitive attribute $a\in\{0, 1\}$ (e.g., gender), we focus on ensuring user-side group fairness, i.e., mitigate the recommendation performance disparity between the advantaged user group ($a = 0$) and the disadvantaged user group ($a = 1$)~\cite{li2021user}. More specifically, the user-side performance disparity at any time period $t$ is defined as follows.
\vspace{-2mm}
\begin{defn}[\textnormal{User-side performance disparity~\cite{li2021user}}] For a time period $t$ %, given a user set $\U_{:t}$ with sensitive attribute $a$, an item set $\I_{:t}$,
    with ground-truth test interaction set $\D^{\rm test}_t$ and for a recommendation metric $\operatorname{Perf(\cdot)}$ (such as $\RM{NDCG}@K$ or $\RM{F1@}K$), % and a recommendation model with parameters $\mathcal{W}_t$,
    the \textnormal{user-side performance disparity} is defined by
\vspace{-1mm}
\begin{equation}\label{eq:perrformance_disparity}
    \pd_t:=\operatorname{Perf}(\D^{\rm test}_t \mid a = 0)-\operatorname{Perf}(\D^{\rm test}_t \mid a = 1).
\end{equation}
\end{defn}

\vspace{-1mm}
\noindent\textbf{Problem definition.} 
% We formally define the problem of dynamic user-side fairness in recommender systems as follows.
We formally define our problem as follows:

\begin{prob}[\textnormal{User-side fairness in dynamic recommender systems}]
\textbf{Input:} (1) a pre-trained recommendation model with parameters $\mathcal{W}_0$; (2) a continually collected dataset $\D_t = \{\U_t, \I_t, \E_t, \mathbf{Y}_t\}, \forall t\in\{1, \ldots, T\}$; (3) a binary sensitive attribute $a\in\{0, 1\}$; (4) a specific performance evaluation metric $\text{Perf}(\cdot)$ to calculate $\pd_t$ (see Eq. \eqref{eq:perrformance_disparity}).

\textbf{Output:} For any time period $t$, a fairness-regularized model with the parameters $\mathcal{W}_t$ that (1) optimizes the $\pd_t$ to be close to zero and (2) achieves high-quality recommendations.
\end{prob}
\vspace{-2mm}
\section{\ours: A Fair Dynamic Recommender}
\label{sec:proposed_method}

In this section, we present \ours, a novel fair dynamic recommender system designed to effectively and efficiently reduce performance disparity over time. We begin with theoretical analyses on fine-tuning v.s.\ retraining in the context of dynamic fair recommendation in \S\ref{ssec:theory}, demonstrating that the best practice is incremental fine-tuning with restart. Then in \S\ref{ssec:overall_loss_fine-tune}, we introduce our incremental fine-tuning strategy that balances both recommendation performance and user-side fairness. To address the non-differentiability challenge, we improve NeuralNDCG \cite{pobrotyn2021neuralndcg} and develop \textit{Differentiable Hit} (DH), an efficient approximation scheme of the non-differentiable ranking metric, in \S\ref{ssec:drm}. Building upon DH, we propose a differentiable and lightweight loss function for user-side fairness in \S\ref{ssec:fairness_loss}. Our method is presented in Algorithm~\ref{alg:method}.

\vspace{-2mm}
\subsection{Fine-Tuning v.s.\ Retraining}\label{ssec:theory}
Common practice for evolving data involves incremental \textit{fine-tuning} and \textit{retraining}. To obtain a deeper understanding of their behaviors in fair dynamic recommendation to guide the design of our method, we theoretically analyze their generalization error (recommendation \& fairness) under distribution shift.
Suppose that the model is currently trained with $\CAL D_0\cup\cdots\cup\CAL D_{t_\textnormal{te}-1}$ and is to be tested on $\CAL D_{t_\textnormal{te}}$.
For each time period $t$, let $m_t:=|\CAL E_t|$ %\hs{more accurately, $|\CAL E_t|$?}
denote the size of dataset $\CAL D_t$, let $\CAL L^{\CAL D_t}(\CAL W)$ denote the empirical loss (recommendation + fairness (e.g., Eq.~\eqref{eq:final_loss})) over dataset $\CAL D_t$, let $\CAL L_t(\CAL W):=\Exp_{\CAL D_t}[\CAL L^{\CAL D_t}(\CAL W)]$ denote the true generalization loss, and let $\CAL L_t^*:=\inf_{\CAL W}\CAL L_t(\CAL W)$ denote the optimal loss value. To obtain concrete yet non-trivial theoretical results, 
we let $m_1=\cdots=m_{t_\textnormal{te}-1}\ll m_0$ and make mild and realistic assumptions for theoretical analysis (see \S\ref{app:assum}).

Next, we introduce our theoretical measure of distribution shift. There are two sources of distribution shift over time: \textit{covariate shift} 
and \textit{concept drift}. 
In dynamic recommendation, covariate shift corresponds to shift of user/item attribute distributions (i.e., the distribution of $(\U_t,\I_t,\E_t)$), and concept drift corresponds to evolution of user preferences (i.e., the conditional distribution $\Y_t|(\U_t,\I_t,\E_t)$).
\begin{table}[t]
\centering\caption{Main symbols used in this paper. 
%\jian{double check the bottomrule, it doesn't render correctly, or you can just use hline instead of all bottomrule, toprule and midrule}
}\resizebox{.49\textwidth}{!}{\renewcommand{\arraystretch}{1.15}\begin{tabular}{p{0.16\linewidth}|p{0.9\linewidth}} \toprule
    \hfil \textbf{Symbol} & \hfil \textbf{Description} \\ \midrule
    \hfil $\D_t$ & Dataset collected at time period $t$ \\
    \hfil $\U_t$, $\I_t$, $\E_t$ & Sets of users, items, and their interactions at time period $t$ \\ 
    \hfil $\mathbf{Y}_t$ & User-item interaction matrix at time period $t$ \\
    %/ 1 if user $u$ has interacted with item $i$ at time period $t$, and 0 otherwise \\  
    \hfil $\mathbf{\widehat Y}_t$ & User-item predicted score matrix at time period $t$ \\
    \hfil $\W_t$ & Set of model parameters at time period $t$\\
    \hfil $a$ & Binary sensitive attribute of a user\\ \midrule
    \hfil $\mathcal{L}_{\text{rec}}$, $\mathcal{L}_{\text{fair}}$ & Recommendation loss and fairness loss, respectively \\
    % \hfil $\mathcal{L}_{\text{fair}}(\D_t)$ & Fairness loss \\
    \hfil $\C_u$, $N$ & Set of candidate items for a user $u$ and its size \\
    % \hfil $\G_u$ & Set of ground-truth items in $\C_u$ \\
    % \hfil $N$ & Size of $\C_u$ \\
    \hfil $\mathbf{s}_u$ & Unsorted list of recommendation scores of items in $\C_u$ \\
    \hfil $\mathbf{r}_u$ & List of items in $\C_u$ ranked by their scores in $\mathbf{s}_u$ \\
    % \hfil $\V_u$ & Set of ground-truth items for a user $u$ \\
    \hfil $\P_{u}$, $\widehat{\P}_{u}$ & Permutation matrix and relaxed permutation matrix for $\mathbf{s}_u$ \\
    \hfil $\lambda$ & Scaling parameter for $\mathcal{L}_{\text{fair}}$ \\
    \hfil $\tau$ & Temperature parameter for $\widehat{\P}_{\mathbf{s}_u}$ \\
    \hfil $\mu$ & The number of negative items in $\C_u$ \\
    \hfil $n$ & The number of negative items for $\mathcal{L}_{\text{rec}}$\\
    \bottomrule
\end{tabular}
}
\label{tab:notations}
%   \vspace{-0.5cm}
\end{table}

% Demographic Changes in Customer Base: Consider an online streaming service that initially attracts mainly young adults and teenagers as its primary user base. Over time, the platform gains popularity among older demographics, such as middle-aged and elderly users. As a result, the distribution of user demographics (age groups) shifts, with a decrease in the proportion of younger users and an increase in the proportion of older users. This change in the distribution of user demographics represents covariate shift.
% Movie Recommendation System: Imagine a movie recommendation system that learns user preferences based on movie genres, actors, and directors. Initially, users may prefer action-packed movies with high-intensity sequences. However, over time, user preferences might change, and they may develop a preference for romantic comedies or documentaries. In this scenario, the evolution of user preferences over time represents concept drif

Regarding covariate shift, a classic measure is the \textit{discrepancy distance} \cite{mansour2009domain} (a generalized \textit{$\CAL H\Delta\CAL H$ distance} \cite{ben2010theory}):
\EQ{d^{\CAL H\Delta\CAL H}_{t,t_\textnormal{te}}\!\!:=\!\!\sup_{\CAL W,\CAL W'}\!\!\big||\CAL L_t(\CAL W)-\CAL L_t(\CAL W')|-|\CAL L_{t_\textnormal{te}}(\CAL W)-\CAL L_{t_\textnormal{te}}(\CAL W')|\big|.}
The intuition is that if there is no covariate shift between $t$ and $t_\text{te}$, then for any two models $\W,\W'$, their difference of $\CAL L$ should not differ between $t$ and $t_\text{te}$, leading to $d^{\CAL H\Delta\CAL H}_{t,t_\textnormal{te}}=0$.
Regarding concept drift, we use a classic measure called \textit{combined error} \cite{ben2010theory}:
\EQ{d^\text{comb}_{t,t_\textnormal{te}}\!\!:=\inf_{\CAL W}\!\big(\CAL L_t(\CAL W)+\CAL L_{t_\textnormal{te}}(\CAL W)\big)-\CAL L_t^*-\CAL L_{t_\textnormal{te}}^*.}
The intuition is that if there is no concept drift between $t$ and $t_\text{te}$, then $\CAL L_t$ and $\CAL L_{t_\text{te}}$ can achieve their minimum values with the same model $\CAL W$, leading to $d^\text{comb}_{t,t_\text{te}}=0$. 
Together, we define a unified measure of distribution shift as follows by combining the measures of covariate shift and concept drift:
\EQ{d_{t,t_\textnormal{te}}:=d^{\CAL H\Delta\CAL H}_{t,t_\textnormal{te}}+d^\text{comb}_{t,t_\textnormal{te}}.}

Building upon the measure of distribution shift, we theoretically analyze the generalization error (recommendation performance \& user-side fairness)
of fine-tuning and retraining in the presence of distribution shift (Theorems~\ref{thm:ft} \& \ref{thm:rt}).
%. Our results are shown in the following Theorems~\ref{thm:ft} \& \ref{thm:rt}.
\begin{theorem}[Fine-tuning]\label{thm:ft}
Let $\CAL L_{t_\textnormal{te}}^\textnormal{ft}$ denote the best possible loss of fine-tuning till $\CAL D_{t_\textnormal{te}-1}$. Suppose that the number of fine-tuning epochs at each time period $t\ge1$ is decided according to the proximity assumption \cite{rajeswaran2019meta} with some $0<\gamma<1$ (see \S\ref{app:assum} for detail). Then with probability at least $1-\delta$, %the fine-tuned model $\HAT h_\textnormal{ft}$ has
\AL{
\CAL L_{t_\textnormal{te}}^\textnormal{ft}\le{}&\CAL L_{t_\textnormal{te}}^*+\textcolor{blue}{2\gamma^{t_\textnormal{te}-1}d_{0,t_\textnormal{te}}+2\sum\limits_{t=1}^{t_\textnormal{te}-1}(1-\gamma)\gamma^{t_\textnormal{te}-t-1}d_{t,t_\textnormal{te}}}\\
&+4\sqrt{\Big(\textcolor{red}{\frac{\gamma^{2t_\textnormal{te}-2}}{\frac{m_0}{\log m_0}}+\frac{(1+\gamma)(1-\gamma^{2t_\textnormal{te}-4})}{(1-\gamma)\frac{m_1}{\log m_1}}}\Big)\log\frac2\delta}\nonumber
.}
%\EQ{\CAL L_{t_\textnormal{te}}^\textnormal{ft}\!\le\!\CAL L_{t_\textnormal{te}}^*+\tfrac{\!(1-\gamma)\!\Big(\textcolor{blue}{2\!\!\sum\limits_{t=0}^{{t_\textnormal{te}}-1}\!\gamma^{t_\textnormal{te}-t-1}\!d_{t,t_\textnormal{te}}}\,+\,4\!\sqrt{\!\textcolor{red}{\big(\!\frac{\gamma^{2t_\textnormal{te}-2}}{\frac{m_0}{\log m_0}}+\frac{1-\gamma^{2t_\textnormal{te}-2}}{(1-\gamma^2)\frac{m_1}{\log m_1}}\!\big)}\!\log\tfrac2\delta\!}\Big)\!}{1-\gamma^{t_\textnormal{te}}}\!.}
\end{theorem}

\begin{theorem}[Retraining]\label{thm:rt}
Let $\CAL L_{t_\textnormal{te}}^\textnormal{rt}$ be the best possible loss of retraining on $\CAL D_0\cup\cdots\cup\CAL D_{t_\textnormal{te}-1}$. With probability at least $1-\delta$, %the retrained model $\HAT h_\textnormal{rt}$ has
\AL{\CAL L_{t_\textnormal{te}}^\textnormal{rt}\le{}&\CAL L_{t_\textnormal{te}}^*+\textcolor{blue}{\frac{2m_0d_{0,t_\textnormal{te}}\,+2\!\!\sum\limits_{t=1}^{t_\textnormal{te}-1}\!m_1d_{t,t_\textnormal{te}}}{m_0+(t_\textnormal{te}-1)m_1}}\\&+4\sqrt{\textcolor{red}{\frac{\log m_0}{m_0+(t_\textnormal{te}-1)m_1}}\!\log\tfrac2\delta}\nonumber.}
\end{theorem}

Proofs are in \S\ref{app:pf-rt-ft}. Theorems~\ref{thm:ft} \& \ref{thm:rt} point out two sources of generalization error: (i) \textcolor{blue}{distribution shift} in terms of $d_{t,t_\textnormal{te}}$ and (ii) \textcolor{red}{learning error} due to the finite dataset size $m_t$. Regarding distribution shift, since larger time differences typically result in larger distribution shifts, we have $d_{0,t_\textnormal{te}}>d_{1,t_\textnormal{te}}>\cdots>d_{t_\textnormal{te}-1,t_\textnormal{te}}$ 
Fine-tuning can exponentially shrink 
(via the \textcolor{blue}{$\gamma^{t_\textnormal{te}-t-1}$} factor) 
the influence of distribution shift while retraining suffers from greater influence of distribution shift.
This is consistent with our intuition since retraining treats old and new data equally while fine-tuning pays more attention to newer data. This suggests that we should use fine-tuning to mitigate the impact of distribution shift. Meanwhile, when $t_\textnormal{te}$ is large, fine-tuning's learning error $\textcolor{red}{\frac{(1+\gamma)(1-\gamma^{2t_\textnormal{te}-4})}{(1-\gamma)\frac{m_1}{\log m_1}}}$ can be greater than retraining's \textcolor{red}{$\frac{\log m_0}{m_0+(t_\text{te}-1)m_1}$} because $m_1\ll m_0$. 
This suggests that the performance of dynamically fine-tuned model will eventually degrade after a number of periods, which is consistent with our empirical observation (refer to Fig.10 in \S\ref{assec:main_result}).

Therefore, to leverage the benefits of fine-tuning without sacrificing performance, it is advisable to fine-tune the model for certain periods $T$ 
until performance degradation is observed, then retrain 
the model from scratch and resume the fine-tuning process.

\vspace{-2mm}
\subsection{Incremental Fine-Tuning Strategy}\label{ssec:overall_loss_fine-tune}
Building upon our theoretical analysis on \textbf{(C1) distribution shifts} and for the sake of \textbf{(C2) time-efficiency}, \ours\ fine-tunes the model parameters incrementally over time only with the new data $\D_t$ collected at $t$. We optimize the following loss functions: 
\begin{align} \label{eq:final_loss}
% \small
    &\mathcal{L}^{\D_t} := \mathcal{L}_{\text{rec}}^{\D_t} + \lambda\mathcal{L}_{\text{fair}}^{\D_t},
\end{align}
where $\mathcal{L}_{\text{rec}}$ is for improving the recommendation performance, $\mathcal{L}_{\text{fair}}$ %(i.e., Eq.~\eqref{eq:fair_loss})
is for regularizing the performance disparity between the disadvantaged and advantaged groups, and $\lambda$ is the scaling parameter for controlling the trade-off between the recommendation performance and the fairness. 
In this paper, we use the classic 
Bayesian personalized ranking (BPR) loss~\cite{rendle2012bpr} as $\mathcal{L}_{\text{rec}}$, i.e.,
\begin{equation}
\mathcal{L}_{\text{rec}}^{\D_t}:=-\frac{1}{|\mathcal{E}_t|}\sum_{(u,i)\in \mathcal{E}_t} \frac{1}{|\CAL N_{ui}|} \sum_{i'\in\CAL N_{ui}} \log(\sigma(s_{ui}-s_{ui'})),
\end{equation}
where $\sigma(\cdot)$ is the sigmoid function, and $\CAL N_{ui}$ is a set of sampled negative items for $u$.
Note that this loss can be replaced with any differentiable recommendation loss that can be optimized by gradient descent. We will define $\CAL L_\text{fair}$ in \S\ref{ssec:fairness_loss}.

By jointly optimizing $\mathcal{L}_{\text{rec}}$ and $\mathcal{L}_{\text{fair}}$ in an end-to-end fashion to fine-tune the model parameters for each time period, we can dynamically and efficiently reduce the performance disparity, which may otherwise worsen as the optimization continues, while simultaneously accurately preserving the user preferences over time.

\vspace{-2mm}
\subsection{\textbf{Differentiable Hit}}\label{ssec:drm}
Most evaluation metrics for top-$K$ recommendations, such as NDCG@$K$, are not differentiable due to their reliance on the ranking/sorting operation of items. As discussed in \S\ref{sec:introduction}, this \textbf{(C3) non-differentiability} presents a challenge when optimizing fairness measures, specifically performance disparity, using gradient descent algorithms. To overcome this challenge, several soft ranking losses have been proposed to directly optimize relaxed ranking metrics~\cite{pobrotyn2021neuralndcg, qin2010general, cao2007learning}. NeuralNDCG \cite{pobrotyn2021neuralndcg} is a recent work on differentiable approximation of ranking metrics. However, due to the use of the Sinkhorn algorithm, NeuralNDCG may lead to the gradient vanishing issue and also poses \textbf{(C2) time-inefficiency}. To address these limitations, we propose {\it Differentiable Hit}, a function that is not only effective but also more lightweight than existing methods, making it well-suited for dynamic recommendation.

First, let us define a standard {\it Hit} function.
Suppose a score vector $\mathbf{s}_u=[s_{u1}, s_{u2}, \dots, s_{uN}]^\T$ for a user $u$ represents the \textit{``unsorted''} list of recommendation scores (i.e., $s_{ui}=\mathbf{\hat Y}_t[u, i]$) of $N$ \textit{candidate items} in a set $\C_u$ (with $|\C_u|=N$), a vector $\textbf{r}_u$ represents the \textit{``sorted''} list of items ranked in the descending order by their scores in $\mathbf{s}_u$,
and $\mathbf{r}_u[k]$ represents the $k$-th ranked item.

With the above definitions, we can define the Hit function, $\operatorname{Hit}(\C_u;k)$ for $k\in \{1,\ldots,K\}$, which indicates whether the $k$-th ranked item $\mathbf{r}_u[k]$ is $u$'s ground-truth item, as follows:
\begin{equation} \label{eq:hit}
    \operatorname{Hit}(\C_
    u;k) := \begin{cases} 1 & \text{if } \mathbf{Y}_t[u, \mathbf{r}_u[k]] = 1, \\
    0 & \text{if } \mathbf{Y}_t[u, \mathbf{r}_u[k]] = 0. \end{cases}
\end{equation}
Here, the \textit{sorting operation} used to produce the $\mathbf{r}_u$, which can also be represented as \textit{a permutation matrix}, renders the $\operatorname{Hit}_u(k)$ non-differentiable. 
However, we can overcome this limitation by using the continuous relaxation for permutation matrices to approximate the deterministic sorting operation with a differentiable continuous sorting~\cite{groverstochastic}.
First, for the deterministic sorting, 
the permutation matrix $\mathbf{P}_{u}\in\mathbb{R}^{N\times N}$ is given by~\cite{groverstochastic}: 
\begin{equation} \label{eq:permutation}
    \P_{u}[k,j] := \begin{cases} 1 & \text{if}~ j = \operatorname{argmax} [(N + 1 - 2k)\textbf{s}_u - \textbf{A}_{u} \mathbf{1}],\\0 & \text{otherwise}, \end{cases}
\end{equation}
where $\mathbf{1}$ is the column vector of all ones and $\mathbf{A}_{u}$
is the absolute distance matrix of $\mathbf{s}_u$ with $\mathbf{A}_{u}[k,j] = |s_{uk}-s_{uj}|$. 
For instance, if we set $k=\floor*{(N+1)/2}$, then the non-zero entry in the $k$-th row, $\P_u[k,:]$, corresponds to the element with the minimum sum of absolute distances to the other elements, and this corresponds to the median element, as desired. 

Then, the argmax operator is replaced by Gumbel-softmax~\cite{jang2016categorical} to obtain a continuous relaxation of the permutation matrix; 
the $k$-th row of the permutation matrix is relaxed as follows~\cite{groverstochastic}:
\begin{equation} \label{eq:relaxed_perm}
% \small
    \widehat{\P}_{u}[k,:] := \operatorname{softmax} \left[ \left((N+1-2k)\textbf{s}_u - \textbf{A}_{u} \mathbf{1}\right)/\tau \right],
\end{equation} 
where $\tau$ is the temperature parameter, and $\widehat{\P}_{u}$ approaches a permutation matrix (i.e., Eq.~\eqref{eq:permutation}) when $\tau\to0^+$.  
Intuitively,  
each entry of $\widehat{\P}_{u}[k,:]$ indicates the probability that the corresponding item will be the $k$-th ranked item. Since $\widehat{\P}_{u}$ is continuous everywhere and differentiable almost everywhere w.r.t. the elements of $\mathbf{s}_u$, we can define a differentiable Hit, as we elaborate below.

Since the $k$-th row of the %deterministic
permutation matrix \mbox{$\P_{u}[k,:]$} (i.e., Eq.~\eqref{eq:permutation}) is equal to the one-hot vector of the $k$-th ranked item,
we can reformulate the Hit function (i.e., Eq.~\eqref{eq:hit}) as follows:
\begin{equation}
    \operatorname{Hit}(\C_u;k) =  \P_{u}[k,:] \cdot \Y_t[u,:]^\T,
\end{equation}
where $\Y_t[u,i]=1$ if the item $i$ is a ground-truth item, and 0 otherwise.
Finally, by replacing $\P_{u}[k,:]$ (Eq.~\eqref{eq:permutation}) with $\widehat{\P}_{u}[k,:]$ (Eq.~\eqref{eq:relaxed_perm}), we %formally
define a \textit{Differentiable Hit} (DH) as follows:
\begin{equation} \label{eq:re_hit}
% \small
    \operatorname{DH}(\C_u;k):=%(\W_t, \mathcal{C}_u; k) = %\operatorname{\widehat{Hit}}_u(k) = 
    \widehat{\P}_{u}[k,:] \cdot \Y_t[u,:]^\T .
\end{equation}
Using DH as a building block, we can differentiably approximate various top-$K$ recommendation metrics. For example,
\begin{equation} \label{eq:drm}
\small
\text{NDCG@}K\approx\frac1{|\U_t|}\sum_{u\in\U_t}\frac1{\OP{maxDCG}(\C_u)}\sum_{k=1}^K \frac{\dh(\C_u;k)}{\text{log}_2(k+1)},
\end{equation}
where $\OP{maxDCG}(\C_u)$ is the maximum possible value of $\sum_{k=1}^K \frac{\dh(\C_u;k)}{\text{log}_2(k+1)}$, computed by decreasingly ordering $i \in \C_u$ by $\Y_t[u,i]$.

\begin{algorithm}[t]
    \caption{Fine-tuning procedure at time period $t$}
    \label{alg:method}
    \begin{algorithmic}[1]
    \State \textbf{Input:} Model parameters $\W_{t-1}$, scaling parameter $\lambda$, temperature parameter $\tau$, the number of negative items $n$ for $\mathcal{L}_{\text{rec}}$ and $\mu$ for $\mathcal{L}_{\text{fair}}$, sensitive attribute $a$, incoming dataset $\D_t = \{\U_t, \I_t, \E_t, \mathbf{Y}_t\}$
    \State \textbf{Output:} Updated model parameters $\W_{t}$
    \State $\W_t \gets \W_{t-1}$;
    \For{epoch}
    \For{mini-batch $\mathcal{B} \text{ obtained from } \E_t$}
        \For{user-item interaction $(u, i) \in \mathcal{B}$}
            \State Sample $n$ negative items as $\CAL N_{ui}$;
            \State Sample $\mu$ negative items as $\CAL N'_{ui}$; $\C_{ui}\gets\{i\}\cup\CAL N'_{ui}$;
        \EndFor
        \State $\mathcal{L}_{\text{rec}}\!\gets{-\frac{1}{|\mathcal{B}|}\sum_{(u,i)\in \mathcal{B}} \frac{1}{|\CAL N_{ui}|} \sum_{i'\in\CAL N_{ui}}\!\! \log(\sigma(s_{ui}-s_{ui'}))}$;
        \State$\diffpd \gets
        \frac{\sum_{(u,i) \in \{\CAL B|a=0\}}\! \dh(\C_{ui};1)}{|\{\CAL B|a=0\}|} - \frac{\sum_{(u,i) \in \{\CAL B|a=1\}}\! \dh(\C_{ui};1)}{|\{\CAL B|a=1\}|}$;
        \State$\mathcal{L}_{\textnormal{fair}}\gets -\log(\sigma({-\diffpd}))$;
        \State Update $\W_t$ based on $\mathcal{L}_{\text{rec}}+\lambda\mathcal{L}_{\text{fair}}$ via gradient descent;
    \EndFor
    \EndFor
    % \State \Return $\W_t$;
    \end{algorithmic}
\end{algorithm}

\vspace{-2mm}
\subsection{\textbf{Fairness Loss}}\label{ssec:fairness_loss}
We design our fairness loss for reducing performance disparity between the advantaged ($a=0$) and disadvantaged ($a=1$) user groups. 
For the sake of training efficiency, we compose each candidate set with only 1 positive item and several negative items and use differentiable Hit@1 in our fairness loss. Formally, for each $(u,i)\in\E_t$, we sample $\mu$ negative items $\CAL N'_{ui}$, compose a candidate set $\C_{ui}:=\{i\}\cup\CAL N'_{ui}$ and use $\dh(\mathcal{C}_{ui}; 1)$ %(see Eq.~\eqref{eq:re_hit}),
as a surrogate of the measure of recommendation quality for a user. 
While this differentiable Hit@1 used for training encourages the top-1 recommendation, it could also potentially benefit %more relaxed
Hit@$K$-based metrics. %in the evaluation stage.
We will empirically demonstrate that these settings consistently yield effective results across various recommendation metrics that rely on the Hit function. 
Based on $\dh$, 
we define the differentiable performance disparity ($\diffpd$) as follows:
\begin{equation} \label{eq:diff_pd}
% \small
    %\begin{aligned}
        \diffpd^{\D_t} :=\,
        \frac{\sum\limits_{(u,i) \in \{\E_t|a=0\}}\!\!\!\!\!\!\!\!\!\! \dh(\C_{ui};1)}{|\{\E_t|a=0\}|} - \frac{\sum\limits_{(u,i) \in \{\E_t|a=1\}}\!\!\!\!\!\!\!\!\!\! \dh(\C_{ui};1)}{|\{\E_t|a=1\}|} ,
    %\end{aligned}
\end{equation}
which is an approximation of $\pd_t$ in Eq.~\eqref{eq:perrformance_disparity} on the sampled item set. 
Then, a na\"ive fairness loss function is to minimize $|{\diffpd}|$:
\begin{equation}\label{eq:fair_loss_abs}
    \mathcal{L}_{\text{fair-abs}}^{\D_t} := -\log(\sigma({-|{\diffpd^{\D_t}}|})),
\end{equation}
where $\sigma(\cdot)$ is the sigmoid function. However, the non-smoothness of $\mathcal{L}_{\text{fair-abs}}$ will cause instability in training, as shown in our experiment (%refer to
Fig.~\ref{fig:trunc-tradeoff-taskR}). To address this limitation, we leverage the property of the sigmoid function and surprisingly prove that removing the absolute value operation $|\cdot|$ can still ensure fairness adaptively. Formally, we propose the following fairness loss:
\begin{equation} \label{eq:fair_loss}
% \small
    \mathcal{L}_{\textnormal{fair}}^{\D_t} := -\log(\sigma({-\diffpd^{\D_t}})).
\end{equation}
Then, we have Proposition~\ref{prop:fair}.

\begin{proposition}\label{prop:fair}
Let $\TLD{\W}_t:=\W_t-\eta\nabla_{\W_t}(\CAL L^{\D_t}_\textnormal{rec}+\lambda\CAL L^{\D_t}_\textnormal{fair})$ denote a gradient descent step with learning rate $\eta>0$. Suppose that $\CAL L_\textnormal{rec}$ causes unfairness (i.e., $\langle\nabla_{\W_t}\CAL L^{\D_t}_\textnormal{rec},\nabla_{\W_t}\diffpd^{\D_t}\rangle\le0$), and that the fairness loss has influence (i.e., $\nabla_{\W_t}\CAL L^{\D_t}_\textnormal{fair}\ne\BM0$). Then, there exists $\lambda\ge0$ such that
\EQ{\OP{sgn}(\diffpd^{\D_t}\!(\W_t))\cdot\lim_{\eta\to+0}\frac{\diffpd^{\D_t}\!(\TLD{\W}_t)-\diffpd^{\D_t}\!(\W_t)}{\eta}\le0.}
In particular, if $\,\diffpd^{\D_t}\!(\W_t)\!\le\!0$, then $\diffpd^{\D_t}\!(\TLD{\W}_t)\!\ge\!\diffpd^{\D_t}\!(\W_t)$ as $\eta\to+0$.
\end{proposition}
Proof is in \S\ref{app:prop:fair}.
Intuitively, our $\mathcal{L}_{\textnormal{fair}}$ aims to benefit the disadvantaged user group ($a\!=\!1$) over the advantaged group ($a\!=\!0$).
Meanwhile, whenever $\diffpd\!<\!0$, the influence of $\mathcal{L}_{\textnormal{fair}}$ will be reduced adaptively, so the unfair $\mathcal{L}_{\textnormal{rec}}$ will push $\diffpd$ back to zero.

\vspace{-3mm}
\subsection{\textbf{{Complexity Analysis}}}\label{ssec:complexity}
Our fairness loss only adds a constant amount of complexity to most existing recommendation models. 
Assuming we employ MF-BPR~\cite{rendle2012bpr} as the base recommendation model with  user/item embeddings of dimensionality $d$, the time complexity of minimizing $\mathcal{L}_{\text{rec}}^{\D_t}$ is $\mathcal{O}(|\E_t|nd)$, where $n$ represents the number of negative items.

Regarding our fairness loss, \textit{for each user interaction}, computing the score vector $\mathbf{s}_u$ has a time complexity of $\mathcal{O}(\mu d)$, and computing DH incurs $\mathcal{O}(\mu^2)$ time complexity due to computing $\widehat{\P}_{u}[k,:]$ (i.e., Eq.~\eqref{eq:relaxed_perm}), which involves computing $\textbf{A}_{u}\in\mathbb{R}^{(\mu+1)\times (\mu+1)}$.
As a result, the time complexity of minimizing $\mathcal{L}_{\text{fair}}^{\D_t}$ becomes $\mathcal{O}(|\E_t|(\mu^2+\mu d))$, which can be approximated as $\mathcal{O}(|\E_t|\mu d)$ since $\mu\ll d$.
Therefore, the time complexity of minimizing the recommendation loss, $\mathcal{L}_{\text{rec}}^{\D_t}$, and the fairness loss, $\mathcal{L}_{\text{fair}}^{\D_t}$, are comparable.

% \vspace{-2mm}
\section{Experiments} 
\label{sec:experiment}
We design experiments to answer the key research questions (RQs):
%\footnote{Note that throughout the subsections for all RQs, we use PD to refer to absolute performance disparity $|{\pd}|$. }: 

\vspace{0mm}
\begin{itemize}
    \item [\textbf{RQ1.}] How does learning new data affect model overall behavior?
    \item [\textbf{RQ2.}] How effective is the fairness loss and fine-tuning in \ours?
    \item [\textbf{RQ3.}] Does \ours\ outperform its fairness-aware competitors?
    \item [\textbf{RQ4.}] How time-efficient is \ours?
    \item [\textbf{RQ5.}] How effective/efficient is the Differentiable Hit in \ours?
    \item [\textbf{RQ6.}] How sensitive is \ours\ to its hyperparameters?
\end{itemize}
% \vspace{-1mm}
% \textit{Note that additional results are in \S B.1-B.5 of the Online Appendix.}

\vspace{-3mm}
\subsection{Experimental Settings}
\subsubsection{Dataset.}
For experiments, we use two real-world recommendation datasets from different domains.
\begin{itemize}[leftmargin=*]
    \item  Movielens\footnote{\href{https://grouplens.org/datasets/movielens/1m/}{https://grouplens.org/datasets/movielens/1m/}}: This dataset contains $836,478$ ratings on $3,628$ movies by $6,039$ users at different timestamps. The sensitive attribute $a$ is determined by the gender of each user, with male users as $a=0$ (adv.) and female users as $a=1$ (disadv.).
    This classification is based on the observation that the dataset is male-dominated, consisting of $4,330$ male users with $627,933$ training instances and $1,709$ female users with $208,545$ training instances~\cite{li2022fairlp}.
    \item ModCloth\footnote{\href{https://github.com/MengtingWan/marketBias}{https://github.com/MengtingWan/marketBias}}~\cite{wan2020addressing}: This e-commerce dataset contains $83,147$ ratings on $1,014$ items (i.e., women's clothing) by $37,142$ users at different timestamps. The sensitive attribute $a$ is determined by the body shape of each user, with "Small" users as $a=0$ (adv.) and "Large" users as $a=1$ (disadv.). The dataset is dominated by "Small" users, comprising $28,374$ "Small" users with $66,663$ training instances and $8,768$ "Large" users with $16,484$ training instances.
\end{itemize}
Following previous works in recommender systems~\cite{zheng2021disentangling,kansal2022flib}, we binarize the 5-star ratings for both datasets. We set $\mathbf{Y}[u, i] = 1$ if user $u$ gives item $i$ a rating greater than 2, and $\mathbf{Y}[u, i] = 0$ otherwise.
% Note that the dataset descriptions provided earlier are based on these pre-processed datasets.

To simulate dynamic settings defined in \S\ref{sec:preliminary},  
we first sort the interactions in the dataset %$\D$
in %the
chronological order and use 60\%/70\% of them as pre-training data, and 28\%/21\% as dynamically observed data for Movielenz/ModCloth. 
We then split the dynamically observed data into 7 periods, each containing an equal number of interactions. This process yields $\{\D_0, \D_1 , \dots, \D_{T}\}$, where $T=7$.

% \hs{mention how we only use the first 7 data here?}
% \hs{Let's discuss how to mention we only use the first 7 data. here, or when we report our results}\RZ{We can just mention the timestamp range used in our dataset instead of saying that we pick 7 out of 10. Btw, we should also show the full results somewhere in the appendix}\hs{hmm nedd to discuss something about this}

\vspace{1mm}
\subsubsection{Compared methods.}
% \noindent\textbf{B -- Compared methods.} 
% To ensure that the effectiveness of \ours~ is independent of the base recommender system used, 
We use two base system including Matrix Factorization (MF) and Neural Collaborative Filtering (NCF), both with the Bayesian Personalized Ranking (BPR) loss~\cite{rendle2012bpr}.
In this setup, we aim to validate the effectiveness of our \textit{fine-tuning strategy} and the \textit{fairness loss} used in \ours\ in ensuring high recommendation quality and user-side fairness over time. To establish a benchmark, we compare \ours\ with the following six combinations:
\begin{itemize}[leftmargin=*]
    \item \pretrain/\pretrainfair: The static model pre-trained on $\D_0$ w/o and w/ the fairness loss, respectively.
    \item \full/\fullfair: Fully retraining the model using the accumulated historical data $\D_{:t}$ at each time period $t$, w/o and w/ the fairness loss, respectively. 
    \item \fine/\oursabs: Fine-tuning the model based on the current %data
    $\D_t$ at each %time period
    $t$, w/o the fairness loss and w/ the (na\"ive) fairness loss $\mathcal{L}_{\text{fair-abs}}$ that uses $|{\diffpd}|$ (Eq.~\eqref{eq:fair_loss_abs}), respectively.

\end{itemize}

In addition, we also compare \ours\ with the other fairness-aware competitors. To ensure a fair comparison, we implemented these methods with a fine-tuning strategy, even though they were originally not based on fine-tuning. The competitors we consider are:
\begin{itemize}[leftmargin=*]
    \item \adver~\cite{li2021towards}: This method is based on adversarial learning technique. It is originally designed to filter out sensitive attributes from user embeddings, but its primary focus is not on reducing the performance disparity among different user groups.
    \item Re-rank~\cite{li2021user}: This method is a fairness-constrained re-ranking approach. At each time period, a fine-tuned base model generates recommendation lists, which are used as the basis for generating new fair recommendation lists using this method. 
\end{itemize}

\input{figs/trunc-tradeoff-taskR}
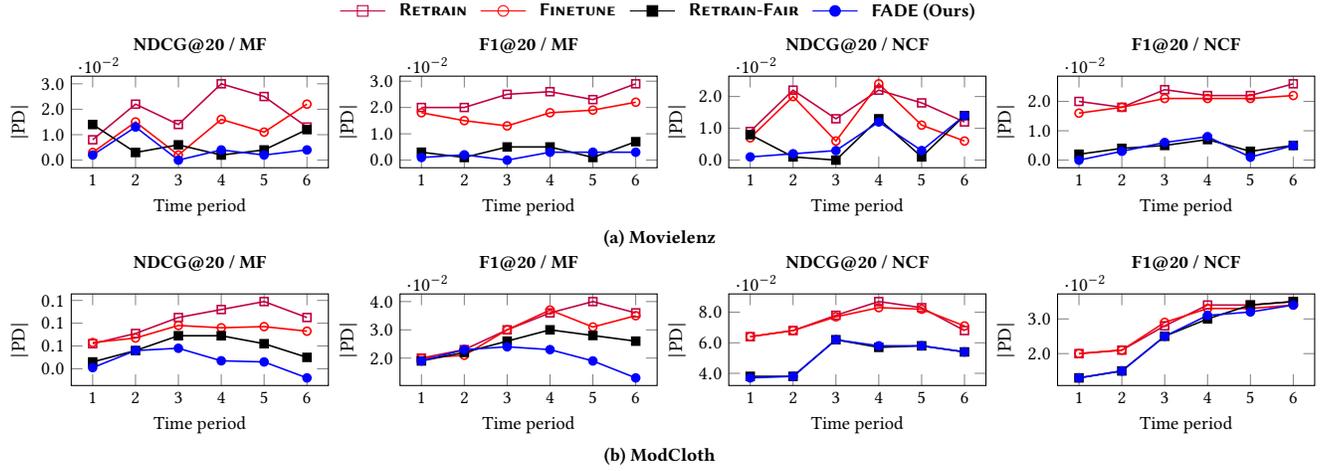
\begin{figure*}[t]
\footnotesize
\centering
\begin{tikzpicture}
\begin{customlegend}[legend columns=5,legend style={align=left,draw=none,column sep=1ex},
        legend entries={\textbf{\full}, \textbf{\fine}, \textbf{\fullfair}, \textbf{\ours\ (Ours)}}]

        \addlegendimage{draw=purple,mark=square}
        \addlegendimage{draw=red,mark=o}
        \addlegendimage{draw=black,mark=square*}   
        \addlegendimage{draw=blue,color=blue,mark=*} 
        \end{customlegend}
\end{tikzpicture}\\
    % \textbf{(a) Task-R}\vspace{-1mm}
\begin{tikzpicture}
\begin{axis}[
height=2.8cm, width=5cm,
title={\textbf{NDCG@20 / MF}},
xtick={1,2,3,4,5,6,7,8,9,10},
xticklabels={1,2,3,4,5,6,7,8,9,10},
ylabel=$|{\pd}|$,xlabel=Time period,
y tick label style={/pgf/number format/.cd,fixed,fixed zerofill,precision=1,/tikz/.cd},]

\addplot[color=purple,mark=square,mark size=1.5pt,line width=0.6pt]
coordinates {(1,0.008) (2,0.022) (3,0.014) (4,0.030) (5,0.025) (6,0.013) };
\addplot[color=red,mark=o,mark size=1.5pt,line width=0.6pt]
coordinates {(1,0.003) (2,0.015) (3,0.002) (4,0.016) (5,0.011) (6,0.022) };
\addplot[color=black,mark=square*,mark size=1.5pt,line width=0.6pt]
coordinates {(1,0.014) (2,0.003) (3,0.006) (4,0.002) (5,0.004) (6,0.012) };
\addplot[color=blue,mark=*,mark size=1.5pt,line width=0.6pt]
coordinates {(1,0.002) (2,0.013) (3,0.000) (4,0.004) (5,0.002) (6,0.004) };

\end{axis}
\end{tikzpicture}
\begin{tikzpicture}
\begin{axis}[
height=2.8cm, width=5cm,
title={\textbf{F1@20 / MF}},
xtick={1,2,3,4,5,6,7,8,9,10},
xticklabels={1,2,3,4,5,6,7,8,9,10},
ylabel=$|{\pd}|$,xlabel=Time period,
y tick label style={/pgf/number format/.cd,fixed,fixed zerofill,precision=1,/tikz/.cd},]
\addplot[color=purple,mark=square,mark size=1.5pt,line width=0.6pt]
coordinates {(1,0.020) (2,0.020) (3,0.025) (4,0.026) (5,0.023) (6,0.029) };
\addplot[color=red,mark=o,mark size=1.5pt,line width=0.6pt]
coordinates {(1,0.018) (2,0.015) (3,0.013) (4,0.018) (5,0.019) (6,0.022) };
\addplot[color=black,mark=square*,mark size=1.5pt,line width=0.6pt]
coordinates {(1,0.003) (2,0.001) (3,0.005) (4,0.005) (5,0.001) (6,0.007) };
\addplot[color=blue,mark=*,mark size=1.5pt,line width=0.6pt]
coordinates {(1,0.001) (2,0.002) (3,0.000) (4,0.003) (5,0.003) (6,0.003) };

\end{axis}
\end{tikzpicture}
\begin{tikzpicture}
\begin{axis}[
height=2.8cm, width=5cm,
title={\textbf{NDCG@20 / NCF}},
xtick={1,2,3,4,5,6,7,8,9,10},
xticklabels={1,2,3,4,5,6,7,8,9,10},
ylabel=$|{\pd}|$,xlabel=Time period,
y tick label style={/pgf/number format/.cd,fixed,fixed zerofill,precision=1,/tikz/.cd},]

\addplot[color=purple,mark=square,mark size=1.5pt,line width=0.6pt]
coordinates {(1,0.009) (2,0.022) (3,0.013) (4,0.022) (5,0.018) (6,0.012) };
\addplot[color=red,mark=o,mark size=1.5pt,line width=0.6pt]
coordinates {(1,0.007) (2,0.020) (3,0.006) (4,0.024) (5,0.011) (6,0.006) };
\addplot[color=black,mark=square*,mark size=1.5pt,line width=0.6pt]
coordinates {(1,0.008) (2,0.001) (3,0.000) (4,0.013) (5,0.001) (6,0.014) };
\addplot[color=blue,mark=*,mark size=1.5pt,line width=0.6pt]
coordinates {(1,0.001) (2,0.002) (3,0.003) (4,0.012) (5,0.003) (6,0.014) };

\end{axis}
\end{tikzpicture}
\begin{tikzpicture}
\begin{axis}[
height=2.8cm, width=5cm,
title={\textbf{F1@20 / NCF}},
xtick={1,2,3,4,5,6,7,8,9,10},
xticklabels={1,2,3,4,5,6,7,8,9,10},
ylabel=$|{\pd}|$,xlabel=Time period,
y tick label style={/pgf/number format/.cd,fixed,fixed zerofill,precision=1,/tikz/.cd},]

\addplot[color=purple,mark=square,mark size=1.5pt,line width=0.6pt]
coordinates {(1,0.020) (2,0.018) (3,0.024) (4,0.022) (5,0.022) (6,0.026) };
\addplot[color=red,mark=o,mark size=1.5pt,line width=0.6pt]
coordinates {(1,0.016) (2,0.018) (3,0.021) (4,0.021) (5,0.021) (6,0.022) };
\addplot[color=black,mark=square*,mark size=1.5pt,line width=0.6pt]
coordinates {(1,0.002) (2,0.004) (3,0.005) (4,0.007) (5,0.003) (6,0.005) };
\addplot[color=blue,mark=*,mark size=1.5pt,line width=0.6pt]
coordinates {(1,0.000) (2,0.003) (3,0.006) (4,0.008) (5,0.001) (6,0.005) };

\end{axis}
\end{tikzpicture}
\\
\textbf{(a) Movielenz}\\
\begin{tikzpicture}
\begin{axis}[
height=2.8cm, width=5cm,
title={\textbf{NDCG@20 / MF}},
xtick={1,2,3,4,5,6,7,8,9,10},
xticklabels={1,2,3,4,5,6,7,8,9,10},
ylabel=$|{\pd}|$,xlabel=Time period,
y tick label style={/pgf/number format/.cd,fixed,fixed zerofill,precision=1,/tikz/.cd},]

\addplot[color=purple,mark=square,mark size=1.5pt,line width=0.6pt]
coordinates {(1,0.062) (2,0.071) (3,0.085) (4,0.092) (5,0.099) (6,0.085) };
\addplot[color=red,mark=o,mark size=1.5pt,line width=0.6pt]
coordinates {(1,0.063) (2,0.067) (3,0.078) (4,0.076) (5,0.077) (6,0.073) };
\addplot[color=black,mark=square*,mark size=1.5pt,line width=0.6pt]
coordinates {(1,0.046) (2,0.056) (3,0.069) (4,0.069) (5,0.062) (6,0.050) };
\addplot[color=blue,mark=*,mark size=1.5pt,line width=0.6pt]
coordinates {(1,0.041) (2,0.056) (3,0.058) (4,0.047) (5,0.046) (6,0.032) };

\end{axis}
\end{tikzpicture}
\begin{tikzpicture}
\begin{axis}[
height=2.8cm, width=5cm,
title={\textbf{F1@20 / MF}},
xtick={1,2,3,4,5,6,7,8,9,10},
xticklabels={1,2,3,4,5,6,7,8,9,10},
ylabel=$|{\pd}|$,xlabel=Time period,
y tick label style={/pgf/number format/.cd,fixed,fixed zerofill,precision=1,/tikz/.cd},]
\addplot[color=purple,mark=square,mark size=1.5pt,line width=0.6pt]
coordinates {(1,0.020) (2,0.023) (3,0.030) (4,0.036) (5,0.040) (6,0.036) };
\addplot[color=red,mark=o,mark size=1.5pt,line width=0.6pt]
coordinates {(1,0.020) (2,0.021) (3,0.030) (4,0.037) (5,0.031) (6,0.035) };
\addplot[color=black,mark=square*,mark size=1.5pt,line width=0.6pt]
coordinates {(1,0.019) (2,0.022) (3,0.026) (4,0.030) (5,0.028) (6,0.026) };
\addplot[color=blue,mark=*,mark size=1.5pt,line width=0.6pt]
coordinates {(1,0.019) (2,0.023) (3,0.024) (4,0.023) (5,0.019) (6,0.013) };

\end{axis}
\end{tikzpicture}
\begin{tikzpicture}
\begin{axis}[
height=2.8cm, width=5cm,
title={\textbf{NDCG@20 / NCF}},
xtick={1,2,3,4,5,6,7,8,9,10},
xticklabels={1,2,3,4,5,6,7,8,9,10},
ylabel=$|{\pd}|$,xlabel=Time period,
y tick label style={/pgf/number format/.cd,fixed,fixed zerofill,precision=1,/tikz/.cd},]

\addplot[color=purple,mark=square,mark size=1.5pt,line width=0.6pt]
coordinates {(1,0.064) (2,0.068) (3,0.078) (4,0.087) (5,0.083) (6,0.068) };
\addplot[color=red,mark=o,mark size=1.5pt,line width=0.6pt]
coordinates {(1,0.064) (2,0.068) (3,0.077) (4,0.083) (5,0.082) (6,0.071) };
\addplot[color=black,mark=square*,mark size=1.5pt,line width=0.6pt]
coordinates {(1,0.038) (2,0.038) (3,0.062) (4,0.057) (5,0.058) (6,0.054) };
\addplot[color=blue,mark=*,mark size=1.5pt,line width=0.6pt]
coordinates {(1,0.037) (2,0.038) (3,0.062) (4,0.058) (5,0.058) (6,0.054) };

\end{axis}
\end{tikzpicture}
\begin{tikzpicture}
\begin{axis}[
height=2.8cm, width=5cm,
title={\textbf{F1@20 / NCF}},
xtick={1,2,3,4,5,6,7,8,9,10},
xticklabels={1,2,3,4,5,6,7,8,9,10},
ylabel=$|{\pd}|$,xlabel=Time period,
y tick label style={/pgf/number format/.cd,fixed,fixed zerofill,precision=1,/tikz/.cd},]

\addplot[color=purple,mark=square,mark size=1.5pt,line width=0.6pt]
coordinates {(1,0.020) (2,0.021) (3,0.028) (4,0.034) (5,0.034) (6,0.035) };
\addplot[color=red,mark=o,mark size=1.5pt,line width=0.6pt]
coordinates {(1,0.020) (2,0.021) (3,0.029) (4,0.033) (5,0.033) (6,0.034) };
\addplot[color=black,mark=square*,mark size=1.5pt,line width=0.6pt]
coordinates {(1,0.013) (2,0.015) (3,0.025) (4,0.030) (5,0.034) (6,0.035) };
\addplot[color=blue,mark=*,mark size=1.5pt,line width=0.6pt]
coordinates {(1,0.013) (2,0.015) (3,0.025) (4,0.031) (5,0.032) (6,0.034) };

\end{axis}
\end{tikzpicture}
\\
\textbf{(b) ModCloth}\\
\caption{The trend of the absolute performance disparity ($|{\pd}|$) in Task-R. Without the fairness loss, the $|{\pd}|$ is relatively high and often increase,
while with the fairness loss, particularly in FADE, the $|{\pd}|$ tends to remain relatively low.
}\label{fig:trunc-trend-taskR}
\vspace{-1mm}
\end{figure*}

\vspace{-1mm}
\subsubsection{Evaluation tasks.}
To evaluate the recommendation performance and PD,  we design two types of recommendation tasks: 
\begin{itemize}[leftmargin=*]
    \item Task-Remain (Task-R): Given the model trained up until time period $t$, the model is tested by recommending items for the remaining time periods with the test set $\D_t^{\text{test}}=\D_{t+1} \cup \dots \cup \D_T$.
    \item Task-Next (Task-N): Given the model trained up until time period $t$, the model is tested by recommending items for the right-next time period with the test set $\D_t^{\text{test}}=\D_{t+1}$.
\end{itemize}
Note that for both tasks, the data at the last time period, $\D_T$, is only used for testing and not for training purposes.
Due to space issue, we put the full results for Task-N in \S\ref{assec:main_result}. 

We use NDCG@20 and F1@20 to evaluate the top-20 recommendation quality.
We adopt a similar approach as previous studies~\cite{li2021user, kim2022meta}, where we randomly sample 100 items that the user has not interacted with as negative samples. These negative samples, along with the ground-truth items, are used for evaluation. 

\vspace{-1mm}
\subsubsection{Implementation details.}
For all compared methods, we set $n$ (the number of negative samples for BPR loss) to 4, the learning rate to 0.001, and L2 regularization to 0.0001. We use the Adam optimization algorithm~\cite{kingma2014adam} to update model parameters.

For \ours\ and \fullfair\ based on both MF and NCF, we set $\tau=3$, $\mu=4$, and the number of dynamic update epochs to 10, which consistently show excellent trade-off between performance and disparity across all metrics and datasets. The %parameter
$\lambda$ is selected within range $[0,4]$ %a range of 0 to 4
for \pretrainfair, \fullfair, \oursabs, and \ours\ in all cases. 
We use a random seed for better reproducibility.
For the implementation details of \rerank~\cite{li2021user}/\adver~\cite{li2021towards}, refer to \S\ref{assec:implentation}

\vspace{-2mm}
\subsection{The Effect of Learning from New Data}\label{ssec:new_data}
For RQ1 and RQ2,
we compare the recommendation performance and performance disparity, both averaged across each dynamic update data, of the five methods (\pretrain, \full, \fine, \pretrainfair, \fullfair) with \ours. 
Fig.~\ref{fig:trunc-tradeoff-taskR} shows the results w.r.t. different metrics, base recommender, and datasets.

First, compared to \pretrain, \full\ and \fine\ yield an average increase of 9.01\% and 4.61\%, respectively, in recommendation performance in \textit{all} cases, indicating that the new data is indeed useful for improving recommendation performance of the models over time. For \pretrainfair, \fullfair, and \ours, the similar trend is observed: an average increase of 4.66\% and 4.09\%, respectively.
However, in some cases on ModCloth, \ours\ performs worse than \pretrainfair\ due to the initial high disparity of \pretrainfair.

Regarding performance disparity, the \PD s of \full\ tend to exceed those of \pretrain, and those of \fine\ tend to fall below but still remain significant. This highlights the need to incorporate fairness considerations when integrating new data.

\vspace{-2mm}
\subsection{Ablation Study of \ours}

% \noindent\textbf{With and without fairness loss.}
\subsubsection{With and without fairness loss.}
To answer RQ2, we continue comparing \ours\ with aforementioned five methods.
First, regarding disparity, Fig.~\ref{fig:trunc-tradeoff-taskR} shows that \fullfair\ and \ours\ yield significantly lower \PD s compared to \full\ and \fine, in \textit{all} cases, with an average reduction of 47.60\% and 48.91\%, respectively. 
The results indicate that our fairness loss indeed helps reduce the performance disparity at each time period.

Furthermore, we examine how disparities change over time with \ours\ and the three methods, \full, \fullfair, \fine, as shown in Fig.~\ref{fig:trunc-trend-taskR}. We can see that without the fairness loss (\full/\fine), the \PD s tend to persist relatively high over time in all cases. 
However, when augmented with the fairness loss (\fullfair/\ours), the \PD s tend to remain stably low.

Besides significant reduction of \PD s, 
\finefair\ has merely marginal sacrifice (2.44\% on average) in recommendation performance compared to \fine,
and similar results are observed for \full\ and \fullfair, with an average decrease of 0.495\%.
This relatively slight decrease in recommendation performance is because \ours\ improves the performance of the disadvantaged group while reducing the performance of the advantaged group, in \textit{all} cases, with an average increase of 2.06\% and decrease of 3.37\%, respectively.

\vspace{-2mm}
\subsubsection{Fine-tuning v.s. Retraining}\label{ssec:finetuning}
Fig.~\ref{fig:trunc-tradeoff-taskR} shows that \fine\ consistently outperform \full\ w.r.t. both PD (an average decrease of 14.79\%) and recommendation performance (an average increase of 1.38\%) in all cases.
\ours\ outperform \fullfair\ w.r.t. PD (an average decrease of 16.47\%) while only slightly compromising recommendation performance (an average decrease of 0.61\%).
These results are consistent with our theoretical findings in \S\ref{ssec:theory}, indicating that retraining is more affected by distribution shifts, while fine-tuning can exponentially shrink this impact.
The lack of a clear advantage for fine-tuned models in recommendation performance is due to their eventual degradation after multiple periods, which is shown, for example, in the results for Movielenz in Fig.8 in \S\ref{assec:main_result}.

\vspace{-2mm}
\subsection{Comparison with Fairness Competitors}
To answer RQ3, we further compare \ours\ with the two fairness-aware competitors, \adver\ and \rerank, 
in Fig.~\ref{fig:trunc-tradeoff-taskR}. 
Note that all of those methods are implemented based on fine-tuning strategy for fair comparison.
First, \ours\ consistently achieves smaller \PD s, averaging 36.53\%, and it offers comparable recommendation performance on average 1.49\% better than \adver.
This is because \adver\ is not designed to reduce the performance gap between user groups; instead, its focus is on removing information related to sensitive attributes from user representations.

\rerank\ and \fine\ yield similar results in many cases, meaning that its re-ranking algorithm struggle to effectively re-rank the given recommendation lists. This is because the given base recommendation lists are already too unfair. For example, for disadvantaged users, the predicted scores may not accurately reflect the user's true interests, resulting in very low predicted scores for the ground-truth items in the list. This issue is exacerbated when the given recommendation lists are short, which is a common in practice. This observation agrees with our intuition that dynamic adaptation is necessary rather than using post-processing.

\vspace{-4mm}
\subsection{Time-efficiency Comparison}
To answer RQ4, we compare running time of \ours\ with the full-retraining based methods and the other fairness-aware techniques. The results are in Table~\ref{table:running_time} and each entry is the average running time of a model across the dynamic update data at each time period.

We have several observations based on the running time, averaged over base models and datasets. Firstly, \fine/\ours\ achieve approximately 323/270 times faster running time compared to \full/\fullfair, indicating that the fine-tuning strategy employed in \ours\ enables the models to achieve high time efficiency, making them ideal for dynamic settings.
Secondly, \fullfair/\ours\ exhibit approximately 1.06/1.27 times slower running time in comparison to \full/\fine. This suggests that the additional computational cost introduced by our fairness loss is not significant. 
Lastly, \ours\ demonstrates a time efficiency around 10.23 times and 94.11 times faster than \adver\ and \rerank, respectively, highlighting the lightweight design of our fairness loss compared to the existing fairness-aware losses.

\vspace{-3mm}
\subsection{Comparison with Soft Ranking Metrics}
Due to the space limit, the results for RQ5 are deferred to \S\ref{assec:soft_ranking}.
In essense, they show that \ours\ outperforms or matches the variant of \ours\ adapting NeuralNDCG in both recommendation performance and disparity, while being approximately four times faster. 
%This is because our differentiable Hit addresses NeuralNDCG's gradient vanishing issue by eliminating Sinkhorn's algorithm.

\begin{table}[t]
% \small
\Large
\centering
\caption{Efficiency comparison on the running time (seconds). 
%\jian{Efficiency comparison on the running time (seconds). -- is it better?} 
}
% \vspace{-0.1cm}
\label{table:running_time}
\resizebox{.45\textwidth}{!}{
\begin{tabular}{c|c|cc|cccc}
\toprule
\multirow{2}{*}{\textbf{Data}} & \multirow{2}{*}{\textbf{Models}} & \multicolumn{2}{c|}{Full-retrain-based} & \multicolumn{4}{c}{Fine-tune-based}  \\
&& \textbf{\full} & \textbf{\fullfair} & \textbf{\adver} & \textbf{\rerank} & \textbf{\fine} & \textbf{\ours} \\ \midrule
% \multicolumn{1}{c}{\textbf{\adver}} & \multicolumn{1}{c}{\textbf{\full}} & \multicolumn{1}{c}{\textbf{\fine}} & \multicolumn{1}{c}{\textbf{\fullfair}}  & \multicolumn{1}{c}{\textbf{\ours}}  \\ \midrule
% \textbf{Runtime (s)}  & \multicolumn{1}{c}{1332.58} & \multicolumn{1}{c}%{3074.88}
% {1698.27} & \multicolumn{1}{c}{3.29} & \multicolumn{1}{c}{9.60}  &  \multicolumn{1}{c}{102.18} \\
\multirow{2}{*}{\textbf{Movie.}} & \textbf{MF} & 1373.17 & 1401.18 & 55.16 & 132.46 & 2.57 & 4.08 \\
 & \textbf{NCF} & 1381.59 & 1488.5 & 61.66 &420.54 & 5.07  &5.93 \\
\multirow{2}{*}{\textbf{Mod.}} & \textbf{MF} & 154.22 &163.12  & 4.01  &250.75 & 0.79&0.93 \\
 & \textbf{NCF} &188.58  & 242.29 & 4.01 &344.51  & 1.15 &1.26  \\ \midrule
\multicolumn{2}{c|}{\textbf{Average}} & 774.39  & 823.77 & 31.21 & 287.06  & 2.40 & 3.05  \\
\bottomrule
\end{tabular}
}
% \vspace{-2mm}
\end{table}

\begin{figure}[t]
\footnotesize
\centering
\begin{tikzpicture}
\begin{customlegend}[legend columns=5,legend style={align=left,draw=none,column sep=1ex},
        legend entries={\textbf{Advantaged group}\text{  }, \textbf{Disadvantaged group}}]
        % \addlegendimage{draw=purple,mark=square, only marks}
        % \addlegendimage{draw=red,mark=o, only marks}
        \addlegendimage{draw=black,mark=square*}   
        \addlegendimage{draw=blue,color=blue,mark=*} 
        \end{customlegend}
\end{tikzpicture}
\\
% \textbf{(a) Task-R}\vspace{-1mm}\\
\begin{tikzpicture}
\begin{axis}[
height=2.6cm, width=4.7cm,
xtick={1, 2, 3, 4, 5, 6, 7,8,9,10,11},
xticklabels={0,0.1,0.3,0.5,0.8,1.0,1.5,2.0},
xticklabel style={align=center, font=\footnotesize,}, 
ylabel=NDCG@20 ,xlabel=$\lambda$,
y tick label style={/pgf/number format/.cd,fixed,fixed zerofill,precision=3,/tikz/.cd},]
\addplot[color=black,mark=square*,mark size=1.5pt,line width=0.6pt]
coordinates {(1,0.845) (2,0.845) (3,0.846) (4,0.847) (5,0.843) (6,0.841) (7,0.836) (8,0.825) };
\addplot[color=blue,mark=*,mark size=1.5pt,line width=0.6pt]
coordinates {(1,0.834) (2,0.835) (3,0.839) (4,0.844) (5,0.846) (6,0.848) (7,0.846) (8,0.848)  };
\end{axis}
\end{tikzpicture}
\hspace{-1mm}
\begin{tikzpicture}
\begin{axis}[
height=2.6cm, width=4.7cm,
xtick={1, 2, 3, 4, 5, 6, 7,8,9,10,11},
xticklabels={0,0.1,0.3,0.5,0.8,1.0,1.5,2.0},
xticklabel style={align=center, font=\footnotesize, }, 
ylabel=F1@20 ,xlabel=$\lambda$,
y tick label style={/pgf/number format/.cd,fixed,fixed zerofill,precision=3,/tikz/.cd},]
\addplot[color=black,mark=square*,mark size=1.5pt,line width=0.6pt]
coordinates {(1,0.339) (2,0.339) (3,0.339) (4,0.339) (5,0.336) (6,0.334) (7,0.330) (8,0.324)  };
\addplot[color=blue,mark=*,mark size=1.5pt,line width=0.6pt]
coordinates {(1,0.322) (2,0.322) (3,0.325) (4,0.326) (5,0.329) (6,0.330) (7,0.330) (8,0.330)   };
\end{axis}
\end{tikzpicture}
\hspace{-1mm}
\caption{The effect of the scaling parameter $\lambda$ on the performance of the advantaged and disadvantaged groups. 
%As $\lambda$ increases, the performance of the advantaged group tend to decrease while that of the disadvantaged group tend to increase, indicating that $\lambda$ effectively regulates the trade-off between the overall recommendation quality and the disparity in quality between two user groups.
}\label{fig:scaling}
\vspace{-4mm}
\end{figure}
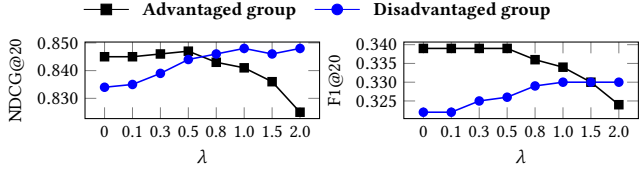
\begin{figure}[t]
\footnotesize
\centering
\begin{tikzpicture}
\begin{customlegend}[legend columns=5,legend style={align=left,draw=none,column sep=1ex},
        legend entries={\textbf{Advantaged group}\text{  }, \textbf{Disadvantaged group}}]
        % \addlegendimage{draw=purple,mark=square, only marks}
        % \addlegendimage{draw=red,mark=o, only marks}
        \addlegendimage{draw=black,mark=square*}   
        \addlegendimage{draw=blue,color=blue,mark=*} 
        \end{customlegend}
\end{tikzpicture}
\\
% \textbf{(a) Task-R}\vspace{-1mm}\\
\begin{tikzpicture}
\begin{axis}[
height=2.6cm, width=4.7cm,
xtick={1,2,3,4,5,6,7,8,9,10,11},
xticklabels={1,5,10,15,20,25,30,35,40,45,50},
ylabel=NDCG@20,xlabel= \textbf{(a)} The number of epochs,
y tick label style={/pgf/number format/.cd,fixed,fixed zerofill,precision=2,/tikz/.cd},]

\addplot[color=black,mark=square*,mark size=1.5pt,line width=0.6pt]
coordinates {(1,0.839) (2,0.851) (3,0.841) (4,0.795) (5,0.762) (6,0.732) (7,0.723) (8,0.717) (9,0.713) (10,0.705) (11,0.708) };
\addplot[color=blue,mark=*,mark size=1.5pt,line width=0.6pt]
coordinates {(1,0.830) (2,0.851) (3,0.848) (4,0.805) (5,0.764) (6,0.740) (7,0.728) (8,0.710) (9,0.708) (10,0.713) (11,0.705) };

\end{axis}
\end{tikzpicture}
\begin{tikzpicture}
\begin{axis}[
height=2.6cm, width=4.7cm,
xtick={1,2,3,4,5,6,7,8,9,10,11},
xticklabels={0.1,0.5,1.0,2.0,3.0,4.0,5.0},
ylabel=NDCG@20,xlabel= \textbf{(b)} Tau $\tau$,
y tick label style={/pgf/number format/.cd,fixed,fixed zerofill,precision=2,/tikz/.cd},]

\addplot[color=black,mark=square*,mark size=1.5pt,line width=0.6pt]
coordinates {(1,0.361) (2,0.666) (3,0.809) (4,0.839) (5,0.841) (6,0.845) (7,0.846) };
\addplot[color=blue,mark=*,mark size=1.5pt,line width=0.6pt]
coordinates {(1,0.855) (2,0.849) (3,0.851) (4,0.849) (5,0.848) (6,0.844) (7,0.842) };

\end{axis}
\end{tikzpicture}
\caption{Effect of hyperparamters. %The effect of the two hyperparameters: (a) the number of epochs of dynamic updates, (b) the temperature parameter $\tau$.
}\label{fig:params}
\vspace{-1mm}
\end{figure}
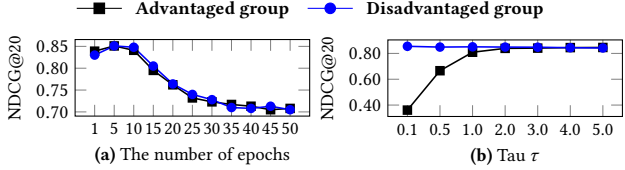

\vspace{-2mm}
\subsection{Hyperparameter Analysis}
For RQ6, we investigate the sensitivity of \ours\ to four hyperparameters: (1) the scaling parameter $\lambda$, (2) the number of epochs of dynamic updates, (3) the temperature parameter $\tau$, and (4) the number of negative items $\mu$. Due to the space limit, we only show the results of \ours\ based on MF on Movielenz for $\lambda$, the number of epochs, and $\tau$ in Figs.~\ref{fig:scaling} \& \ref{fig:params}. Please refer to \S\ref{assec:hyperparam} for the full results. They illustrate the performance of the advantaged and disadvantaged user groups for different values of these hyperparameters. 

\vspace{-1mm}
\subsubsection{Effect of scaling parameter $\lambda$ for the fairness Loss}
Fig.~\ref{fig:scaling} shows that the performance of the advantaged group tend to decrease while that of the disadvantaged group tend to increase as $\lambda$ increases. 
In other words, the performance disparity between the two user groups steadily reduces until $\lambda$ reaches an optimal value, which varies depending on the specific metric used. 
The results indicate that $\lambda$ effectively controls the trade-off between recommendation performance and performance disparity.

\vspace{-1mm}
\subsubsection{Effect of the number of epochs of dynamic updates.}
Fig.~\ref{fig:params}-(a) shows that the performance of both user groups increases as the number of epochs of dynamic fine-tuning increases until reaching a peak around epoch 5 or 10. Subsequently, the performance gradually declines with further increases in the number of epochs. 
We suspect that setting the number of epochs too low may result in the model not learning enough from the current data. Conversely, when the number of epochs is set too high, the model potentially loses the knowledge acquired from historical data.
We argue that this phenomenon is well-suited for the dynamic environment, as setting a low value for the number of epochs results in high %time
efficiency.

\vspace{-1mm}
\subsubsection{Effect of temperature parameter $\tau$ in the relaxed permutation matrix.} 
Higher values of $\tau$ %($>\!0$)
result in smoother rows in the relaxed permutation matrix, $\widehat{\P}_{u}[i,:]$.
Fig~\ref{fig:params}-(b) shows that the performance of both user groups increases until $\tau=2$, and then stabilizes.
These findings indicate that \ours\ is not highly sensitive to $\tau$, consistently delivering excellent performance for both user groups as long as $\tau$ is not too small.
When $\tau$ is set too low, the Gumbel-softmax distribution becomes sharp, resulting in a nearly deterministic decision-making process for the model, i.e., $\widehat{\P}_{u}[i,:]$ will be close to the one-hot vector of the $i$-th ranked item. As a result, the entry corresponding to the positive item in that vector is likely to have an extremely small value, from the initial phase of training, potentially hindering the the fairness regularization.
\vspace{-2mm}
\section{Related Work}
\label{sec:related_work}

\vspace{0mm}
\noindent\textbf{Dynamic recommender systems.}
Instead of fully retraining with the entire dataset when new data is collected, which can be time-inefficient, we can fine-tune the model parameters using only the new data, which is referred to as dynamic/online recommender systems in the literature. To effectively learn from relatively sparse new data, several methods have been proposed based on reweighting either (1) the impact of each user-item interaction~\cite{he2016fast, shu2019meta} or (2) that of each model parameter~\cite{du2019sequential, li2017meta, zhang2020retrain}; \cite{kim2022meta} utilizes both approaches. One unique advantage of the fairness loss in \ours\ is that it can be easily applied to any existing dynamic recommender systems optimized using gradient-based algorithms.

\noindent\textbf{Fair recommender systems in dynamic scenarios.}
Various fairness demands exist in recommender systems, including user-side~\cite{tang2023fairness}, item-side~\cite{chen2023fairly}, and two-sided fairness~\cite{wu2021tfrom}, as well as fairness on unipartite networks~\cite{akpinar2022long}.
User-side fairness ensures fair recommendation quality for different users, while item-side fairness concentrates on equal exposure opportunities for items 
%in recommendations. 
Two-sided fairness seeks to balance these two aspects. While the literature~\cite{zhu2021popularity, ge2021towards, morik2020controlling} has addressed item-side fairness in dynamic recommendations, such as the work by \cite{zhu2021popularity} that scales predicted ratings by item popularity with higher strength over time, user-side fairness in dynamic settings remains unexplored, to our knowledge.
%to the best of our knowledge.
%do2021two,
%,li2021user
%singh2018fairness,

As described in Section \ref{sec:introduction}, existing user-side fairness-aware re-ranking methods~\cite{li2021user,fu2020fairness} face the difficulties in dynamic settings.
% following challenges in dynamic settings: (1) Frequent model updates, (2) Distribution shifts, and (3) Non-differentiabilty of ranking metrics.
These methods tend to be time-inefficient, involving optimization problem akin to 0-1 integer programming problem. Furthermore, their non-differentiable fairness constraint, separating fairness optimization from that of recommendation quality, precludes model parameters from being regularized by fairness constraints. This hinders adaptation to distribution shifts in dynamic settings.

Another line of research into user-side fairness~\cite{beigi2020privacy, wu2021fairness,wu2021learning} employs adversarial functions to generate fair user representations independent of sensitive user attributes. However, these formulations do not explicitly address the reduction of performance disparity.

\vspace{-2mm}
\section{Conclusion}
\label{sec:conclusion}

In this paper, we study the problem of user-side fairness in the dynamic recommendation scenario. 
We point out three key challenges in this problem: (1) distribution shifts, (2) frequent model updates, and (3) non-differentiability of ranking metrics. 
To address these challenges, we begin with theoretical analyses on fine-tuning v.s. retraining, showing that the best practice is incremental fine-tuning with restart. Guided by these insights, we propose \ul{FA}ir \ul{D}ynamic r\ul{E}commender (\ours), an end-to-end fine-tuning framework that dynamically ensures user-side fairness over time. It incorporates our fairness loss equipped with our lightweight Differentiable Hit, which enhances efficiency over the recent NeuralNDCG method.
%alleviating the gradient vanishing issue in the recent NeuralNDCG method and enhances efficiency. 
Through extensive experiments, we verify that \ours\ effectively and efficiently alleviates the performance disparity without significantly sacrificing recommendation performance.

\vspace{-2mm}
\begin{acks}
This work is partially supported by NSF (1947135, 2134079, 1939725), DHS (17STQAC00001-07-00), and NIFA (2020-67021-32799).
\end{acks}

\normalem
\balance
\bibliographystyle{ACM-Reference-Format}

\bibliography{references}

% \normalem
% \bibliographystyle
% \balance
% \bibliography

\newpage
\appendix

\section{Theoretical Analyses}\label{app:pf}

\subsection{Assumptions}\label{app:assum}

In this subsection, we introduce our theoretical assumptions, which are quite mild and realistic.
To ensure that the dataset has a good coverage of the underlying distribution, a common assumption in literature is independence:
\begin{ASS}[Data independence]
For every $t$, the data tuples in $\CAL D_t$ are mutually independent.
\end{ASS}

Regarding the loss function, a well-behaved loss function should be able to be minimized. Common loss functions satisfy this property. This leads us to the following Assumption~\ref{ASS:inf}:
\begin{ASS}[Existence of infima]\label{ASS:inf}
For every $t$, the infimum $\CAL L_t^*:=\inf_{\CAL W}\CAL L_t(\CAL W)$ exists.
\end{ASS}
Note that we do not assume \emph{realizability}, i.e., we do not assume that there exists $\CAL W$ that can achieve this infimum. Our Assumption~\ref{ASS:inf} is realistic in machine learning. For example, neural networks can arbitrarily approximate any continuous function over any compact domain \cite{pinkus1999approximation}, but they may not be exactly equal to that function.

Besides that, since data tuples are mutually independent, each data tuple in the dataset should not have dominant influence on the overall loss function, which means that the loss function should use the whole dataset. This leads us to the following Assumption~\ref{ASS:subg}:
\begin{ASS}[No dominant influence]\label{ASS:subg}
For every $t$, for each data tuple $z\in\CAL D_t$, we assume that $\sup_{\CAL W}|\CAL L^{\CAL D_t}(\CAL W)-\CAL L_t(\CAL W)|$ conditioned on $\CAL D_t\setminus\{z\}$ is $\big(\frac{\varsigma\sqrt{\log m_t}}{m_t}\big)^2$-subgaussian. Without loss of generality, we can assume $\varsigma=1$ by rescaling $\CAL L$.
\end{ASS}
The subgaussian property is a common assumption in machine learning \cite{maurer2021concentration}, and $\varsigma$ in our Assumption~\ref{ASS:subg} can be viewed as a generalization of the Vapnik--Chervonenkis dimension \cite{vapnik1971uniform} and Pollard's pseudodimension \cite{pollard1990section}. Since there exist various definitions of the subgaussian property (yet equivalent up to constant factors), we clarify our definition as follows:
\begin{DEF}[Subgaussian property]
For $\varsigma>0$, a real-valued random variable $X$ is said to be \emph{$\varsigma^2$-subgaussian} if
\EQ{\Exp[\RM e^{v(X-\Exp[X])}]\le\RM e^{\varsigma^2v^2/2},\quad\forall v\in\BB R.}
\end{DEF}
The equality holds for univariate Gaussians with variance $\varsigma^2$.

Finally, we state our assumption on fine-tuning and retraining. For each $t\ge1$, let $\CAL W_t^\textnormal{ft}$ denote the model parameters fine-tuned till $\CAL D_t$. To characterize the fact that the fine-tuned $\CAL W_t^\textnormal{ft}$ does not completely forget the previously learned knowledge in $\CAL W_{t-1}^\textnormal{ft}$, we assume that all time periods share the same parameter space and use the following classic Assumption~\ref{ASS:prox} (adapted from \cite{rajeswaran2019meta}):
\begin{ASS}[Proximal fine-tuning]\label{ASS:prox}
There is $0<\gamma<1$ such that for each $t\ge1$, the number of fine-tuning epochs is decided such that the fine-tuned $\CAL W_t^\textnormal{ft}$ is minimizing
\EQ{\ell_t(\CAL W):=(1-\gamma)\CAL L^{\CAL D_t}(\CAL W)+\gamma\ell_{t-1}(\CAL W),}
where $\ell_0(\CAL W):=\CAL L^{\CAL D_0}(\CAL W)$ denotes the pretraining loss function.
\end{ASS}

For retraining, we assume that the influence of each time period $t$ to the retraining loss is a proportional to the size $m_t$ of $\CAL D_t$:
\begin{ASS}[Retraining loss]\label{ASS:rt}
\EQ{\CAL L^\textnormal{rt}_{t_\textnormal{te}-1}(\CAL W):=\frac{\sum_{t=0}^{t_\textnormal{te}-1}m_t\CAL L^{\CAL D_t}(\CAL W)}{\sum_{t=0}^{t_\textnormal{te}-1}m_t}.}
\end{ASS}
Although this is a simplification of the retraining loss in practice, it still captures the essential properties of retraining.

\subsection{Proofs of Theorems~\ref{thm:ft} \& \ref{thm:rt}}\label{app:pf-rt-ft}

Our proofs of Theorems~\ref{thm:ft} \& \ref{thm:rt} rely on the following Lemma~\ref{lem:alpha}.

\begin{lemma}\label{lem:alpha}
For $\BM\alpha\in\BB R_{\ge0}^{t_\textnormal{te}}$ with $\sum_{t=0}^{t_\textnormal{te}-1}\alpha_t=1$ and for $\epsilon>0$, let $\CAL W^{\BM\alpha,\epsilon}_{t_\textnormal{te}-1}$ denote some model parameters such that
\EQ{\sum_{t=0}^{t_\textnormal{te}-1}\alpha_t\CAL L_t(\CAL W^{\BM\alpha,\epsilon}_{t_\textnormal{te}-1})\le\epsilon+\inf_{\CAL W}\sum_{t=0}^{t_\textnormal{te}-1}\alpha_t\CAL L^{\CAL D_t}\!(\CAL W).}
Then with probability at least $1-\delta$,
\AM{\CAL L_{t_\textnormal{te}}(\CAL W^{\BM\alpha,\epsilon}_{t_\textnormal{te}-1})\le\CAL L_{t_\textnormal{te}}^*+\epsilon+2\sum_{t=0}^{t_\textnormal{te}-1}\alpha_td_{t,t_\textnormal{te}}+4\sqrt{\sum_{t=0}^{t_\textnormal{te}-1}\frac{\alpha_t^2}{\frac{m_t}{\log m_t}}\log\frac2\delta}.}
\end{lemma}

\begin{proof}[Proof of Lemma~\ref{lem:alpha}]
Generalized from \cite{ben2010theory}. For $k\ge1$, let
\AL{\CAL W^k_{t}&\in\CAL L_t^{-1}\big(\big({-\infty},\CAL L_t^*+\tfrac1k\big]\big),\\
\CAL W^k_{t,t_\textnormal{te}}&\in(\CAL L_t+\CAL L_{t_\textnormal{te}})^{-1}\big(\big({-\infty},\CAL L_t^*+\CAL L_{t_\textnormal{te}}^*+d_{t,t_\textnormal{te}}^\text{comb}+\tfrac1k\big]\big).}
Then for any $\CAL W$, by the triangle inequality,
\AL{{}&\big|(\CAL L_t(\CAL W)-\CAL L_t^*)-(\CAL L_{t_\textnormal{te}}(\CAL W)-\CAL L_{t_\textnormal{te}}^*)\big|\\
={}&\big||\CAL L_t(\CAL W)-\CAL L_t^*|-|\CAL L_{t_\textnormal{te}}(\CAL W)-\CAL L_{t_\textnormal{te}}^*|\big|\\
={}&\big|\big(|\CAL L_t(\CAL W)-\CAL L_t(\CAL W^k_{t,t_\textnormal{te}})|-|\CAL L_{t_\textnormal{te}}(\CAL W)-\CAL L_{t_\textnormal{te}}(\CAL W^k_{t,t_\textnormal{te}})|\big)\nonumber\\
&+\big(|\CAL L_t(\CAL W)-\CAL L_t^*|-|\CAL L_t(\CAL W)-\CAL L_t(\CAL W^k_{t,t_\textnormal{te}})|\big)\\
&-\big(|\CAL L_{t_\textnormal{te}}(\CAL W)-\CAL L_{t_\textnormal{te}}^*|-|\CAL L_{t_\textnormal{te}}(\CAL W)-\CAL L_{t_\textnormal{te}}(\CAL W^k_{t,t_\textnormal{te}})|\big)\big|\nonumber\\
\le{}&\big||\CAL L_t(\CAL W)-\CAL L_t(\CAL W^k_{t,t_\textnormal{te}})|-|\CAL L_{t_\textnormal{te}}(\CAL W)-\CAL L_{t_\textnormal{te}}(\CAL W^k_{t,t_\textnormal{te}})|\big|\nonumber\\
&+\big||\CAL L_t(\CAL W)-\CAL L_t^*|-|\CAL L_t(\CAL W)-\CAL L_t(\CAL W^k_{t,t_\textnormal{te}})|\big|\\
&+\big||\CAL L_{t_\textnormal{te}}(\CAL W)-\CAL L_{t_\textnormal{te}}^*|-|\CAL L_{t_\textnormal{te}}(\CAL W)-\CAL L_{t_\textnormal{te}}(\CAL W^k_{t,t_\textnormal{te}})|\big|\nonumber\\
\le{}&d^{\CAL H\Delta\CAL H}_{t,t_\textnormal{te}}+\big|(\CAL L_t(\CAL W)-\CAL L_t^*)-(\CAL L_t(\CAL W)-\CAL L_t(\CAL W^k_{t,t_\textnormal{te}}))\big|\nonumber\\
&+\big|(\CAL L_{t_\textnormal{te}}(\CAL W)-\CAL L_{t_\textnormal{te}}^*)-(\CAL L_{t_\textnormal{te}}(\CAL W)-\CAL L_{t_\textnormal{te}}(\CAL W^k_{t,t_\textnormal{te}}))\big|\\
={}&d^{\CAL H\Delta\CAL H}_{t,t_\textnormal{te}}+|\CAL L_t(\CAL W_{t,t_\textnormal{te}}^k)-\CAL L_t^*|+|\CAL L_{t_\textnormal{te}}(\CAL W_{t,t_\textnormal{te}}^k)-\CAL L_{t_\textnormal{te}}^*|\\
={}&d^{\CAL H\Delta\CAL H}_{t,t_\textnormal{te}}+\CAL L_t(\CAL W_{t,t_\textnormal{te}}^k)-\CAL L_t^*+\CAL L_{t_\textnormal{te}}(\CAL W_{t,t_\textnormal{te}}^k)-\CAL L_{t_\textnormal{te}}^*\\
\le{}&d^{\CAL H\Delta\CAL H}_{t,t_\textnormal{te}}+d^\text{comb}_{t,t_\textnormal{te}}+\frac1k\\
={}&d_{t,t_\textnormal{te}}+\frac1k
.}
Thus,
\AL{
&\bigg|\sum_{t=0}^{t_\textnormal{te}-1}\alpha_t(\CAL L_t(\CAL W)-\CAL L_t^*)-(\CAL L_{t_\textnormal{te}}(\CAL W)-\CAL L_{t_\textnormal{te}}^*)\bigg|\\
={}&\bigg|\sum_{t=0}^{t_\textnormal{te}-1}\alpha_t(\CAL L_t(\CAL W)-\CAL L_t^*)-\sum_{t=0}^{t_\textnormal{te}-1}\alpha_t(\CAL L_{t_\textnormal{te}}(\CAL W)-\CAL L_{t_\textnormal{te}}^*)\bigg|\\
={}&\bigg|\sum_{t=0}^{t_\textnormal{te}-1}\alpha_t((\CAL L_t(\CAL W)-\CAL L_t^*)-(\CAL L_{t_\textnormal{te}}(\CAL W)-\CAL L_{t_\textnormal{te}}^*))\bigg|\\
\le{}&\sum_{t=0}^{t_\textnormal{te}-1}\alpha_t|(\CAL L_t(\CAL W)-\CAL L_t^*)-(\CAL L_{t_\textnormal{te}}(\CAL W)-\CAL L_{t_\textnormal{te}}^*)|\\
\le{}&\sum_{t=0}^{t_\textnormal{te}-1}\alpha_t\Big(d_{t,t_\textnormal{te}}+\frac1k\Big)\\
={}&\sum_{t=0}^{t_\textnormal{te}-1}\alpha_td_{t,t_\textnormal{te}}+\frac1k
.}

Besides that, by Theorem~3 in \cite{maurer2021concentration} and Assumption~\ref{ASS:subg},
\AL{
&\Prb\bigg\{\sup_{\CAL W}\bigg|\sum_{t=0}^{t_\textnormal{te}-1}\alpha_t\CAL L^{D_t}(\CAL W)-\sum_{t=0}^{t_\textnormal{te}-1}\alpha_t\CAL L_t(\CAL W)\bigg|\ge\epsilon\bigg\}\\
\le{}&\Prb\bigg\{\sum_{t=0}^{t_\textnormal{te}-1}\alpha_t\sup_{\CAL W}|\CAL L^{D_t}(\CAL W)-\CAL L_t(\CAL W)|\ge\epsilon\bigg\}\\
\le{}&2\exp\bigg({-\frac{\epsilon^2}{4\sum_{t=0}^{t_\textnormal{te}-1}m_t\big(\alpha_t\frac{\varsigma\sqrt{\log m_t}}{m_t}\big)^2}}\bigg)\\
={}&2\exp\bigg({-\frac{\epsilon^2}{4\varsigma^2\sum_{t=0}^{t_\textnormal{te}-1}\frac{\alpha_t^2}{\frac{m_t}{\log m_t}}}}\bigg).
}

Then for $\varsigma=1$, with probability at least $1-\delta$, for all $\CAL W$, %simultaneously
\EQ{\bigg|\sum_{t=0}^{t_\textnormal{te}-1}\alpha_t\CAL L^{\CAL D_t}\!(\CAL W)-\sum_{t=0}^{t_\textnormal{te}-1}\alpha_t\CAL L_t(\CAL W)\bigg|\le2\sqrt{\sum_{t=0}^{t_\textnormal{te}-1}\frac{\alpha_t^2}{\frac{m_t}{\log m_t}}\log\frac2\delta}.}
Thus,
\AL{&\CAL L_{t_\textnormal{te}}(\CAL W^{\BM\alpha,\epsilon}_{t_\textnormal{te}-1})\\
={}&\CAL L_{t_\textnormal{te}}^*+\CAL L_{t_\textnormal{te}}(\CAL W^{\BM\alpha,\epsilon}_{t_\textnormal{te}-1})-\CAL L_{t_\textnormal{te}}^*\\
\le{}&\CAL L_{t_\textnormal{te}}^*+\sum_{t=0}^{t_\textnormal{te}-1}\alpha_t(\CAL L_t(\CAL W^{\BM\alpha,\epsilon}_{t_\textnormal{te}-1})-\CAL L_t^*)+\sum_{t=0}^{t_\textnormal{te}-1}\alpha_td_{t,t_\textnormal{te}}+\frac1k\\
\le{}&\CAL L_{t_\textnormal{te}}^*+\sum_{t=0}^{t_\textnormal{te}-1}\alpha_t(\CAL L^{\CAL D_t}\!(\CAL W^{\BM\alpha,\epsilon}_{t_\textnormal{te}-1})-\CAL L_t^*)+\sum_{t=0}^{t_\textnormal{te}-1}\alpha_td_{t,t_\textnormal{te}}\nonumber\\&+\frac1k+2\sqrt{\sum_{t=0}^{t_\textnormal{te}-1}\frac{\alpha_t^2}{\frac{m_t}{\log m_t}}\log\frac2\delta}\\
\le{}&\CAL L_{t_\textnormal{te}}^*+\epsilon+\sum_{t=0}^{t_\textnormal{te}-1}\alpha_t(\CAL L^{\CAL D_t}\!(\CAL W^k_{t_\textnormal{te}})-\CAL L_t^*)+\sum_{t=0}^{t_\textnormal{te}-1}\alpha_td_{t,t_\textnormal{te}}\nonumber\\&+\frac1k+2\sqrt{\sum_{t=0}^{t_\textnormal{te}-1}\frac{\alpha_t^2}{\frac{m_t}{\log m_t}}\log\frac2\delta}\\
\le{}&\CAL L_{t_\textnormal{te}}^*+\epsilon+\sum_{t=0}^{t_\textnormal{te}-1}\alpha_t(\CAL L_t(\CAL W^k_{t_\textnormal{te}})-\CAL L_t^*)+\sum_{t=0}^{t_\textnormal{te}-1}\alpha_td_{t,t_\textnormal{te}}\nonumber\\&+\frac1k+4\sqrt{\sum_{t=0}^{t_\textnormal{te}-1}\frac{\alpha_t^2}{\frac{m_t}{\log m_t}}\log\frac2\delta}\\
\le{}&\CAL L_{t_\textnormal{te}}^*+\epsilon+\CAL L_{t_\textnormal{te}}(\CAL W^k_{t_\textnormal{te}})-\CAL L_{t_\textnormal{te}}^*+2\sum_{t=0}^{t_\textnormal{te}-1}\alpha_td_{t,t_\textnormal{te}}\nonumber\\
&+\frac2k+4\sqrt{\sum_{t=0}^{t_\textnormal{te}-1}\frac{\alpha_t^2}{\frac{m_t}{\log m_t}}\log\frac2\delta}\\
\le{}&\CAL L_{t_\textnormal{te}}^*+\epsilon+2\sum_{t=0}^{t_\textnormal{te}-1}\alpha_td_{t,t_\textnormal{te}}+\frac3k+4\sqrt{\sum_{t=0}^{t_\textnormal{te}-1}\frac{\alpha_t^2}{\frac{m_t}{\log m_t}}\log\frac2\delta}
.} 
It follows from the continuity of probability that
\AL{&\Prb\bigg\{\CAL L_{t_\textnormal{te}}(\CAL W^{\BM\alpha,\epsilon}_{t_\textnormal{te}-1})>\CAL L_{t_\textnormal{te}}^*+\epsilon+2\sum_{t=0}^{t_\textnormal{te}-1}\alpha_td_{t,t_\textnormal{te}}\nonumber\\
&\qquad\qquad\qquad\qquad\;\;+4\sqrt{\sum_{t=0}^{t_\textnormal{te}-1}\frac{\alpha_t^2}{\frac{m_t}{\log m_t}}\log\frac2\delta}\bigg\}\\
={}&\Prb\bigg[\bigcup_{k=1}^\infty\bigg\{\CAL L_{t_\textnormal{te}}(\CAL W^{\BM\alpha,\epsilon}_{t_\textnormal{te}-1})\ge\CAL L_{t_\textnormal{te}}^*+\epsilon+2\sum_{t=0}^{t_\textnormal{te}-1}\alpha_td_{t,t_\textnormal{te}}\nonumber\\
&\qquad\qquad\qquad\qquad\quad\;\;+\frac3k+4\sqrt{\sum_{t=0}^{t_\textnormal{te}-1}\frac{\alpha_t^2}{\frac{m_t}{\log m_t}}\log\frac2\delta}\bigg\}\bigg]\\
={}&\lim_{k\to\infty}\Prb\bigg\{\CAL L_{t_\textnormal{te}}(\CAL W^{\BM\alpha,\epsilon}_{t_\textnormal{te}-1})\ge\CAL L_{t_\textnormal{te}}^*+\epsilon+2\sum_{t=0}^{t_\textnormal{te}-1}\alpha_td_{t,t_\textnormal{te}}\nonumber\\
&\qquad\qquad\qquad\qquad\quad\;\;+\frac3k+4\sqrt{\sum_{t=0}^{t_\textnormal{te}-1}\frac{\alpha_t^2}{\frac{m_t}{\log m_t}}\log\frac2\delta}\bigg\}\\
\le{}&\lim_{k\to\infty}\delta=\delta
.\qedhere}
\end{proof}

\newpage
\begin{corollary}\label{cor:alpha}
Under the setup of Lemma~\ref{lem:alpha}, let
\EQ{\CAL L_{t_\textnormal{te}}^{\BM\alpha}:=\inf_{\begin{subarray}{c}
\epsilon \in \BB Q_{>0}
\end{subarray}}\CAL L_{t_\textnormal{te}}(\CAL W^{\BM\alpha,\epsilon}_{t_\textnormal{te}-1})}
denote the best possible loss w.r.t.\ $\BM\alpha$. With probability at least $1-\delta$,
\EQ{\CAL L_{t_\textnormal{te}}^{\BM\alpha}\le\CAL L_{t_\textnormal{te}}^*+2\sum_{t=0}^{t_\textnormal{te}-1}\alpha_td_{t,t_\textnormal{te}}+4\sqrt{\sum_{t=0}^{t_\textnormal{te}-1}\frac{\alpha_t^2}{\frac{m_t}{\log m_t}}\log\frac2\delta}.}
\end{corollary}

\begin{proof}[Proof of Corollary~\ref{cor:alpha}]
By the continuity of probability,
\AL{
&\Prb\bigg\{\CAL L_{t_\textnormal{te}}^{\BM\alpha}>\CAL L_{t_\textnormal{te}}^*+2\sum_{t=0}^{t_\textnormal{te}-1}\alpha_td_{t,t_\textnormal{te}}+4\sqrt{\sum_{t=0}^{t_\textnormal{te}-1}\frac{\alpha_t^2}{\frac{m_t}{\log m_t}}\log\frac2\delta}\bigg\}\\
={}&\Prb\bigg[\bigcup_{k=1}^\infty\bigg\{\CAL L_{t_\textnormal{te}}^{\BM\alpha}\ge\CAL L_{t_\textnormal{te}}^*+\frac1k+2\sum_{t=0}^{t_\textnormal{te}-1}\alpha_td_{t,t_\textnormal{te}}+4\sqrt{\sum_{t=0}^{t_\textnormal{te}-1}\frac{\alpha_t^2}{\frac{m_t}{\log m_t}}\log\frac2\delta}\bigg\}\bigg]\\
={}&\lim_{k\to\infty}\Prb\bigg\{\CAL L_{t_\textnormal{te}}^{\BM\alpha}\ge\CAL L_{t_\textnormal{te}}^*+\frac1k+2\sum_{t=0}^{t_\textnormal{te}-1}\alpha_td_{t,t_\textnormal{te}}+4\sqrt{\sum_{t=0}^{t_\textnormal{te}-1}\frac{\alpha_t^2}{\frac{m_t}{\log m_t}}\log\frac2\delta}\bigg\}\\
\le{}&\limsup_{k\to\infty}\Prb\bigg\{\CAL L_{t_\textnormal{te}}\big(\CAL W^{\BM\alpha,\frac1k}_{t_\textnormal{te}-1}\big)\ge\CAL L_{t_\textnormal{te}}^*+\frac1k+2\sum_{t=0}^{t_\textnormal{te}-1}\alpha_td_{t,t_\textnormal{te}}\nonumber\\&\qquad\qquad\qquad\qquad\qquad\;\;+4\sqrt{\sum_{t=0}^{t_\textnormal{te}-1}\frac{\alpha_t^2}{\frac{m_t}{\log m_t}}\log\frac2\delta}\bigg\}\\
\le{}&\limsup_{k\to\infty}\delta=\delta
.\qedhere}
\end{proof}

Now we give the proofs of Theorems~\ref{thm:ft} \& \ref{thm:rt}.

\begin{proof}[Proof of Theorem~\ref{thm:ft}]
By Assumption~\ref{ASS:prox},
\AL{\ell_{t_\textnormal{te}-1}(\CAL W)&=(1-\gamma)\CAL L^{\CAL D_{t_\textnormal{te}-1}}(\CAL W)+\gamma\ell_{t_\textnormal{te}-2}(\CAL W)\\
&=\gamma^{t_\textnormal{te}-1}\CAL L^{\CAL D_0}(\CAL W)+\sum_{t=1}^{t_\textnormal{te}-1}(1-\gamma)\gamma^{t_\textnormal{te}-t-1}\CAL L^{\CAL D_t}(\CAL W).}
Thus, the coefficients are
\EQ{\alpha_t^\text{ft}:=\begin{cases}
\gamma^{t_\textnormal{te}-1},&\text{for }t=0,\\
(1-\gamma)\gamma^{t_\textnormal{te}-t-1},&\text{for }t=1,\dots,t_\textnormal{te}-1.
\end{cases}}
Since $\sum_{t=0}^{t_\textnormal{te}-1}\alpha_t^\text{ft}=1$, then by Corollary~\ref{cor:alpha},
\AL{&\CAL L^\text{ft}_{t_\textnormal{te}}=\CAL L_{t_\textnormal{te}}^{\BM\alpha^\text{ft}}\\
\le{}&\CAL L_{t_\textnormal{te}}^*+2\sum_{t=0}^{t_\textnormal{te}-1}\alpha_t^\text{ft}d_{t,t_\textnormal{te}}+4\sqrt{\sum_{t=0}^{t_\textnormal{te}-1}\frac{(\alpha_t^\text{ft})^2}{\frac{m_t}{\log m_t}}\log\frac2\delta}\\
%={}&\CAL L_{t_\textnormal{te}}^*+\frac{(1-\gamma)\Big(2\!\!\sum\limits_{t=0}^{t_\textnormal{te}-1}\gamma^{t_\textnormal{te}-t-1}d_{t,t_\textnormal{te}}+4\sqrt{\big(\frac{\gamma^{2t_\textnormal{te}-2}}{\frac{m_0}{\log m_0}}+\frac{1-\gamma^{2t_\textnormal{te}-2}}{(1-\gamma^2)\frac{m_1}{\log m_1}}\big)\log\frac2\delta}\Big)}{1-\gamma^{t_\textnormal{te}}}
={}&\CAL L_{t_\textnormal{te}}^*+2\gamma^{t_\textnormal{te}-1}d_{0,t_\textnormal{te}}+2\sum\limits_{t=1}^{t_\textnormal{te}-1}(1-\gamma)\gamma^{t_\textnormal{te}-t-1}d_{t,t_\textnormal{te}}\\
&+4\sqrt{\Big(\frac{\gamma^{2t_\textnormal{te}-2}}{\frac{m_0}{\log m_0}}+\frac{(1+\gamma)(1-\gamma^{2t_\textnormal{te}-4})}{(1-\gamma)\frac{m_1}{\log m_1}}\Big)\log\frac2\delta}\nonumber
.\qedhere}
\end{proof}

\begin{proof}[Proof of Theorem~\ref{thm:rt}]
By Assumption~\ref{ASS:rt}, we have
\EQ{\alpha_t^\text{rt}:=\frac{m_t}{\sum_{t'=0}^{t_\textnormal{te}-1}m_{t'}}=\frac{m_t}{m_0+(t_\textnormal{te}-1)m_1}.}
It follows from Corollary~\ref{cor:alpha} that
\AL{&\CAL L^\text{rt}_{t_\textnormal{te}}=\CAL L_{t_\textnormal{te}}^{\BM\alpha^\text{rt}}\\
\le{}&\CAL L_{t_\textnormal{te}}^*+2\sum_{t=0}^{t_\textnormal{te}-1}\alpha_t^\text{rt}d_{t,t_\textnormal{te}}+4\sqrt{\sum_{t=0}^{t_\textnormal{te}-1}\frac{(\alpha_t^\text{rt})^2}{\frac{m_t}{\log m_t}}\log\frac2\delta}\\
={}&\CAL L_{t_\textnormal{te}}^*+\frac{2m_0d_{0,t_\textnormal{te}}+2\sum\limits_{t=1}^{t_\textnormal{te}-1}m_1d_{t,t_\textnormal{te}}}{m_0+(t_\textnormal{te}-1)m_1}+4\sqrt{\frac{\sum_{t=0}^{t_\textnormal{te}-1}m_t\log m_t}{(m_0+(t_\textnormal{te}-1)m_1)^2}\log\tfrac2\delta}\\
\le{}&\CAL L_{t_\textnormal{te}}^*+\frac{2m_0d_{0,t_\textnormal{te}}+2\sum\limits_{t=1}^{t_\textnormal{te}-1}m_1d_{t,t_\textnormal{te}}}{m_0+(t_\textnormal{te}-1)m_1}+4\sqrt{\frac{\sum_{t=0}^{t_\textnormal{te}-1}m_t\log m_0}{(m_0+(t_\textnormal{te}-1)m_1)^2}\log\tfrac2\delta}\\
={}&\CAL L_{t_\textnormal{te}}^*+\frac{2m_0d_{0,t_\textnormal{te}}+2\sum\limits_{t=1}^{t_\textnormal{te}-1}m_1d_{t,t_\textnormal{te}}}{m_0+(t_\textnormal{te}-1)m_1}+4\sqrt{\frac{\log m_0}{m_0+(t_\textnormal{te}-1)m_1}\log\tfrac2\delta}
.\qedhere}
\end{proof}

\subsection{Proof of Proposition~\ref{prop:fair}}\label{app:prop:fair}
\begin{proof}[Proof of Proposition~\ref{prop:fair}]
Note that
\AL{\nabla_{\W_t}\CAL L^{\D_t}_\textnormal{fair}(\W_t)&=\nabla_{\W_t}(-\log(\sigma({-\diffpd^{\D_t}(\W_t)})))\\
&=\sigma(\diffpd^{\D_t}(\W_t))\nabla_{\W_t}\diffpd^{\D_t}(\W_t).}
Since $\nabla_{\W_t}\CAL L^{\D_t}_\textnormal{fair}(\W_t)\ne\BM0$, %and $\sigma\ne0$, 
then $\nabla_{\W_t}\diffpd^{\D_t}(\W_t)\ne\BM0$. Consider
\EQ{\lambda:=\frac{-2\langle\nabla_{\W_t}\CAL L^{\D_t}_\textnormal{rec}(\W_t),\nabla_{\W_t}\diffpd^{\D_t}(\W_t)\rangle}{\|\nabla_{\W_t}\diffpd^{\D_t}(\W_t)\|_2^2}\ge0.}
By the chain rule,
\AL{&\lim_{\eta\to+0}\frac{\diffpd^{\D_t}(\TLD{\W}_t)-\diffpd^{\D_t}\!(\W_t)}{\eta}\\
={}&{-\langle\nabla_{\W_t}\CAL L^{\D_t}_\textnormal{rec}(\W_t)+\lambda\nabla_{\W_t}\CAL L^{\D_t}_\textnormal{fair}(\W_t),\nabla_{\W_t}\diffpd^{\D_t}(\W_t)\rangle}\\
={}&(1-2\sigma(\diffpd^{\D_t}(\W_t)))(-\langle\nabla_{\W_t}\CAL L^{\D_t}_\textnormal{rec}(\W_t),\nabla_{\W_t}\diffpd^{\D_t}(\W_t)\rangle)
.}
The conclusion follows from the fact that
\EQ{\OP{sgn}(x)(1-2\sigma(x))\le0,\qquad\forall x\in\BB R.\qedhere}
\end{proof}

\newpage
\section{Experiments}\label{asec:experimental_results}

% \begin{itemize}[leftmargin=*]
%     \item %Due to space constraints, all content for \S B.1-B.5 is available in the 
%     Online Appendix: \href{https://sites.google.com/view/fade-www24/home}{https://sites.google.com/view/fade-www24/home}.
%     % Alternatively, you can also refer to the arXiv version of the paper.
%     \item GitHub repository: \href{https://github.com/hsyoo32/fade}{https://github.com/hsyoo32/fade}
%     \item Official code DOI: \href{https://doi.org/10.5281/zenodo.10669096}{https://doi.org/10.5281/zenodo.10669096}.
% \end{itemize}

\subsection{Implementation Details of Competitors}\label{assec:implentation}
% Due to space constraints, the content for \S B.1 is available in the Online Appendix: \href{https://sites.google.com/view/fade-www24/b-1}{https://sites.google.com/view/fade-www24/b-1}.

For \adver, the adversarial coefficient $\gamma$ is selected from the suggested range [1, 10, 20, 50], as mentioned in their paper~\cite{li2021towards}. The filter modules are two-layer neural networks with the LeakyReLU activation. The discriminators are multi-layer perceptrons with 7 layers, LeakyReLU activation function, and a dropout rate of 0.3. The discriminators are trained for 10 steps.

In the original paper of \rerank~\cite{li2021user}, they use a re-ranking technique under a fairness-constraint based on the test positive data, which does not align with our assumption that we cannot access future data when serving the recommendation list. Thus, we adopt this method by designating items with predicted scores above a certain threshold as ground-truth items. In our experiments, the predicted scores are normalized to the range of 0 to 1, and we set the threshold to 0.7.

\subsection{Software and Hardware Configuration.}
% Due to space constraints, the content for \S B.2 is available in the Online Appendix: \href{https://sites.google.com/view/fade-www24/b-2}{https://sites.google.com/view/fade-www24/b-2}.
% \noindent\textbf{E -- Software and Hardware configuration.}
All codes are programmed in Python 3.6.9 and PyTorch 1.4.0. All experiments
are performed on a Linux server with 2 Intel Xeon Gold 6240R
CPUs and 1 Nvidia Tesla V100 SXM2 GPU with 32 GB GPU memory.

\subsection{Additional Effectiveness Results}\label{assec:main_result}
% Due to space constraints, the content for \S B.3 is available in the Online Appendix: \href{https://sites.google.com/view/fade-www24/b-3}{https://sites.google.com/view/fade-www24/b-3}.
\begin{itemize}
    \item Fig.~\ref{appenfig:trunc-tradeoff-taskN} show the results for the trade-off between recommendation performance and absolute performance disparity in Task-N. The results for Task-R is in the main body.
    \item Fig.~\ref{appenfig:trunc-trend-taskN} shows the results for the trend of performance disparity in Task-N. The results for Task-R is in the main body.
    \item Fig.~\ref{appenfig:trunc-trend-perf-taskN} and Fig.~\ref{appenfig:trunc-trend-perf-taskR} show the trend of recommendation performance in Task-R and Task-N, respectively.
    \item Fig.~\ref{appenfig:full-trend-perf-taskR} displays the trend in recommendation performance of \pretrain, \full, \fine, \fullfair, and \ours\ in Task-R on Movielenz. This includes the results immediately after pretraining (i.e., $t=0$) and subsequent time periods (i.e., $t=7, 8, 9$). Notably, the results for MF demonstrate that fine-tuning-based methods outperform retraining-based methods in the earlier time periods because fine-tuning is less affected by distribution shifts. However, in later periods, the performance of fine-tuned models eventually degrades, falling even below that of retrained models due to accumulated learning errors. These observations are consistent with our theoretical analyses in \S\ref{ssec:theory} and suggest that the best practice involves incremental fine-tuning with restart.
\end{itemize}

\subsection{Comparison of Soft Ranking Methods}\label{assec:soft_ranking}
Fig.~\ref{appenfig:trunc-softrank-taskR} presents \ours\ adapting different soft ranking metrics, including ApproxNDCG~\cite{qin2010general} and NeuralNDCG~\cite{pobrotyn2021neuralndcg}, as well as \ours\ incorporating the differentiable Hit in Task-R. The legend also provides the average running time for each method.

First, \ours\ outperforms or matches the NeuralNDCG variant in both recommendation performance and performance disparity, while being approximately four times faster. This is because the differentiable Hit addresses NeuralNDCG's gradient vanishing issue by eliminating several processes, including the sinkhorn algorithm.

In comparison to ApproxNDCG, \ours\ generally achieves smaller performance disparity. Although ApproxNDCG may yield lower disparity in some cases, it excessively sacrifices recommendation quality, which is undesirable.

\subsection{Hyperparameter Analysis}\label{assec:hyperparam}
\subsubsection{Effect of the scaling parameter $\lambda$ for the fairness Loss.}
Fig.~\ref{appenfig:trunc-param-lambda-taskR} and Fig.~\ref{appenfig:trunc-param-lambda-taskN} show the effect of the scaling parameter $\lambda$ on the recommendation performances of the advantaged and disadvantaged groups in Task-R and Task-N, respectively.

\subsubsection{Effect of the number of dynamic update epochs.}
Fig.~\ref{appenfig:trunc-param-tepoch-taskR} and Fig.~\ref{appenfig:trunc-param-tepoch-taskN} show the effect of the number of dynamic update epochs on the recommendation performances of the advantaged and disadvantaged groups in Task-R and Task-N, respectively.

\subsubsection{Effect of temperature parameter $\tau$ in the relaxed permutation matrix.}
Fig.~\ref{appenfig:trunc-param-tau-taskR} and Fig.~\ref{appenfig:trunc-param-tau-taskN} show the effect of the temperature parameter $\tau$ on the recommendation performances of the advantaged and disadvantaged groups in Task-R and Task-N, respectively.

\subsubsection{Effect of the number of negative items $\mu$.}
Fig.~\ref{appenfig:trunc-param-numneg-taskR} and Fig.~\ref{appenfig:trunc-param-numneg-taskN} show the effect of the number of negative candidate items $\mu$ for a user in our fairness loss on the recommendation performances of the advantaged and disadvantaged groups in Task-R and Task-N, respectively.

In general, the results suggest that \ours\ performance remains relatively stable when varying the number of negative items in most cases. Thus, setting $\mu$ to 4 results in comparable performance while also enhancing the execution time of \ours.

\input{appenfig/trunc-tradeoff-taskN}
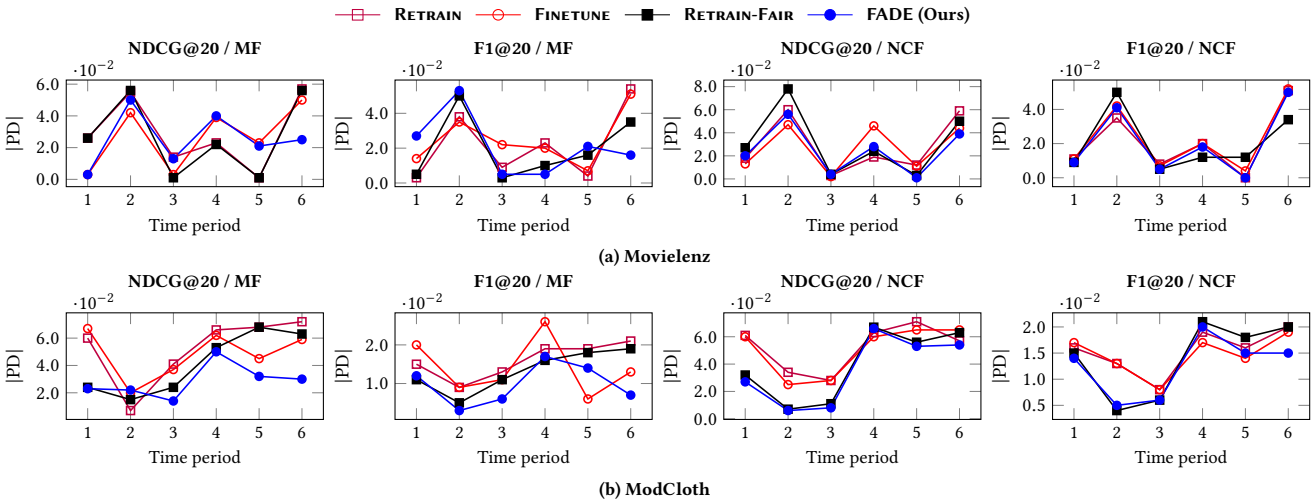
\begin{figure*}[t]
\footnotesize
\centering
\begin{tikzpicture}
\begin{customlegend}[legend columns=5,legend style={align=left,draw=none,column sep=1ex},
        legend entries={\textbf{\full}, \textbf{\fine}, \textbf{\fullfair}, \textbf{\ours\ (Ours)}}]

        \addlegendimage{draw=purple,mark=square}
        \addlegendimage{draw=red,mark=o}
        \addlegendimage{draw=black,mark=square*}   
        \addlegendimage{draw=blue,color=blue,mark=*} 
        \end{customlegend}
\end{tikzpicture}\\
    % \textbf{(a) Task-R}\vspace{-1mm}
\begin{tikzpicture}
\begin{axis}[
height=3.0cm, width=5cm,
title={\textbf{NDCG@20 / MF}},
xtick={1,2,3,4,5,6,7,8,9,10},
xticklabels={1,2,3,4,5,6,7,8,9,10},
ylabel=$|{\pd}|$,xlabel=Time period,
y tick label style={/pgf/number format/.cd,fixed,fixed zerofill,precision=1,/tikz/.cd},]

\addplot[color=purple,mark=square,mark size=1.5pt,line width=0.6pt]
coordinates {(1,0.026) (2,0.055) (3,0.014) (4,0.023) (5,0.001) (6,0.057) };
\addplot[color=red,mark=o,mark size=1.5pt,line width=0.6pt]
coordinates {(1,0.003) (2,0.042) (3,0.003) (4,0.039) (5,0.023) (6,0.050) };
\addplot[color=black,mark=square*,mark size=1.5pt,line width=0.6pt]
coordinates {(1,0.026) (2,0.056) (3,0.001) (4,0.022) (5,0.001) (6,0.056) };
\addplot[color=blue,mark=*,mark size=1.5pt,line width=0.6pt]
coordinates {(1,0.003) (2,0.050) (3,0.013) (4,0.040) (5,0.021) (6,0.025) };

\end{axis}
\end{tikzpicture}
\begin{tikzpicture}
\begin{axis}[
height=3.0cm, width=5cm,
title={\textbf{F1@20 / MF}},
xtick={1,2,3,4,5,6,7,8,9,10},
xticklabels={1,2,3,4,5,6,7,8,9,10},
ylabel=$|{\pd}|$,xlabel=Time period,
y tick label style={/pgf/number format/.cd,fixed,fixed zerofill,precision=1,/tikz/.cd},]
\addplot[color=purple,mark=square,mark size=1.5pt,line width=0.6pt]
coordinates {(1,0.003) (2,0.038) (3,0.009) (4,0.023) (5,0.004) (6,0.054) };
\addplot[color=red,mark=o,mark size=1.5pt,line width=0.6pt]
coordinates {(1,0.014) (2,0.035) (3,0.022) (4,0.020) (5,0.007) (6,0.051) };
\addplot[color=black,mark=square*,mark size=1.5pt,line width=0.6pt]
coordinates {(1,0.005) (2,0.050) (3,0.003) (4,0.010) (5,0.016) (6,0.035) };
\addplot[color=blue,mark=*,mark size=1.5pt,line width=0.6pt]
coordinates {(1,0.027) (2,0.053) (3,0.005) (4,0.005) (5,0.021) (6,0.016) };

\end{axis}
\end{tikzpicture}
\begin{tikzpicture}
\begin{axis}[
height=3.0cm, width=5cm,
title={\textbf{NDCG@20 / NCF}},
xtick={1,2,3,4,5,6,7,8,9,10},
xticklabels={1,2,3,4,5,6,7,8,9,10},
ylabel=$|{\pd}|$,xlabel=Time period,
y tick label style={/pgf/number format/.cd,fixed,fixed zerofill,precision=1,/tikz/.cd},]

\addplot[color=purple,mark=square,mark size=1.5pt,line width=0.6pt]
coordinates {(1,0.018) (2,0.060) (3,0.003) (4,0.019) (5,0.012) (6,0.059) };
\addplot[color=red,mark=o,mark size=1.5pt,line width=0.6pt]
coordinates {(1,0.013) (2,0.047) (3,0.002) (4,0.046) (5,0.011) (6,0.044) };
\addplot[color=black,mark=square*,mark size=1.5pt,line width=0.6pt]
coordinates {(1,0.027) (2,0.078) (3,0.004) (4,0.024) (5,0.003) (6,0.050) };
\addplot[color=blue,mark=*,mark size=1.5pt,line width=0.6pt]
coordinates {(1,0.020) (2,0.056) (3,0.004) (4,0.028) (5,0.001) (6,0.039) };

\end{axis}
\end{tikzpicture}
\begin{tikzpicture}
\begin{axis}[
height=3.0cm, width=5cm,
title={\textbf{F1@20 / NCF}},
xtick={1,2,3,4,5,6,7,8,9,10},
xticklabels={1,2,3,4,5,6,7,8,9,10},
ylabel=$|{\pd}|$,xlabel=Time period,
y tick label style={/pgf/number format/.cd,fixed,fixed zerofill,precision=1,/tikz/.cd},]

\addplot[color=purple,mark=square,mark size=1.5pt,line width=0.6pt]
coordinates {(1,0.011) (2,0.035) (3,0.008) (4,0.020) (5,0.000) (6,0.051) };
\addplot[color=red,mark=o,mark size=1.5pt,line width=0.6pt]
coordinates {(1,0.011) (2,0.042) (3,0.007) (4,0.020) (5,0.004) (6,0.052) };
\addplot[color=black,mark=square*,mark size=1.5pt,line width=0.6pt]
coordinates {(1,0.009) (2,0.050) (3,0.005) (4,0.012) (5,0.012) (6,0.034) };
\addplot[color=blue,mark=*,mark size=1.5pt,line width=0.6pt]
coordinates {(1,0.009) (2,0.041) (3,0.005) (4,0.018) (5,0.000) (6,0.050) };

\end{axis}
\end{tikzpicture}
\\
\textbf{(a) Movielenz}\\
\begin{tikzpicture}
\begin{axis}[
height=3.0cm, width=5cm,
title={\textbf{NDCG@20 / MF}},
xtick={1,2,3,4,5,6,7,8,9,10},
xticklabels={1,2,3,4,5,6,7,8,9,10},
ylabel=$|{\pd}|$,xlabel=Time period,
y tick label style={/pgf/number format/.cd,fixed,fixed zerofill,precision=1,/tikz/.cd},]

\addplot[color=purple,mark=square,mark size=1.5pt,line width=0.6pt]
coordinates {(1,0.060) (2,0.007) (3,0.041) (4,0.066) (5,0.068) (6,0.072) };
\addplot[color=red,mark=o,mark size=1.5pt,line width=0.6pt]
coordinates {(1,0.067) (2,0.020) (3,0.037) (4,0.062) (5,0.045) (6,0.059) };
\addplot[color=black,mark=square*,mark size=1.5pt,line width=0.6pt]
coordinates {(1,0.024) (2,0.015) (3,0.024) (4,0.053) (5,0.068) (6,0.063) };
\addplot[color=blue,mark=*,mark size=1.5pt,line width=0.6pt]
coordinates {(1,0.023) (2,0.022) (3,0.014) (4,0.050) (5,0.032) (6,0.030) };

\end{axis}
\end{tikzpicture}
\begin{tikzpicture}
\begin{axis}[
height=3.0cm, width=5cm,
title={\textbf{F1@20 / MF}},
xtick={1,2,3,4,5,6,7,8,9,10},
xticklabels={1,2,3,4,5,6,7,8,9,10},
ylabel=$|{\pd}|$,xlabel=Time period,
y tick label style={/pgf/number format/.cd,fixed,fixed zerofill,precision=1,/tikz/.cd},]
\addplot[color=purple,mark=square,mark size=1.5pt,line width=0.6pt]
coordinates {(1,0.015) (2,0.009) (3,0.013) (4,0.019) (5,0.019) (6,0.021) };
\addplot[color=red,mark=o,mark size=1.5pt,line width=0.6pt]
coordinates {(1,0.020) (2,0.009) (3,0.011) (4,0.026) (5,0.006) (6,0.013) };
\addplot[color=black,mark=square*,mark size=1.5pt,line width=0.6pt]
coordinates {(1,0.011) (2,0.005) (3,0.011) (4,0.016) (5,0.018) (6,0.019) };
\addplot[color=blue,mark=*,mark size=1.5pt,line width=0.6pt]
coordinates {(1,0.012) (2,0.003) (3,0.006) (4,0.017) (5,0.014) (6,0.007) };

\end{axis}
\end{tikzpicture}
\begin{tikzpicture}
\begin{axis}[
height=3.0cm, width=5cm,
title={\textbf{NDCG@20 / NCF}},
xtick={1,2,3,4,5,6,7,8,9,10},
xticklabels={1,2,3,4,5,6,7,8,9,10},
ylabel=$|{\pd}|$,xlabel=Time period,
y tick label style={/pgf/number format/.cd,fixed,fixed zerofill,precision=1,/tikz/.cd},]

\addplot[color=purple,mark=square,mark size=1.5pt,line width=0.6pt]
coordinates {(1,0.061) (2,0.034) (3,0.028) (4,0.063) (5,0.071) (6,0.057) };
\addplot[color=red,mark=o,mark size=1.5pt,line width=0.6pt]
coordinates {(1,0.060) (2,0.025) (3,0.028) (4,0.060) (5,0.065) (6,0.065) };
\addplot[color=black,mark=square*,mark size=1.5pt,line width=0.6pt]
coordinates {(1,0.032) (2,0.007) (3,0.011) (4,0.067) (5,0.056) (6,0.063) };
\addplot[color=blue,mark=*,mark size=1.5pt,line width=0.6pt]
coordinates {(1,0.027) (2,0.006) (3,0.008) (4,0.066) (5,0.053) (6,0.054) };

\end{axis}
\end{tikzpicture}
\begin{tikzpicture}
\begin{axis}[
height=3.0cm, width=5cm,
title={\textbf{F1@20 / NCF}},
xtick={1,2,3,4,5,6,7,8,9,10},
xticklabels={1,2,3,4,5,6,7,8,9,10},
ylabel=$|{\pd}|$,xlabel=Time period,
y tick label style={/pgf/number format/.cd,fixed,fixed zerofill,precision=1,/tikz/.cd},]

\addplot[color=purple,mark=square,mark size=1.5pt,line width=0.6pt]
coordinates {(1,0.016) (2,0.013) (3,0.008) (4,0.019) (5,0.016) (6,0.020) };
\addplot[color=red,mark=o,mark size=1.5pt,line width=0.6pt]
coordinates {(1,0.017) (2,0.013) (3,0.008) (4,0.017) (5,0.014) (6,0.019) };
\addplot[color=black,mark=square*,mark size=1.5pt,line width=0.6pt]
coordinates {(1,0.015) (2,0.004) (3,0.006) (4,0.021) (5,0.018) (6,0.020) };
\addplot[color=blue,mark=*,mark size=1.5pt,line width=0.6pt]
coordinates {(1,0.014) (2,0.005) (3,0.006) (4,0.020) (5,0.015) (6,0.015) };

\end{axis}
\end{tikzpicture}
\\
\textbf{(b) ModCloth}\\
\caption{Trend of absolute performance disparity in Task-N.
}\label{appenfig:trunc-trend-taskN}
\vspace{-1mm}
\end{figure*}
\input{appenfig/trunc-trend-perf-taskR}
\input{appenfig/trunc-trend-perf-taskN}
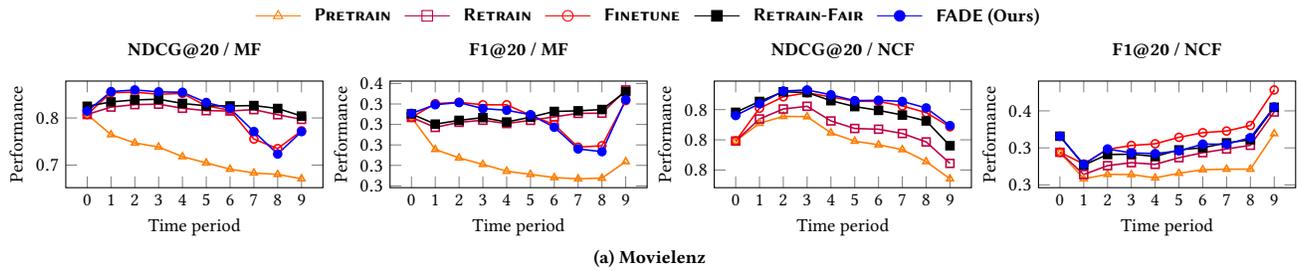
\begin{figure*}[t]
\footnotesize
\centering
\begin{tikzpicture}
\begin{customlegend}[legend columns=6,legend style={align=left,draw=none,column sep=1ex},
        legend entries={\textbf{\pretrain}, \textbf{\full}, \textbf{\fine}, %\textbf{\pretrain-fair}, 
        \textbf{\fullfair}, \textbf{\ours\ (Ours)}}]
        \addlegendimage{draw=cc, mark=triangle}
        \addlegendimage{draw=purple,mark=square}
        \addlegendimage{draw=red,mark=o}
        % \addlegendimage{draw=violet,color=violet,mark=triangle*} 
        \addlegendimage{draw=black,mark=square*}   
        \addlegendimage{draw=blue,color=blue,mark=*} 
        \end{customlegend}
\end{tikzpicture}\\
    % \textbf{(a) Task-R}\vspace{-1mm}
\begin{tikzpicture}
\begin{axis}[
height=3.0cm, width=5cm,
title={\textbf{NDCG@20 / MF}},
xtick={1,2,3,4,5,6,7,8,9,10},
xticklabels={0,1,2,3,4,5,6,7,8,9},
ylabel=Performance,xlabel=Time period,
y tick label style={/pgf/number format/.cd,fixed,fixed zerofill,precision=1,/tikz/.cd},]

\addplot[color=cc,mark=triangle,mark size=1.5pt,line width=0.6pt]
coordinates {(1,0.8074) (2,0.7646) (3,0.747) (4,0.7384) (5,0.718) (6,0.7047) (7,0.6914) (8,0.6833) (9,0.6803) (10,0.671) };
\addplot[color=purple,mark=square,mark size=1.5pt,line width=0.6pt]
coordinates {(1,0.8074) (2,0.8231) (3,0.8283) (4,0.8293) (5,0.8206) (6,0.8155) (7,0.8148) (8,0.8177) (9,0.8069) (10,0.797) };
\addplot[color=red,mark=o,mark size=1.5pt,line width=0.6pt]
coordinates {(1,0.8074) (2,0.8537) (3,0.8547) (4,0.8503) (5,0.8521) (6,0.8255) (7,0.8166) (8,0.7553) (9,0.735) (10,0.7728) };
\addplot[color=black,mark=square*,mark size=1.5pt,line width=0.6pt]
coordinates {(1,0.8248) (2,0.834) (3,0.8383) (4,0.8397) (5,0.8308) (6,0.8258) (7,0.8255) (8,0.8263) (9,0.8201) (10,0.804) };
\addplot[color=blue,mark=*,mark size=1.5pt,line width=0.6pt]
coordinates {(1,0.8145) (2,0.8559) (3,0.8594) (4,0.8554) (5,0.8544) (6,0.8328) (7,0.8201) (8,0.7712) (9,0.7235) (10,0.7711) };

\end{axis}
\end{tikzpicture}
\begin{tikzpicture}
\begin{axis}[
height=3.0cm, width=5cm,
title={\textbf{F1@20 / MF}},
xtick={1,2,3,4,5,6,7,8,9,10},
xticklabels={0,1,2,3,4,5,6,7,8,9},
ylabel=Performance,xlabel=Time period,
y tick label style={/pgf/number format/.cd,fixed,fixed zerofill,precision=1,/tikz/.cd},]
\addplot[color=cc,mark=triangle,mark size=1.5pt,line width=0.6pt]
coordinates {(1,0.3269) (2,0.2956) (3,0.2873) (4,0.281) (5,0.2743) (6,0.2712) (7,0.2681) (8,0.2669) (9,0.2675) (10,0.2837) };
\addplot[color=purple,mark=square,mark size=1.5pt,line width=0.6pt]
coordinates {(1,0.3269) (2,0.3171) (3,0.3221) (4,0.3241) (5,0.3209) (6,0.324) (7,0.3274) (8,0.3308) (9,0.331) (10,0.3537) };
\addplot[color=red,mark=o,mark size=1.5pt,line width=0.6pt]
coordinates {(1,0.3269) (2,0.3403) (3,0.3417) (4,0.3392) (5,0.3392) (6,0.329) (7,0.3199) (8,0.2978) (9,0.2992) (10,0.3433) };
\addplot[color=black,mark=square*,mark size=1.5pt,line width=0.6pt]
coordinates {(1,0.3303) (2,0.3203) (3,0.324) (4,0.3267) (5,0.3224) (6,0.3271) (7,0.3327) (8,0.3333) (9,0.3346) (10,0.3521) };
\addplot[color=blue,mark=*,mark size=1.5pt,line width=0.6pt]
coordinates {(1,0.3308) (2,0.3393) (3,0.3414) (4,0.3355) (5,0.3339) (6,0.3295) (7,0.3173) (8,0.2961) (9,0.2934) (10,0.3442) };

\end{axis}
\end{tikzpicture}
\begin{tikzpicture}
\begin{axis}[
height=3.0cm, width=5cm,
title={\textbf{NDCG@20 / NCF}},
xtick={1,2,3,4,5,6,7,8,9,10},
xticklabels={0,1,2,3,4,5,6,7,8,9},
ylabel=Performance,xlabel=Time period,
y tick label style={/pgf/number format/.cd,fixed,fixed zerofill,precision=1,/tikz/.cd},]

\addplot[color=cc,mark=triangle,mark size=1.5pt,line width=0.6pt]
coordinates {(1,0.8192) (2,0.8309) (3,0.8355) (4,0.8353) (5,0.8246) (6,0.8191) (7,0.8167) (8,0.8135) (9,0.8056) (10,0.7942) };
\addplot[color=purple,mark=square,mark size=1.5pt,line width=0.6pt]
coordinates {(1,0.8192) (2,0.8338) (3,0.8404) (4,0.8422) (5,0.8325) (6,0.8275) (7,0.827) (8,0.8242) (9,0.8186) (10,0.8044) };
\addplot[color=red,mark=o,mark size=1.5pt,line width=0.6pt]
coordinates {(1,0.8192) (2,0.8409) (3,0.8485) (4,0.8511) (5,0.8484) (6,0.8456) (7,0.8454) (8,0.8424) (9,0.838) (10,0.8286) };
\addplot[color=black,mark=square*,mark size=1.5pt,line width=0.6pt]
coordinates {(1,0.8382) (2,0.8454) (3,0.8518) (4,0.8513) (5,0.8459) (6,0.842) (7,0.8395) (8,0.8364) (9,0.8325) (10,0.8161) };
\addplot[color=blue,mark=*,mark size=1.5pt,line width=0.6pt]
coordinates {(1,0.8361) (2,0.8439) (3,0.8523) (4,0.853) (5,0.8498) (6,0.8459) (7,0.8462) (8,0.8454) (9,0.8411) (10,0.8294) };

\end{axis}
\end{tikzpicture}
\begin{tikzpicture}
\begin{axis}[
height=3.0cm, width=5cm,
title={\textbf{F1@20 / NCF}},
xtick={1,2,3,4,5,6,7,8,9,10},
xticklabels={0,1,2,3,4,5,6,7,8,9},
ylabel=Performance,xlabel=Time period,
y tick label style={/pgf/number format/.cd,fixed,fixed zerofill,precision=1,/tikz/.cd},]

\addplot[color=cc,mark=triangle,mark size=1.5pt,line width=0.6pt]
coordinates {(1,0.3376) (2,0.3233) (3,0.3257) (4,0.3256) (5,0.3238) (6,0.3263) (7,0.3282) (8,0.3285) (9,0.3285) (10,0.3477) };
\addplot[color=purple,mark=square,mark size=1.5pt,line width=0.6pt]
coordinates {(1,0.3376) (2,0.3256) (3,0.3304) (4,0.332) (5,0.3311) (6,0.3347) (7,0.3374) (8,0.3394) (9,0.3414) (10,0.3594) };
\addplot[color=red,mark=o,mark size=1.5pt,line width=0.6pt]
coordinates {(1,0.3376) (2,0.3313) (3,0.3392) (4,0.3413) (5,0.3422) (6,0.3457) (7,0.3482) (8,0.3491) (9,0.3521) (10,0.3713) };
\addplot[color=black,mark=square*,mark size=1.5pt,line width=0.6pt]
coordinates {(1,0.3463) (2,0.3307) (3,0.3364) (4,0.3366) (5,0.3354) (6,0.3389) (7,0.3401) (8,0.3427) (9,0.3444) (10,0.3619) };
\addplot[color=blue,mark=*,mark size=1.5pt,line width=0.6pt]
coordinates {(1,0.3463) (2,0.3311) (3,0.3394) (4,0.3373) (5,0.3369) (6,0.3384) (7,0.342) (8,0.342) (9,0.3455) (10,0.3621) };

\end{axis}
\end{tikzpicture}
\\
\textbf{(a) Movielenz}\\
\caption{Trend of recommendation performance in Task-R on Movielenz, including subsequent time periods.
}\label{appenfig:full-trend-perf-taskR}
\vspace{-1mm}
\end{figure*}

\begin{figure*}[t]
\footnotesize
\centering
\begin{tikzpicture}
\begin{customlegend}[legend columns=7,legend style={align=left,draw=none,column sep=1ex},
        legend entries={
        %\textbf{MF}\text{ (avg. runtime: 2.64s)  }, 
        %\textbf{ListNet}\text{ (avg. runtime: 2.82s)  }, 
        \textbf{ApproxNDCG}\text{ (avg. runtime: 3.13s)  }, \textbf{NeuralNDCG}\text{ (avg. runtime: 12.15)  },  \textbf{\ours\ (Ours)}\text{ (avg. runtime: 3.05s)}}]
        % \addlegendimage{draw=aa, mark=diamond, only marks}
        % \addlegendimage{draw=cc, mark=triangle, only marks}
        % \addlegendimage{draw=purple,mark=square, only marks}
        % \addlegendimage{draw=red,mark=o, only marks}
        \addlegendimage{draw=violet,color=violet,mark=triangle*, only marks} 
        \addlegendimage{draw=black,mark=square*, only marks}   
        \addlegendimage{draw=blue,color=blue,mark=*, only marks} 
        %\addlegendimage{draw=aa, mark=diamond}
        \end{customlegend}
\end{tikzpicture}
\\
% \textbf{(a) MF}  \hspace{7cm} \textbf{(b) NCF} \\
\begin{tikzpicture}
\begin{axis}[
height=3.3cm, width=3.1cm,
title={\textbf{MF}},
ylabel=$|{\pd}|$, 
xlabel=NDCG@20,
y tick label style={/pgf/number format/.cd,fixed,fixed zerofill,precision=1,/tikz/.cd},
x tick label style={align=center, /pgf/number format/.cd,fixed,fixed zerofill,precision=2,/tikz/.cd},]
% \addplot[only marks, color=purple,mark=square,]
% coordinates {(0.845,0.004) };
% \addplot[only marks, color=red,mark=o,]
% coordinates {(0.845,0.007) };
\addplot[only marks, color=violet,mark=triangle*,]
coordinates {(0.840,0.007) };
\addplot[only marks, color=black,mark=square*,]
coordinates {(0.843,0.004) };
\addplot[only marks, color=blue,mark=*,]
coordinates {(0.846,0.004) };
\end{axis}
\end{tikzpicture}
\hspace{-2mm}
\begin{tikzpicture}
\begin{axis}[
height=3.3cm, width=3.1cm,
title={\textbf{MF}},
xlabel=F1@20,
y tick label style={/pgf/number format/.cd,fixed,fixed zerofill,precision=1,/tikz/.cd},
x tick label style={align=center, /pgf/number format/.cd,fixed,fixed zerofill,precision=2,/tikz/.cd},]
% \addplot[only marks, color=purple,mark=square,]
% coordinates {(0.335,0.008) };
% \addplot[only marks, color=red,mark=o,]
% coordinates {(0.331,0.003) };
\addplot[only marks, color=violet,mark=triangle*,]
coordinates {(0.334,0.013) };
\addplot[only marks, color=black,mark=square*,]
coordinates {(0.327,0.006) };
\addplot[only marks, color=blue,mark=*,]
coordinates {(0.330,0.002) };
\end{axis}
\end{tikzpicture}
\hspace{-2mm}
\begin{tikzpicture}
\begin{axis}[
height=3.3cm, width=3.1cm,
title={\textbf{NCF}},
xlabel=NDCG@20,
y tick label style={/pgf/number format/.cd,fixed,fixed zerofill,precision=1,/tikz/.cd},
x tick label style={align=center, /pgf/number format/.cd,fixed,fixed zerofill,precision=2,/tikz/.cd},]
% \addplot[only marks, color=purple,mark=square,]
% coordinates {(0.843,0.008) };
% \addplot[only marks, color=red,mark=o,]
% coordinates {(0.848,0.009) };
\addplot[only marks, color=violet,mark=triangle*,]
coordinates {(0.836,0.005) };
\addplot[only marks, color=black,mark=square*,]
coordinates {(0.847,0.007) };
\addplot[only marks, color=blue,mark=*,]
coordinates {(0.848,0.006) };
\end{axis}
\end{tikzpicture}
\hspace{-2mm}
\begin{tikzpicture}
\begin{axis}[
height=3.3cm, width=3.1cm,
title={\textbf{NCF}},
xlabel=F1@20,
y tick label style={/pgf/number format/.cd,fixed,fixed zerofill,precision=1,/tikz/.cd},
x tick label style={align=center, /pgf/number format/.cd,fixed,fixed zerofill,precision=2,/tikz/.cd},]
% \addplot[only marks, color=purple,mark=square,]
% coordinates {(0.329,0.006) };
% \addplot[only marks, color=red,mark=o,]
% coordinates {(0.341,0.008) };
\addplot[only marks, color=violet,mark=triangle*,]
coordinates {(0.331,0.006) };
\addplot[only marks, color=black,mark=square*,]
coordinates {(0.340,0.004) };
\addplot[only marks, color=blue,mark=*,]
coordinates {(0.339,0.004) };
\end{axis}
\end{tikzpicture}
\hspace{-2mm}
\begin{tikzpicture}
\begin{axis}[
height=3.3cm, width=3.1cm,
title={\textbf{MF}},
xlabel=NDCG@20,
y tick label style={/pgf/number format/.cd,fixed,fixed zerofill,precision=1,/tikz/.cd},
x tick label style={align=center, /pgf/number format/.cd,fixed,fixed zerofill,precision=2,/tikz/.cd},]
% \addplot[only marks, color=purple,mark=square,]
% coordinates {(0.267,0.043) };
% \addplot[only marks, color=red,mark=o,]
% coordinates {(0.275,0.059) };
\addplot[only marks, color=violet,mark=triangle*,]
coordinates {(0.236,0.017) };
\addplot[only marks, color=black,mark=square*,]
coordinates {(0.269,0.057) };
\addplot[only marks, color=blue,mark=*,]
coordinates {(0.270,0.047) };
\end{axis}
\end{tikzpicture}
\hspace{-2mm}
\begin{tikzpicture}
\begin{axis}[
height=3.3cm, width=3.1cm,
title={\textbf{MF}},
xlabel=F1@20,
xmin=0.068, xmax=0.1,
y tick label style={/pgf/number format/.cd,fixed,fixed zerofill,precision=1,/tikz/.cd},
x tick label style={align=center, /pgf/number format/.cd,fixed,fixed zerofill,precision=2,/tikz/.cd},]
% \addplot[only marks, color=purple,mark=square,]
% coordinates {(0.085,0.018) };
% \addplot[only marks, color=red,mark=o,]
% coordinates {(0.088,0.022) };
\addplot[only marks, color=violet,mark=triangle*,]
coordinates {(0.070,0.006) };
\addplot[only marks, color=black,mark=square*,]
coordinates {(0.084,0.020) };
\addplot[only marks, color=blue,mark=*,]
coordinates {(0.087,0.020) };
\end{axis}
\end{tikzpicture}
\hspace{-2mm}
\begin{tikzpicture}
\begin{axis}[
height=3.3cm, width=3.1cm,
title={\textbf{NCF}},
xlabel=NDCG@20,
xmin=0.260, xmax=0.290,
y tick label style={/pgf/number format/.cd,fixed,fixed zerofill,precision=1,/tikz/.cd},
x tick label style={align=center, /pgf/number format/.cd,fixed,fixed zerofill,precision=2,/tikz/.cd},]
% \addplot[only marks, color=purple,mark=square,]
% coordinates {(0.283,0.062) };
% \addplot[only marks, color=red,mark=o,]
% coordinates {(0.286,0.063) };
\addplot[only marks, color=violet,mark=triangle*,]
coordinates {(0.282,0.054) };
\addplot[only marks, color=black,mark=square*,]
coordinates {(0.285,0.062) };
\addplot[only marks, color=blue,mark=*,]
coordinates {(0.272,0.051) };
\end{axis}
\end{tikzpicture}
\hspace{-2mm}
\begin{tikzpicture}
\begin{axis}[
height=3.3cm, width=3.1cm,
title={\textbf{NCF}},
xlabel=F1@20,
y tick label style={/pgf/number format/.cd,fixed,fixed zerofill,precision=1,/tikz/.cd},
x tick label style={align=center, /pgf/number format/.cd,fixed,fixed zerofill,precision=2,/tikz/.cd},]
% \addplot[only marks, color=purple,mark=square,]
% coordinates {(0.091,0.027) };
% \addplot[only marks, color=red,mark=o,]
% coordinates {(0.092,0.027) };
\addplot[only marks, color=violet,mark=triangle*,]
coordinates {(0.088,0.023) };
\addplot[only marks, color=black,mark=square*,]
coordinates {(0.092,0.027) };
\addplot[only marks, color=blue,mark=*,]
coordinates {(0.090,0.025) };
\end{axis}
\end{tikzpicture}
\\
\textbf{(a) Movielenz} \hspace{7cm} \textbf{(b) ModCloth}
\caption{Trade-off between recommendation performance and absolute performance disparity in Task-R.}
\label{appenfig:trunc-softrank-taskR}
\vspace{-2mm}
\end{figure*}
\begin{figure*}[t]
\footnotesize
\centering
\begin{tikzpicture}
\begin{customlegend}[legend columns=7,legend style={align=left,draw=none,column sep=1ex},
        legend entries={
        %\textbf{MF}\text{ (avg. runtime: 2.64s)  }, 
        % \textbf{ListNet}\text{ (avg. runtime: 2.82s)  }, 
        \textbf{ApproxNDCG}\text{ (avg. runtime: 3.13s)  }, \textbf{NeuralNDCG}\text{ (avg. runtime: 12.15)  },  \textbf{\ours\ (Ours)}\text{ (avg. runtime: 3.05s)}}]
        % \addlegendimage{draw=aa, mark=diamond, only marks}
        % \addlegendimage{draw=cc, mark=triangle, only marks}
        % \addlegendimage{draw=purple,mark=square, only marks}
        % \addlegendimage{draw=red,mark=o, only marks}
        \addlegendimage{draw=violet,color=violet,mark=triangle*, only marks} 
        \addlegendimage{draw=black,mark=square*, only marks}   
        \addlegendimage{draw=blue,color=blue,mark=*, only marks} 
        %\addlegendimage{draw=aa, mark=diamond}
        \end{customlegend}
\end{tikzpicture}
\\
% \textbf{(a) MF}  \hspace{7cm} \textbf{(b) NCF} \\
\begin{tikzpicture}
\begin{axis}[
height=3.3cm, width=3.1cm,
title={\textbf{MF}},
ylabel=$|{\pd}|$, 
xlabel=NDCG@20,
y tick label style={/pgf/number format/.cd,fixed,fixed zerofill,precision=1,/tikz/.cd},
x tick label style={align=center, /pgf/number format/.cd,fixed,fixed zerofill,precision=2,/tikz/.cd},]
% \addplot[only marks, color=purple,mark=square,]
% coordinates {(0.739,0.027) };
% \addplot[only marks, color=red,mark=o,]
% coordinates {(0.739,0.021) };
\addplot[only marks, color=violet,mark=triangle*,]
coordinates {(0.734,0.025) };
\addplot[only marks, color=black,mark=square*,]
coordinates {(0.741,0.024) };
\addplot[only marks, color=blue,mark=*,]
coordinates {(0.740,0.025) };
\end{axis}
\end{tikzpicture}
\hspace{-2mm}
\begin{tikzpicture}
\begin{axis}[
height=3.3cm, width=3.1cm,
title={\textbf{MF}},
xlabel=F1@20,
y tick label style={/pgf/number format/.cd,fixed,fixed zerofill,precision=1,/tikz/.cd},
x tick label style={align=center, /pgf/number format/.cd,fixed,fixed zerofill,precision=2,/tikz/.cd},]
% \addplot[only marks, color=purple,mark=square,]
% coordinates {(0.292,0.020) };
% \addplot[only marks, color=red,mark=o,]
% coordinates {(0.297,0.022) };
\addplot[only marks, color=violet,mark=triangle*,]
coordinates {(0.290,0.019) };
\addplot[only marks, color=black,mark=square*,]
coordinates {(0.282,0.021) };
\addplot[only marks, color=blue,mark=*,]
coordinates {(0.292,0.021) };
\end{axis}
\end{tikzpicture}
\hspace{-2mm}
\begin{tikzpicture}
\begin{axis}[
height=3.3cm, width=3.1cm,
title={\textbf{NCF}},
xlabel=NDCG@20,
y tick label style={/pgf/number format/.cd,fixed,fixed zerofill,precision=1,/tikz/.cd},
x tick label style={align=center, /pgf/number format/.cd,fixed,fixed zerofill,precision=2,/tikz/.cd},]
% \addplot[only marks, color=purple,mark=square,]
% coordinates {(0.741,0.026) };
% \addplot[only marks, color=red,mark=o,]
% coordinates {(0.739,0.026) };
\addplot[only marks, color=violet,mark=triangle*,]
coordinates {(0.730,0.026) };
\addplot[only marks, color=black,mark=square*,]
coordinates {(0.739,0.026) };
\addplot[only marks, color=blue,mark=*,]
coordinates {(0.740,0.025) };
\end{axis}
\end{tikzpicture}
\hspace{-2mm}
\begin{tikzpicture}
\begin{axis}[
height=3.3cm, width=3.1cm,
title={\textbf{NCF}},
xlabel=F1@20,
y tick label style={/pgf/number format/.cd,fixed,fixed zerofill,precision=1,/tikz/.cd},
x tick label style={align=center, /pgf/number format/.cd,fixed,fixed zerofill,precision=2,/tikz/.cd},]
% \addplot[only marks, color=purple,mark=square,]
% coordinates {(0.293,0.019) };
% \addplot[only marks, color=red,mark=o,]
% coordinates {(0.303,0.019) };
\addplot[only marks, color=violet,mark=triangle*,]
coordinates {(0.302,0.020) };
\addplot[only marks, color=black,mark=square*,]
coordinates {(0.303,0.020) };
\addplot[only marks, color=blue,mark=*,]
coordinates {(0.303,0.020) };
\end{axis}
\end{tikzpicture}
\hspace{-2mm}
\begin{tikzpicture}
\begin{axis}[
height=3.3cm, width=3.1cm,
title={\textbf{MF}},
xlabel=NDCG@20,
y tick label style={/pgf/number format/.cd,fixed,fixed zerofill,precision=1,/tikz/.cd},
x tick label style={align=center, /pgf/number format/.cd,fixed,fixed zerofill,precision=2,/tikz/.cd},]
% \addplot[only marks, color=purple,mark=square,]
% coordinates {(0.175,0.024) };
% \addplot[only marks, color=red,mark=o,]
% coordinates {(0.194,0.034) };
\addplot[only marks, color=violet,mark=triangle*,]
coordinates {(0.161,0.020) };
\addplot[only marks, color=black,mark=square*,]
coordinates {(0.188,0.033) };
\addplot[only marks, color=blue,mark=*,]
coordinates {(0.179,0.028) };
\end{axis}
\end{tikzpicture}
\hspace{-2mm}
\begin{tikzpicture}
\begin{axis}[
height=3.3cm, width=3.1cm,
title={\textbf{MF}},
xlabel=F1@20,
y tick label style={/pgf/number format/.cd,fixed,fixed zerofill,precision=1,/tikz/.cd},
x tick label style={align=center, /pgf/number format/.cd,fixed,fixed zerofill,precision=2,/tikz/.cd},]
% \addplot[only marks, color=purple,mark=square,]
% coordinates {(0.048,0.009) };
% \addplot[only marks, color=red,mark=o,]
% coordinates {(0.054,0.012) };
\addplot[only marks, color=violet,mark=triangle*,]
coordinates {(0.044,0.003) };
\addplot[only marks, color=black,mark=square*,]
coordinates {(0.050,0.011) };
\addplot[only marks, color=blue,mark=*,]
coordinates {(0.050,0.010) };
\end{axis}
\end{tikzpicture}
\hspace{-2mm}
\begin{tikzpicture}
\begin{axis}[
height=3.3cm, width=3.1cm,
title={\textbf{NCF}},
xlabel=NDCG@20,
y tick label style={/pgf/number format/.cd,fixed,fixed zerofill,precision=1,/tikz/.cd},
x tick label style={align=center, /pgf/number format/.cd,fixed,fixed zerofill,precision=2,/tikz/.cd},]
% \addplot[only marks, color=purple,mark=square,]
% coordinates {(0.187,0.042) };
% \addplot[only marks, color=red,mark=o,]
% coordinates {(0.188,0.042) };
\addplot[only marks, color=violet,mark=triangle*,]
coordinates {(0.185,0.033) };
\addplot[only marks, color=black,mark=square*,]
coordinates {(0.186,0.041) };
\addplot[only marks, color=blue,mark=*,]
coordinates {(0.176,0.036) };
\end{axis}
\end{tikzpicture}
\hspace{-2mm}
\begin{tikzpicture}
\begin{axis}[
height=3.3cm, width=3.1cm,
title={\textbf{NCF}},
xlabel=F1@20,
y tick label style={/pgf/number format/.cd,fixed,fixed zerofill,precision=1,/tikz/.cd},
x tick label style={align=center, /pgf/number format/.cd,fixed,fixed zerofill,precision=2,/tikz/.cd},]
% \addplot[only marks, color=purple,mark=square,]
% coordinates {(0.050,0.014) };
% \addplot[only marks, color=red,mark=o,]
% coordinates {(0.050,0.013) };
\addplot[only marks, color=violet,mark=triangle*,]
coordinates {(0.048,0.010) };
\addplot[only marks, color=black,mark=square*,]
coordinates {(0.049,0.013) };
\addplot[only marks, color=blue,mark=*,]
coordinates {(0.048,0.012) };
\end{axis}
\end{tikzpicture}
\\
\textbf{(a) Movielenz} \hspace{7cm} \textbf{(b) ModCloth}
\caption{Trade-off between recommendation performance and absolute performance disparity in Task-N.}
\label{appenfig:trunc-softrank-taskN}
\vspace{-2mm}
\end{figure*}

\begin{figure*}[t]
\footnotesize
\centering
\begin{tikzpicture}
\begin{customlegend}[legend columns=5,legend style={align=left,draw=none,column sep=1ex},
        legend entries={\textbf{Advantaged group}\text{  }, \textbf{Disadvantaged group}}]
        % \addlegendimage{draw=purple,mark=square, only marks}
        % \addlegendimage{draw=red,mark=o, only marks}
        \addlegendimage{draw=black,mark=square*}   
        \addlegendimage{draw=blue,color=blue,mark=*} 
        \end{customlegend}
\end{tikzpicture}\\
    % \textbf{(a) Task-R}\vspace{-1mm}
\begin{tikzpicture}
\begin{axis}[
height=3cm, width=5cm,
title={\textbf{NDCG@20 / MF}},
xtick={1,2,3,4,5,6,7,8,9,10,11},
xticklabels={0,0.1,0.3,0.5,0.8,1.0,1.5,2.0,2.5,3.0,4.0},
ylabel=Performance,xlabel=$\lambda$,
y tick label style={/pgf/number format/.cd,fixed,fixed zerofill,precision=2,/tikz/.cd},]

\addplot[color=black,mark=square*,mark size=1.5pt,line width=0.6pt]
coordinates {(1,0.845) (2,0.845) (3,0.846) (4,0.847) (5,0.843) (6,0.841) (7,0.836) (8,0.825) (9,0.808) (10,0.793)  };
\addplot[color=blue,mark=*,mark size=1.5pt,line width=0.6pt]
coordinates {(1,0.834) (2,0.835) (3,0.839) (4,0.844) (5,0.846) (6,0.848) (7,0.846) (8,0.848) (9,0.851) (10,0.850) };
% \addplot[color=black,mark=square*,mark size=1.5pt,line width=0.6pt]
% coordinates {(1,0.845) (2,0.845) (3,0.846) (4,0.847) (5,0.843) (6,0.841) (7,0.836) (8,0.825) (9,0.808) (10,0.793) (11,0.730) };
% \addplot[color=blue,mark=*,mark size=1.5pt,line width=0.6pt]
% coordinates {(1,0.834) (2,0.835) (3,0.839) (4,0.844) (5,0.846) (6,0.848) (7,0.846) (8,0.848) (9,0.851) (10,0.850) (11,0.856) };

\end{axis}
\end{tikzpicture}
\begin{tikzpicture}
\begin{axis}[
height=3cm, width=5cm,
title={\textbf{F1@20 / MF}},
xtick={1,2,3,4,5,6,7,8,9,10,11},
xticklabels={0,0.1,0.3,0.5,0.8,1.0,1.5,2.0,2.5,3.0,4.0},
ylabel=Performance,xlabel=$\lambda$,
y tick label style={/pgf/number format/.cd,fixed,fixed zerofill,precision=2,/tikz/.cd},]

\addplot[color=black,mark=square*,mark size=1.5pt,line width=0.6pt]
coordinates {(1,0.339) (2,0.339) (3,0.339) (4,0.339) (5,0.336) (6,0.334) (7,0.330) (8,0.324) (9,0.313) (10,0.304)  };
\addplot[color=blue,mark=*,mark size=1.5pt,line width=0.6pt]
coordinates {(1,0.322) (2,0.322) (3,0.325) (4,0.326) (5,0.329) (6,0.330) (7,0.330) (8,0.330) (9,0.332) (10,0.330) };

\end{axis}
\end{tikzpicture}
\begin{tikzpicture}
\begin{axis}[
height=3cm, width=5cm,
title={\textbf{NDCG@20 / NCF}},
xtick={1,2,3,4,5,6,7,8,9,10,11},
xticklabels={0,0.1,0.3,0.5,0.8,1.0,1.5,2.0,2.5,3.0,4.0},
ylabel=Performance,xlabel=$\lambda$,
y tick label style={/pgf/number format/.cd,fixed,fixed zerofill,precision=2,/tikz/.cd},]

\addplot[color=black,mark=square*,mark size=1.5pt,line width=0.6pt]
coordinates {(1,0.850) (2,0.851) (3,0.851) (4,0.851) (5,0.851) (6,0.851) (7,0.850) (8,0.848) (9,0.846) (10,0.844) };
\addplot[color=blue,mark=*,mark size=1.5pt,line width=0.6pt]
coordinates {(1,0.837) (2,0.839) (3,0.840) (4,0.841) (5,0.842) (6,0.843) (7,0.845) (8,0.847) (9,0.850) (10,0.850)  };

\end{axis}
\end{tikzpicture}
\begin{tikzpicture}
\begin{axis}[
height=3cm, width=5cm,
title={\textbf{F1@20 / NCF}},
xtick={1,2,3,4,5,6,7,8,9,10,11},
xticklabels={0,0.1,0.3,0.5,0.8,1.0,1.5,2.0,2.5,3.0,4.0},
ylabel=Performance,xlabel=$\lambda$,
y tick label style={/pgf/number format/.cd,fixed,fixed zerofill,precision=2,/tikz/.cd},]

\addplot[color=black,mark=square*,mark size=1.5pt,line width=0.6pt]
coordinates {(1,0.346) (2,0.346) (3,0.347) (4,0.346) (5,0.346) (6,0.345) (7,0.343) (8,0.341) (9,0.340) (10,0.337)  };
\addplot[color=blue,mark=*,mark size=1.5pt,line width=0.6pt]
coordinates {(1,0.327) (2,0.327) (3,0.329) (4,0.330) (5,0.331) (6,0.332) (7,0.334) (8,0.336) (9,0.337) (10,0.338)};

\end{axis}
\end{tikzpicture}
\\
\textbf{(a) Movielenz}\\
\begin{tikzpicture}
\begin{axis}[
height=3cm, width=5cm,
title={\textbf{NDCG@20 / MF}},
xtick={1,2,3,4,5,6,7,8,9,10,11},
xticklabels={0,0.1,0.3,0.5,0.8,1.0,1.5,2.0,2.5,3.0,4.0},
ylabel=Performance,xlabel=$\lambda$,
y tick label style={/pgf/number format/.cd,fixed,fixed zerofill,precision=2,/tikz/.cd},]

\addplot[color=black,mark=square*,mark size=1.5pt,line width=0.6pt]
coordinates {(1,0.298) (2,0.298) (3,0.296) (4,0.294) (5,0.293) (6,0.290) (7,0.286) (8,0.281) (9,0.279) (10,0.281)  };
\addplot[color=blue,mark=*,mark size=1.5pt,line width=0.6pt]
coordinates {(1,0.225) (2,0.229) (3,0.228) (4,0.229) (5,0.228) (6,0.229) (7,0.226) (8,0.229) (9,0.230) (10,0.232) };

\end{axis}
\end{tikzpicture}
\begin{tikzpicture}
\begin{axis}[
height=3cm, width=5cm,
title={\textbf{F1@20 / MF}},
xtick={1,2,3,4,5,6,7,8,9,10,11},
xticklabels={0,0.1,0.3,0.5,0.8,1.0,1.5,2.0,2.5,3.0,4.0},
ylabel=Performance,xlabel=$\lambda$,
y tick label style={/pgf/number format/.cd,fixed,fixed zerofill,precision=2,/tikz/.cd},]
\addplot[color=black,mark=square*,mark size=1.5pt,line width=0.6pt]
coordinates {(1,0.098) (2,0.098) (3,0.097) (4,0.097) (5,0.096) (6,0.094) (7,0.092) (8,0.090) (9,0.089) (10,0.090) };
\addplot[color=blue,mark=*,mark size=1.5pt,line width=0.6pt]
coordinates {(1,0.069) (2,0.070) (3,0.070) (4,0.070) (5,0.070) (6,0.071) (7,0.069) (8,0.069) (9,0.069) (10,0.069)  };

\end{axis}
\end{tikzpicture}
\begin{tikzpicture}
\begin{axis}[
height=3cm, width=5cm,
title={\textbf{NDCG@20 / NCF}},
xtick={1,2,3,4,5,6,7,8,9,10,11},
xticklabels={0,0.1,0.3,0.5,0.8,1.0,1.5,2.0,2.5,3.0,4.0},
ylabel=Performance,xlabel=$\lambda$,
y tick label style={/pgf/number format/.cd,fixed,fixed zerofill,precision=2,/tikz/.cd},]
\addplot[color=black,mark=square*,mark size=1.5pt,line width=0.6pt]
coordinates {(1,0.297) (2,0.297) (3,0.299) (4,0.300) (5,0.299) (6,0.298) (7,0.295) (8,0.293) (9,0.290) (10,0.287)  };
\addplot[color=blue,mark=*,mark size=1.5pt,line width=0.6pt]
coordinates {(1,0.223) (2,0.225) (3,0.229) (4,0.232) (5,0.235) (6,0.236) (7,0.235) (8,0.233) (9,0.232) (10,0.231) };

\end{axis}
\end{tikzpicture}
\begin{tikzpicture}
\begin{axis}[
height=3cm, width=5cm,
title={\textbf{F1@20 / NCF}},
xtick={1,2,3,4,5,6,7,8,9,10,11},
xticklabels={0,0.1,0.3,0.5,0.8,1.0,1.5,2.0,2.5,3.0,4.0},
ylabel=Performance,xlabel=$\lambda$,
y tick label style={/pgf/number format/.cd,fixed,fixed zerofill,precision=2,/tikz/.cd},]
\addplot[color=black,mark=square*,mark size=1.5pt,line width=0.6pt]
coordinates {(1,0.097) (2,0.097) (3,0.097) (4,0.097) (5,0.097) (6,0.097) (7,0.096) (8,0.096) (9,0.095) (10,0.095) };
\addplot[color=blue,mark=*,mark size=1.5pt,line width=0.6pt]
coordinates {(1,0.069) (2,0.069) (3,0.069) (4,0.069) (5,0.070) (6,0.070) (7,0.070) (8,0.069) (9,0.069) (10,0.069)  };

\end{axis}
\end{tikzpicture}
\\
\textbf{(b) ModCloth}\\
\caption{The effect of the scaling factor $\lambda$ on the recommendation performances of the advantaged and disadvantaged groups in Task-R.
}\label{appenfig:trunc-param-lambda-taskR}
\vspace{-1mm}
\end{figure*}
\begin{figure*}[t]
\footnotesize
\centering
\begin{tikzpicture}
\begin{customlegend}[legend columns=5,legend style={align=left,draw=none,column sep=1ex},
        legend entries={\textbf{Advantaged group}\text{  }, \textbf{Disadvantaged group}}]
        % \addlegendimage{draw=purple,mark=square, only marks}
        % \addlegendimage{draw=red,mark=o, only marks}
        \addlegendimage{draw=black,mark=square*}   
        \addlegendimage{draw=blue,color=blue,mark=*} 
        \end{customlegend}
\end{tikzpicture}\\
    % \textbf{(a) Task-R}\vspace{-1mm}
\begin{tikzpicture}
\begin{axis}[
height=3cm, width=5cm,
title={\textbf{NDCG@20 / MF}},
xtick={1,2,3,4,5,6,7,8,9,10,11},
xticklabels={1,5,10,15,20,25,30,35,40,45,50},
ylabel=Performance,xlabel=The number of epochs,
y tick label style={/pgf/number format/.cd,fixed,fixed zerofill,precision=2,/tikz/.cd},]

\addplot[color=black,mark=square*,mark size=1.5pt,line width=0.6pt]
coordinates {(1,0.839) (2,0.851) (3,0.841) (4,0.795) (5,0.762) (6,0.732) (7,0.723) (8,0.717) (9,0.713) (10,0.705) (11,0.708) };
\addplot[color=blue,mark=*,mark size=1.5pt,line width=0.6pt]
coordinates {(1,0.830) (2,0.851) (3,0.848) (4,0.805) (5,0.764) (6,0.740) (7,0.728) (8,0.710) (9,0.708) (10,0.713) (11,0.705) };

\end{axis}
\end{tikzpicture}
\begin{tikzpicture}
\begin{axis}[
height=3cm, width=5cm,
title={\textbf{F1@20 / MF}},
xtick={1,2,3,4,5,6,7,8,9,10,11},
xticklabels={1,5,10,15,20,25,30,35,40,45,50},
ylabel=Performance,xlabel=The number of epochs,
y tick label style={/pgf/number format/.cd,fixed,fixed zerofill,precision=2,/tikz/.cd},]

\addplot[color=black,mark=square*,mark size=1.5pt,line width=0.6pt]
coordinates {(1,0.328) (2,0.341) (3,0.334) (4,0.306) (5,0.285) (6,0.266) (7,0.258) (8,0.250) (9,0.246) (10,0.242) (11,0.243) };
\addplot[color=blue,mark=*,mark size=1.5pt,line width=0.6pt]
coordinates {(1,0.313) (2,0.332) (3,0.330) (4,0.307) (5,0.279) (6,0.264) (7,0.253) (8,0.243) (9,0.241) (10,0.244) (11,0.240) };

\end{axis}
\end{tikzpicture}
\begin{tikzpicture}
\begin{axis}[
height=3cm, width=5cm,
title={\textbf{NDCG@20 / NCF}},
xtick={1,2,3,4,5,6,7,8,9,10,11},
xticklabels={1,5,10,15,20,25,30,35,40,45,50},
ylabel=Performance,xlabel=The number of epochs,
y tick label style={/pgf/number format/.cd,fixed,fixed zerofill,precision=2,/tikz/.cd},]

\addplot[color=black,mark=square*,mark size=1.5pt,line width=0.6pt]
coordinates {(1,0.841) (2,0.846) (3,0.851) (4,0.854) (5,0.855) (6,0.856) (7,0.857) (8,0.857) (9,0.857) (10,0.857) (11,0.857) };
\addplot[color=blue,mark=*,mark size=1.5pt,line width=0.6pt]
coordinates {(1,0.830) (2,0.837) (3,0.843) (4,0.848) (5,0.851) (6,0.852) (7,0.853) (8,0.854) (9,0.854) (10,0.855) (11,0.854) };

\end{axis}
\end{tikzpicture}
\begin{tikzpicture}
\begin{axis}[
height=3cm, width=5cm,
title={\textbf{F1@20 / NCF}},
xtick={1,2,3,4,5,6,7,8,9,10,11},
xticklabels={1,5,10,15,20,25,30,35,40,45,50},
ylabel=Performance,xlabel=The number of epochs,
y tick label style={/pgf/number format/.cd,fixed,fixed zerofill,precision=2,/tikz/.cd},]

\addplot[color=black,mark=square*,mark size=1.5pt,line width=0.6pt]
coordinates {(1,0.335) (2,0.341) (3,0.345) (4,0.347) (5,0.348) (6,0.348) (7,0.348) (8,0.348) (9,0.348) (10,0.348) (11,0.348) };
\addplot[color=blue,mark=*,mark size=1.5pt,line width=0.6pt]
coordinates {(1,0.321) (2,0.328) (3,0.332) (4,0.334) (5,0.335) (6,0.336) (7,0.336) (8,0.337) (9,0.337) (10,0.337) (11,0.337) };

\end{axis}
\end{tikzpicture}
\\
\textbf{(a) Movielenz}\\
\begin{tikzpicture}
\begin{axis}[
height=3cm, width=5cm,
title={\textbf{NDCG@20 / MF}},
xtick={1,2,3,4,5,6,7,8,9,10,11},
xticklabels={1,5,10,15,20,25,30,35,40,45,50},
ylabel=Performance,xlabel=The number of epochs,
y tick label style={/pgf/number format/.cd,fixed,fixed zerofill,precision=2,/tikz/.cd},]

\addplot[color=black,mark=square*,mark size=1.5pt,line width=0.6pt]
coordinates {(1,0.294) (2,0.289) (3,0.290) (4,0.292) (5,0.288) (6,0.287) (7,0.283) (8,0.280) (9,0.278) (10,0.277) (11,0.276) };
\addplot[color=blue,mark=*,mark size=1.5pt,line width=0.6pt]
coordinates {(1,0.219) (2,0.225) (3,0.229) (4,0.234) (5,0.236) (6,0.237) (7,0.237) (8,0.233) (9,0.233) (10,0.228) (11,0.226) };

\end{axis}
\end{tikzpicture}
\begin{tikzpicture}
\begin{axis}[
height=3cm, width=5cm,
title={\textbf{F1@20 / MF}},
xtick={1,2,3,4,5,6,7,8,9,10,11},
xticklabels={1,5,10,15,20,25,30,35,40,45,50},
ylabel=Performance,xlabel=The number of epochs,
y tick label style={/pgf/number format/.cd,fixed,fixed zerofill,precision=2,/tikz/.cd},]
\addplot[color=black,mark=square*,mark size=1.5pt,line width=0.6pt]
coordinates {(1,0.096) (2,0.096) (3,0.094) (4,0.093) (5,0.091) (6,0.090) (7,0.089) (8,0.087) (9,0.086) (10,0.085) (11,0.084) };
\addplot[color=blue,mark=*,mark size=1.5pt,line width=0.6pt]
coordinates {(1,0.068) (2,0.069) (3,0.071) (4,0.071) (5,0.070) (6,0.069) (7,0.068) (8,0.067) (9,0.067) (10,0.065) (11,0.065) };
\end{axis}
\end{tikzpicture}
\begin{tikzpicture}
\begin{axis}[
height=3cm, width=5cm,
title={\textbf{NDCG@20 / NCF}},
xtick={1,2,3,4,5,6,7,8,9,10,11},
xticklabels={1,5,10,15,20,25,30,35,40,45,50},
ylabel=Performance,xlabel=The number of epochs,
y tick label style={/pgf/number format/.cd,fixed,fixed zerofill,precision=2,/tikz/.cd},]
\addplot[color=black,mark=square*,mark size=1.5pt,line width=0.6pt]
coordinates {(1,0.297) (2,0.297) (3,0.298) (4,0.299) (5,0.299) (6,0.299) (7,0.300) (8,0.300) (9,0.300) (10,0.300) (11,0.301) };
\addplot[color=blue,mark=*,mark size=1.5pt,line width=0.6pt]
coordinates {(1,0.232) (2,0.234) (3,0.236) (4,0.237) (5,0.237) (6,0.237) (7,0.238) (8,0.238) (9,0.239) (10,0.239) (11,0.240) };

\end{axis}
\end{tikzpicture}
\begin{tikzpicture}
\begin{axis}[
height=3cm, width=5cm,
title={\textbf{F1@20 / NCF}},
xtick={1,2,3,4,5,6,7,8,9,10,11},
xticklabels={1,5,10,15,20,25,30,35,40,45,50},
ylabel=Performance,xlabel=The number of epochs,
y tick label style={/pgf/number format/.cd,fixed,fixed zerofill,precision=2,/tikz/.cd},]

\addplot[color=black,mark=square*,mark size=1.5pt,line width=0.6pt]
coordinates {(1,0.096) (2,0.096) (3,0.097) (4,0.097) (5,0.097) (6,0.097) (7,0.098) (8,0.098) (9,0.098) (10,0.099) (11,0.099) };
\addplot[color=blue,mark=*,mark size=1.5pt,line width=0.6pt]
coordinates {(1,0.069) (2,0.069) (3,0.070) (4,0.070) (5,0.071) (6,0.071) (7,0.071) (8,0.071) (9,0.071) (10,0.071) (11,0.072) };

\end{axis}
\end{tikzpicture}
\\
\textbf{(b) ModCloth}\\

\caption{The effect of the number of dynamic update epochs on the recommendation performances of the advantaged and disadvantaged groups in Task-R.
}\label{appenfig:trunc-param-tepoch-taskR}
\vspace{-1mm}
\end{figure*}
\begin{figure*}[t]
\footnotesize
\centering
\begin{tikzpicture}
\begin{customlegend}[legend columns=5,legend style={align=left,draw=none,column sep=1ex},
        legend entries={\textbf{Advantaged group}\text{  }, \textbf{Disadvantaged group}}]
        % \addlegendimage{draw=purple,mark=square, only marks}
        % \addlegendimage{draw=red,mark=o, only marks}
        \addlegendimage{draw=black,mark=square*}   
        \addlegendimage{draw=blue,color=blue,mark=*} 
        \end{customlegend}
\end{tikzpicture}\\
    % \textbf{(a) Task-R}\vspace{-1mm}
\begin{tikzpicture}
\begin{axis}[
height=3cm, width=5cm,
title={\textbf{NDCG@20 / MF}},
xtick={1,2,3,4,5,6,7,8,9,10,11},
xticklabels={0.1,0.5,1.0,2.0,3.0,4.0,5.0},
ylabel=Performance,xlabel=Tau $\tau$,
y tick label style={/pgf/number format/.cd,fixed,fixed zerofill,precision=2,/tikz/.cd},]

\addplot[color=black,mark=square*,mark size=1.5pt,line width=0.6pt]
coordinates {(1,0.361) (2,0.666) (3,0.809) (4,0.839) (5,0.841) (6,0.845) (7,0.846) };
\addplot[color=blue,mark=*,mark size=1.5pt,line width=0.6pt]
coordinates {(1,0.855) (2,0.849) (3,0.851) (4,0.849) (5,0.848) (6,0.844) (7,0.842) };

\end{axis}
\end{tikzpicture}
\begin{tikzpicture}
\begin{axis}[
height=3cm, width=5cm,
title={\textbf{F1@20 / MF}},
xtick={1,2,3,4,5,6,7,8,9,10,11},
xticklabels={0.1,0.5,1.0,2.0,3.0,4.0,5.0},
ylabel=Performance,xlabel=Tau $\tau$,
y tick label style={/pgf/number format/.cd,fixed,fixed zerofill,precision=2,/tikz/.cd},]

\addplot[color=black,mark=square*,mark size=1.5pt,line width=0.6pt]
coordinates {(1,0.079) (2,0.236) (3,0.314) (4,0.331) (5,0.334) (6,0.337) (7,0.338) };
\addplot[color=blue,mark=*,mark size=1.5pt,line width=0.6pt]
coordinates {(1,0.343) (2,0.334) (3,0.330) (4,0.330) (5,0.330) (6,0.327) (7,0.326) };

\end{axis}
\end{tikzpicture}
\begin{tikzpicture}
\begin{axis}[
height=3cm, width=5cm,
title={\textbf{NDCG@20 / NCF}},
xtick={1,2,3,4,5,6,7,8,9,10,11},
xticklabels={0.1,0.5,1.0,2.0,3.0,4.0,5.0},
ylabel=Performance,xlabel=Tau $\tau$,
y tick label style={/pgf/number format/.cd,fixed,fixed zerofill,precision=2,/tikz/.cd},]

\addplot[color=black,mark=square*,mark size=1.5pt,line width=0.6pt]
coordinates {(1,0.839) (2,0.844) (3,0.847) (4,0.850) (5,0.851) (6,0.851) (7,0.851) };
\addplot[color=blue,mark=*,mark size=1.5pt,line width=0.6pt]
coordinates {(1,0.836) (2,0.842) (3,0.845) (4,0.845) (5,0.843) (6,0.842) (7,0.841) };
\end{axis}
\end{tikzpicture}
\begin{tikzpicture}
\begin{axis}[
height=3cm, width=5cm,
title={\textbf{F1@20 / NCF}},
xtick={1,2,3,4,5,6,7,8,9,10,11},
xticklabels={0.1,0.5,1.0,2.0,3.0,4.0,5.0},
ylabel=Performance,xlabel=Tau $\tau$,
y tick label style={/pgf/number format/.cd,fixed,fixed zerofill,precision=2,/tikz/.cd},]

\addplot[color=black,mark=square*,mark size=1.5pt,line width=0.6pt]
coordinates {(1,0.319) (2,0.336) (3,0.340) (4,0.344) (5,0.345) (6,0.346) (7,0.346) };
\addplot[color=blue,mark=*,mark size=1.5pt,line width=0.6pt]
coordinates {(1,0.333) (2,0.338) (3,0.337) (4,0.334) (5,0.332) (6,0.331) (7,0.330) };

\end{axis}
\end{tikzpicture}
\\
\textbf{(a) Movielenz}\\
\begin{tikzpicture}
\begin{axis}[
height=3cm, width=5cm,
title={\textbf{NDCG@20 / MF}},
xtick={1,2,3,4,5,6,7,8,9,10,11},
xticklabels={0.1,0.5,1.0,2.0,3.0,4.0,5.0},
ylabel=Performance,xlabel=Tau $\tau$,
y tick label style={/pgf/number format/.cd,fixed,fixed zerofill,precision=2,/tikz/.cd},]

\addplot[color=black,mark=square*,mark size=1.5pt,line width=0.6pt]
coordinates {(1,0.217) (2,0.288) (3,0.288) (4,0.286) (5,0.290) (6,0.292) (7,0.294) };
\addplot[color=blue,mark=*,mark size=1.5pt,line width=0.6pt]
coordinates {(1,0.226) (2,0.232) (3,0.232) (4,0.229) (5,0.229) (6,0.227) (7,0.229) };

\end{axis}
\end{tikzpicture}
\begin{tikzpicture}
\begin{axis}[
height=3cm, width=5cm,
title={\textbf{F1@20 / MF}},
xtick={1,2,3,4,5,6,7,8,9,10,11},
xticklabels={0.1,0.5,1.0,2.0,3.0,4.0,5.0},
ylabel=Performance,xlabel=Tau $\tau$,
y tick label style={/pgf/number format/.cd,fixed,fixed zerofill,precision=2,/tikz/.cd},]
\addplot[color=black,mark=square*,mark size=1.5pt,line width=0.6pt]
coordinates {(1,0.072) (2,0.093) (3,0.092) (4,0.093) (5,0.094) (6,0.095) (7,0.096) };
\addplot[color=blue,mark=*,mark size=1.5pt,line width=0.6pt]
coordinates {(1,0.067) (2,0.070) (3,0.070) (4,0.069) (5,0.071) (6,0.070) (7,0.070) };
\end{axis}
\end{tikzpicture}
\begin{tikzpicture}
\begin{axis}[
height=3cm, width=5cm,
title={\textbf{NDCG@20 / NCF}},
xtick={1,2,3,4,5,6,7,8,9,10,11},
xticklabels={0.1,0.5,1.0,2.0,3.0,4.0,5.0},
ylabel=Performance,xlabel=Tau $\tau$,
y tick label style={/pgf/number format/.cd,fixed,fixed zerofill,precision=2,/tikz/.cd},]
\addplot[color=black,mark=square*,mark size=1.5pt,line width=0.6pt]
coordinates {(1,0.294) (2,0.294) (3,0.296) (4,0.297) (5,0.298) (6,0.299) (7,0.299) };
\addplot[color=blue,mark=*,mark size=1.5pt,line width=0.6pt]
coordinates {(1,0.239) (2,0.234) (3,0.235) (4,0.236) (5,0.236) (6,0.235) (7,0.233) };

\end{axis}
\end{tikzpicture}
\begin{tikzpicture}
\begin{axis}[
height=3cm, width=5cm,
title={\textbf{F1@20 / NCF}},
xtick={1,2,3,4,5,6,7,8,9,10,11},
xticklabels={0.1,0.5,1.0,2.0,3.0,4.0,5.0},
ylabel=Performance,xlabel=Tau $\tau$,
y tick label style={/pgf/number format/.cd,fixed,fixed zerofill,precision=2,/tikz/.cd},]

\addplot[color=black,mark=square*,mark size=1.5pt,line width=0.6pt]
coordinates {(1,0.093) (2,0.096) (3,0.096) (4,0.096) (5,0.097) (6,0.097) (7,0.097) };
\addplot[color=blue,mark=*,mark size=1.5pt,line width=0.6pt]
coordinates {(1,0.069) (2,0.069) (3,0.070) (4,0.070) (5,0.070) (6,0.070) (7,0.069) };

\end{axis}
\end{tikzpicture}
\\
\textbf{(b) ModCloth}\\

\caption{The effect of the hyperparameter $\tau$ in our Differentiable Hit (DH) on the recommendation performances of the advantaged and disadvantaged groups in Task-R.
}\label{appenfig:trunc-param-tau-taskR}
\vspace{-1mm}
\end{figure*}
\begin{figure*}[t]
\footnotesize
\centering
\begin{tikzpicture}
\begin{customlegend}[legend columns=5,legend style={align=left,draw=none,column sep=1ex},
        legend entries={\textbf{Advantaged group}\text{  }, \textbf{Disadvantaged group}}]
        % \addlegendimage{draw=purple,mark=square, only marks}
        % \addlegendimage{draw=red,mark=o, only marks}
        \addlegendimage{draw=black,mark=square*}   
        \addlegendimage{draw=blue,color=blue,mark=*} 
        \end{customlegend}
\end{tikzpicture}\\
    % \textbf{(a) Task-R}\vspace{-1mm}
\begin{tikzpicture}
\begin{axis}[
height=3cm, width=5cm,
title={\textbf{NDCG@20 / MF}},
xtick={1,2,3,4,5,6,7,8,9,10,11},
xticklabels={4,8,12,16,20,24,28,32,36,40},
ylabel=Performance,xlabel=The number of negative items $\mu$,
y tick label style={/pgf/number format/.cd,fixed,fixed zerofill,precision=2,/tikz/.cd},]
\addplot[color=black,mark=square*,mark size=1.5pt,line width=0.6pt]
coordinates {(1,0.841) (2,0.841) (3,0.841) (4,0.841) (5,0.840) (6,0.840) (7,0.844) (8,0.843) (9,0.842) (10,0.840) };
\addplot[color=blue,mark=*,mark size=1.5pt,line width=0.6pt]
coordinates {(1,0.848) (2,0.845) (3,0.847) (4,0.845) (5,0.849) (6,0.846) (7,0.847) (8,0.850) (9,0.846) (10,0.845) };

\end{axis}
\end{tikzpicture}
\begin{tikzpicture}
\begin{axis}[
height=3cm, width=5cm,
title={\textbf{F1@20 / MF}},
xtick={1,2,3,4,5,6,7,8,9,10,11},
xticklabels={4,8,12,16,20,24,28,32,36,40},
ylabel=Performance,xlabel=The number of negative items $\mu$,
y tick label style={/pgf/number format/.cd,fixed,fixed zerofill,precision=2,/tikz/.cd},]

\addplot[color=black,mark=square*,mark size=1.5pt,line width=0.6pt]
coordinates {(1,0.334) (2,0.334) (3,0.335) (4,0.335) (5,0.334) (6,0.334) (7,0.336) (8,0.334) (9,0.333) (10,0.334) };
\addplot[color=blue,mark=*,mark size=1.5pt,line width=0.6pt]
coordinates {(1,0.330) (2,0.329) (3,0.330) (4,0.330) (5,0.330) (6,0.330) (7,0.332) (8,0.330) (9,0.331) (10,0.330) };

\end{axis}
\end{tikzpicture}
\begin{tikzpicture}
\begin{axis}[
height=3cm, width=5cm,
title={\textbf{NDCG@20 / NCF}},
xtick={1,2,3,4,5,6,7,8,9,10,11},
xticklabels={4,8,12,16,20,24,28,32,36,40},
ylabel=Performance,xlabel=The number of negative items $\mu$,
y tick label style={/pgf/number format/.cd,fixed,fixed zerofill,precision=2,/tikz/.cd},]

\addplot[color=black,mark=square*,mark size=1.5pt,line width=0.6pt]
coordinates {(1,0.851) (2,0.850) (3,0.851) (4,0.850) (5,0.851) (6,0.850) (7,0.851) (8,0.851) (9,0.849) (10,0.849) };
\addplot[color=blue,mark=*,mark size=1.5pt,line width=0.6pt]
coordinates {(1,0.843) (2,0.842) (3,0.845) (4,0.842) (5,0.847) (6,0.844) (7,0.845) (8,0.844) (9,0.845) (10,0.844) };
\end{axis}
\end{tikzpicture}
\begin{tikzpicture}
\begin{axis}[
height=3cm, width=5cm,
title={\textbf{F1@20 / NCF}},
xtick={1,2,3,4,5,6,7,8,9,10,11},
xticklabels={4,8,12,16,20,24,28,32,36,40},
ylabel=Performance,xlabel=The number of negative items $\mu$,
y tick label style={/pgf/number format/.cd,fixed,fixed zerofill,precision=2,/tikz/.cd},]

\addplot[color=black,mark=square*,mark size=1.5pt,line width=0.6pt]
coordinates {(1,0.345) (2,0.344) (3,0.345) (4,0.345) (5,0.344) (6,0.344) (7,0.344) (8,0.344) (9,0.344) (10,0.344) };
\addplot[color=blue,mark=*,mark size=1.5pt,line width=0.6pt]
coordinates {(1,0.332) (2,0.332) (3,0.334) (4,0.334) (5,0.335) (6,0.335) (7,0.336) (8,0.335) (9,0.335) (10,0.334) };

\end{axis}
\end{tikzpicture}
\\
\textbf{(a) Movielenz}\\
\begin{tikzpicture}
\begin{axis}[
height=3cm, width=5cm,
title={\textbf{NDCG@20 / MF}},
xtick={1,2,3,4,5,6,7,8,9,10,11},
xticklabels={4,8,12,16,20,24,28,32,36,40},
ylabel=Performance,xlabel=The number of negative items $\mu$,
y tick label style={/pgf/number format/.cd,fixed,fixed zerofill,precision=2,/tikz/.cd},]

\addplot[color=black,mark=square*,mark size=1.5pt,line width=0.6pt]
coordinates {(1,0.290) (2,0.291) (3,0.291) (4,0.291) (5,0.290) (6,0.288) (7,0.289) (8,0.289) (9,0.291) (10,0.288) };
\addplot[color=blue,mark=*,mark size=1.5pt,line width=0.6pt]
coordinates {(1,0.229) (2,0.225) (3,0.224) (4,0.226) (5,0.229) (6,0.229) (7,0.225) (8,0.231) (9,0.223) (10,0.222) };

\end{axis}
\end{tikzpicture}
\begin{tikzpicture}
\begin{axis}[
height=3cm, width=5cm,
title={\textbf{F1@20 / MF}},
xtick={1,2,3,4,5,6,7,8,9,10,11},
xticklabels={4,8,12,16,20,24,28,32,36,40},
ylabel=Performance,xlabel=The number of negative items $\mu$,
y tick label style={/pgf/number format/.cd,fixed,fixed zerofill,precision=2,/tikz/.cd},]
\addplot[color=black,mark=square*,mark size=1.5pt,line width=0.6pt]
coordinates {(1,0.094) (2,0.093) (3,0.095) (4,0.094) (5,0.095) (6,0.095) (7,0.095) (8,0.094) (9,0.096) (10,0.094) };
\addplot[color=blue,mark=*,mark size=1.5pt,line width=0.6pt]
coordinates {(1,0.071) (2,0.069) (3,0.071) (4,0.071) (5,0.071) (6,0.071) (7,0.069) (8,0.072) (9,0.070) (10,0.069) };
\end{axis}
\end{tikzpicture}
\begin{tikzpicture}
\begin{axis}[
height=3cm, width=5cm,
title={\textbf{NDCG@20 / NCF}},
xtick={1,2,3,4,5,6,7,8,9,10,11},
xticklabels={4,8,12,16,20,24,28,32,36,40},
ylabel=Performance,xlabel=The number of negative items $\mu$,
y tick label style={/pgf/number format/.cd,fixed,fixed zerofill,precision=2,/tikz/.cd},]
\addplot[color=black,mark=square*,mark size=1.5pt,line width=0.6pt]
coordinates {(1,0.298) (2,0.298) (3,0.297) (4,0.296) (5,0.295) (6,0.295) (7,0.294) (8,0.294) (9,0.294) (10,0.293) };
\addplot[color=blue,mark=*,mark size=1.5pt,line width=0.6pt]
coordinates {(1,0.236) (2,0.235) (3,0.234) (4,0.233) (5,0.234) (6,0.232) (7,0.232) (8,0.233) (9,0.234) (10,0.232) };

\end{axis}
\end{tikzpicture}
\begin{tikzpicture}
\begin{axis}[
height=3cm, width=5cm,
title={\textbf{F1@20 / NCF}},
xtick={1,2,3,4,5,6,7,8,9,10,11},
xticklabels={4,8,12,16,20,24,28,32,36,40},
ylabel=Performance,xlabel=The number of negative items $\mu$,
y tick label style={/pgf/number format/.cd,fixed,fixed zerofill,precision=2,/tikz/.cd},]

\addplot[color=black,mark=square*,mark size=1.5pt,line width=0.6pt]
coordinates {(1,0.097) (2,0.096) (3,0.096) (4,0.096) (5,0.096) (6,0.096) (7,0.095) (8,0.096) (9,0.095) (10,0.095) };
\addplot[color=blue,mark=*,mark size=1.5pt,line width=0.6pt]
coordinates {(1,0.070) (2,0.070) (3,0.069) (4,0.069) (5,0.069) (6,0.069) (7,0.069) (8,0.069) (9,0.069) (10,0.069) };

\end{axis}
\end{tikzpicture}
\\
\textbf{(b) ModCloth}\\

\caption{The effect of the number of negative items for each user in our fairness loss on the recommendation performances of the advantaged and disadvantaged groups in Task-R.
}\label{appenfig:trunc-param-numneg-taskR}
\vspace{-1mm}
\end{figure*}

\begin{figure*}[t]
\footnotesize
\centering
\begin{tikzpicture}
\begin{customlegend}[legend columns=5,legend style={align=left,draw=none,column sep=1ex},
        legend entries={\textbf{Advantaged group}\text{  }, \textbf{Disadvantaged group}}]
        % \addlegendimage{draw=purple,mark=square, only marks}
        % \addlegendimage{draw=red,mark=o, only marks}
        \addlegendimage{draw=black,mark=square*}   
        \addlegendimage{draw=blue,color=blue,mark=*} 
        \end{customlegend}
\end{tikzpicture}\\
    % \textbf{(a) Task-R}\vspace{-1mm}
\begin{tikzpicture}
\begin{axis}[
height=3cm, width=5cm,
title={\textbf{NDCG@20 / MF}},
xtick={1,2,3,4,5,6,7,8,9,10,11},
xticklabels={0,0.1,0.3,0.5,0.8,1.0,1.5,2.0,2.5,3.0,4.0},
ylabel=Performance,xlabel=$\lambda$,
y tick label style={/pgf/number format/.cd,fixed,fixed zerofill,precision=2,/tikz/.cd},]

\addplot[color=black,mark=square*,mark size=1.5pt,line width=0.6pt]
coordinates {(1,0.739) (2,0.740) (3,0.740) (4,0.739) (5,0.737) (6,0.735) (7,0.732) (8,0.726) (9,0.713) (10,0.702) };
\addplot[color=blue,mark=*,mark size=1.5pt,line width=0.6pt]
coordinates {(1,0.734) (2,0.737) (3,0.742) (4,0.746) (5,0.747) (6,0.751) (7,0.752) (8,0.754) (9,0.752) (10,0.751) };

\end{axis}
\end{tikzpicture}
\begin{tikzpicture}
\begin{axis}[
height=3cm, width=5cm,
title={\textbf{F1@20 / MF}},
xtick={1,2,3,4,5,6,7,8,9,10,11},
xticklabels={0,0.1,0.3,0.5,0.8,1.0,1.5,2.0,2.5,3.0,4.0},
ylabel=Performance,xlabel=$\lambda$,
y tick label style={/pgf/number format/.cd,fixed,fixed zerofill,precision=2,/tikz/.cd},]

\addplot[color=black,mark=square*,mark size=1.5pt,line width=0.6pt]
coordinates {(1,0.301) (2,0.301) (3,0.301) (4,0.300) (5,0.299) (6,0.297) (7,0.295) (8,0.292) (9,0.283) (10,0.277) };
\addplot[color=blue,mark=*,mark size=1.5pt,line width=0.6pt]
coordinates {(1,0.290) (2,0.291) (3,0.292) (4,0.291) (5,0.294) (6,0.295) (7,0.293) (8,0.295) (9,0.294) (10,0.292) };

\end{axis}
\end{tikzpicture}
\begin{tikzpicture}
\begin{axis}[
height=3cm, width=5cm,
title={\textbf{NDCG@20 / NCF}},
xtick={1,2,3,4,5,6,7,8,9,10,11},
xticklabels={0,0.1,0.3,0.5,0.8,1.0,1.5,2.0,2.5,3.0,4.0},
ylabel=Performance,xlabel=$\lambda$,
y tick label style={/pgf/number format/.cd,fixed,fixed zerofill,precision=2,/tikz/.cd},]

\addplot[color=black,mark=square*,mark size=1.5pt,line width=0.6pt]
coordinates {(1,0.742) (2,0.742) (3,0.742) (4,0.740) (5,0.739) (6,0.738) (7,0.735) (8,0.733) (9,0.728) (10,0.725) };
\addplot[color=blue,mark=*,mark size=1.5pt,line width=0.6pt]
coordinates {(1,0.736) (2,0.737) (3,0.739) (4,0.742) (5,0.741) (6,0.743) (7,0.743) (8,0.746) (9,0.748) (10,0.750) };

\end{axis}
\end{tikzpicture}
\begin{tikzpicture}
\begin{axis}[
height=3cm, width=5cm,
title={\textbf{F1@20 / NCF}},
xtick={1,2,3,4,5,6,7,8,9,10,11},
xticklabels={0,0.1,0.3,0.5,0.8,1.0,1.5,2.0,2.5,3.0,4.0},
ylabel=Performance,xlabel=$\lambda$,
y tick label style={/pgf/number format/.cd,fixed,fixed zerofill,precision=2,/tikz/.cd},]

\addplot[color=black,mark=square*,mark size=1.5pt,line width=0.6pt]
coordinates {(1,0.305) (2,0.305) (3,0.305) (4,0.305) (5,0.303) (6,0.303) (7,0.301) (8,0.298) (9,0.296) (10,0.294) };
\addplot[color=blue,mark=*,mark size=1.5pt,line width=0.6pt]
coordinates {(1,0.296) (2,0.297) (3,0.297) (4,0.298) (5,0.298) (6,0.300) (7,0.302) (8,0.303) (9,0.303) (10,0.302) };

\end{axis}
\end{tikzpicture}
\\
\textbf{(a) Movielenz}\\
\begin{tikzpicture}
\begin{axis}[
height=3cm, width=5cm,
title={\textbf{NDCG@20 / MF}},
xtick={1,2,3,4,5,6,7,8,9,10,11},
xticklabels={0,0.1,0.3,0.5,0.8,1.0,1.5,2.0,2.5,3.0,4.0},
ylabel=Performance,xlabel=$\lambda$,
y tick label style={/pgf/number format/.cd,fixed,fixed zerofill,precision=2,/tikz/.cd},]

\addplot[color=black,mark=square*,mark size=1.5pt,line width=0.6pt]
coordinates {(1,0.205) (2,0.204) (3,0.203) (4,0.203) (5,0.200) (6,0.199) (7,0.193) (8,0.187) (9,0.188) (10,0.186) };
\addplot[color=blue,mark=*,mark size=1.5pt,line width=0.6pt]
coordinates {(1,0.157) (2,0.159) (3,0.157) (4,0.158) (5,0.160) (6,0.159) (7,0.156) (8,0.153) (9,0.155) (10,0.155) };

\end{axis}
\end{tikzpicture}
\begin{tikzpicture}
\begin{axis}[
height=3cm, width=5cm,
title={\textbf{F1@20 / MF}},
xtick={1,2,3,4,5,6,7,8,9,10,11},
xticklabels={0,0.1,0.3,0.5,0.8,1.0,1.5,2.0,2.5,3.0,4.0},
ylabel=Performance,xlabel=$\lambda$,
y tick label style={/pgf/number format/.cd,fixed,fixed zerofill,precision=2,/tikz/.cd},]
\addplot[color=black,mark=square*,mark size=1.5pt,line width=0.6pt]
coordinates {(1,0.057) (2,0.057) (3,0.057) (4,0.056) (5,0.056) (6,0.056) (7,0.053) (8,0.052) (9,0.053) (10,0.052) };
\addplot[color=blue,mark=*,mark size=1.5pt,line width=0.6pt]
coordinates {(1,0.043) (2,0.043) (3,0.042) (4,0.042) (5,0.042) (6,0.042) (7,0.041) (8,0.039) (9,0.040) (10,0.040) };

\end{axis}
\end{tikzpicture}
\begin{tikzpicture}
\begin{axis}[
height=3cm, width=5cm,
title={\textbf{NDCG@20 / NCF}},
xtick={1,2,3,4,5,6,7,8,9,10,11},
xticklabels={0,0.1,0.3,0.5,0.8,1.0,1.5,2.0,2.5,3.0,4.0},
ylabel=Performance,xlabel=$\lambda$,
y tick label style={/pgf/number format/.cd,fixed,fixed zerofill,precision=2,/tikz/.cd},]
\addplot[color=black,mark=square*,mark size=1.5pt,line width=0.6pt]
coordinates {(1,0.193) (2,0.195) (3,0.196) (4,0.198) (5,0.196) (6,0.195) (7,0.193) (8,0.189) (9,0.186) (10,0.185) };
\addplot[color=blue,mark=*,mark size=1.5pt,line width=0.6pt]
coordinates {(1,0.143) (2,0.144) (3,0.150) (4,0.155) (5,0.153) (6,0.152) (7,0.152) (8,0.150) (9,0.148) (10,0.147) };

\end{axis}
\end{tikzpicture}
\begin{tikzpicture}
\begin{axis}[
height=3cm, width=5cm,
title={\textbf{F1@20 / NCF}},
xtick={1,2,3,4,5,6,7,8,9,10,11},
xticklabels={0,0.1,0.3,0.5,0.8,1.0,1.5,2.0,2.5,3.0,4.0},
ylabel=Performance,xlabel=$\lambda$,
y tick label style={/pgf/number format/.cd,fixed,fixed zerofill,precision=2,/tikz/.cd},]
\addplot[color=black,mark=square*,mark size=1.5pt,line width=0.6pt]
coordinates {(1,0.053) (2,0.053) (3,0.053) (4,0.052) (5,0.052) (6,0.052) (7,0.051) (8,0.051) (9,0.051) (10,0.051) };
\addplot[color=blue,mark=*,mark size=1.5pt,line width=0.6pt]
coordinates {(1,0.038) (2,0.039) (3,0.039) (4,0.039) (5,0.039) (6,0.038) (7,0.039) (8,0.039) (9,0.038) (10,0.038) };

\end{axis}
\end{tikzpicture}
\\
\textbf{(b) ModCloth}\\
\caption{The effect of the scaling factor $\lambda$ on the recommendation performances of the advantaged and disadvantaged groups in Task-N.
}\label{appenfig:trunc-param-lambda-taskN}
\vspace{-1mm}
\end{figure*}
\begin{figure*}[t]
\footnotesize
\centering
\begin{tikzpicture}
\begin{customlegend}[legend columns=5,legend style={align=left,draw=none,column sep=1ex},
        legend entries={\textbf{Advantaged group}\text{  }, \textbf{Disadvantaged group}}]
        % \addlegendimage{draw=purple,mark=square, only marks}
        % \addlegendimage{draw=red,mark=o, only marks}
        \addlegendimage{draw=black,mark=square*}   
        \addlegendimage{draw=blue,color=blue,mark=*} 
        \end{customlegend}
\end{tikzpicture}\\
    % \textbf{(a) Task-R}\vspace{-1mm}
\begin{tikzpicture}
\begin{axis}[
height=3cm, width=5cm,
title={\textbf{NDCG@20 / MF}},
xtick={1,2,3,4,5,6,7,8,9,10,11},
xticklabels={1,5,10,15,20,25,30,35,40,45,50},
ylabel=Performance,xlabel=The number of epochs,
y tick label style={/pgf/number format/.cd,fixed,fixed zerofill,precision=2,/tikz/.cd},]

\addplot[color=black,mark=square*,mark size=1.5pt,line width=0.6pt]
coordinates {(1,0.726) (2,0.742) (3,0.735) (4,0.710) (5,0.680) (6,0.664) (7,0.662) (8,0.661) (9,0.648) (10,0.638) (11,0.647) };
\addplot[color=blue,mark=*,mark size=1.5pt,line width=0.6pt]
coordinates {(1,0.724) (2,0.747) (3,0.751) (4,0.732) (5,0.684) (6,0.655) (7,0.667) (8,0.645) (9,0.649) (10,0.655) (11,0.633) };

\end{axis}
\end{tikzpicture}
\begin{tikzpicture}
\begin{axis}[
height=3cm, width=5cm,
title={\textbf{F1@20 / MF}},
xtick={1,2,3,4,5,6,7,8,9,10,11},
xticklabels={1,5,10,15,20,25,30,35,40,45,50},
ylabel=Performance,xlabel=The number of epochs,
y tick label style={/pgf/number format/.cd,fixed,fixed zerofill,precision=2,/tikz/.cd},]

\addplot[color=black,mark=square*,mark size=1.5pt,line width=0.6pt]
coordinates {(1,0.294) (2,0.300) (3,0.297) (4,0.287) (5,0.273) (6,0.260) (7,0.258) (8,0.256) (9,0.251) (10,0.245) (11,0.248) };
\addplot[color=blue,mark=*,mark size=1.5pt,line width=0.6pt]
coordinates {(1,0.289) (2,0.295) (3,0.295) (4,0.284) (5,0.267) (6,0.253) (7,0.256) (8,0.248) (9,0.249) (10,0.254) (11,0.245) };

\end{axis}
\end{tikzpicture}
\begin{tikzpicture}
\begin{axis}[
height=3cm, width=5cm,
title={\textbf{NDCG@20 / NCF}},
xtick={1,2,3,4,5,6,7,8,9,10,11},
xticklabels={1,5,10,15,20,25,30,35,40,45,50},
ylabel=Performance,xlabel=The number of epochs,
y tick label style={/pgf/number format/.cd,fixed,fixed zerofill,precision=2,/tikz/.cd},]

\addplot[color=black,mark=square*,mark size=1.5pt,line width=0.6pt]
coordinates {(1,0.723) (2,0.731) (3,0.738) (4,0.743) (5,0.746) (6,0.748) (7,0.747) (8,0.749) (9,0.749) (10,0.748) (11,0.749) };
\addplot[color=blue,mark=*,mark size=1.5pt,line width=0.6pt]
coordinates {(1,0.728) (2,0.736) (3,0.743) (4,0.748) (5,0.750) (6,0.753) (7,0.754) (8,0.752) (9,0.752) (10,0.751) (11,0.750) };

\end{axis}
\end{tikzpicture}
\begin{tikzpicture}
\begin{axis}[
height=3cm, width=5cm,
title={\textbf{F1@20 / NCF}},
xtick={1,2,3,4,5,6,7,8,9,10,11},
xticklabels={1,5,10,15,20,25,30,35,40,45,50},
ylabel=Performance,xlabel=The number of epochs,
y tick label style={/pgf/number format/.cd,fixed,fixed zerofill,precision=2,/tikz/.cd},]

\addplot[color=black,mark=square*,mark size=1.5pt,line width=0.6pt]
coordinates {(1,0.296) (2,0.301) (3,0.303) (4,0.304) (5,0.305) (6,0.305) (7,0.305) (8,0.305) (9,0.305) (10,0.305) (11,0.304) };
\addplot[color=blue,mark=*,mark size=1.5pt,line width=0.6pt]
coordinates {(1,0.294) (2,0.299) (3,0.300) (4,0.300) (5,0.300) (6,0.298) (7,0.298) (8,0.297) (9,0.296) (10,0.296) (11,0.296) };

\end{axis}
\end{tikzpicture}
\\
\textbf{(a) Movielenz}\\
\begin{tikzpicture}
\begin{axis}[
height=3cm, width=5cm,
title={\textbf{NDCG@20 / MF}},
xtick={1,2,3,4,5,6,7,8,9,10,11},
xticklabels={1,5,10,15,20,25,30,35,40,45,50},
ylabel=Performance,xlabel=The number of epochs,
y tick label style={/pgf/number format/.cd,fixed,fixed zerofill,precision=2,/tikz/.cd},]

\addplot[color=black,mark=square*,mark size=1.5pt,line width=0.6pt]
coordinates {(1,0.195) (2,0.195) (3,0.199) (4,0.201) (5,0.197) (6,0.199) (7,0.195) (8,0.194) (9,0.195) (10,0.196) (11,0.196) };
\addplot[color=blue,mark=*,mark size=1.5pt,line width=0.6pt]
coordinates {(1,0.145) (2,0.153) (3,0.159) (4,0.165) (5,0.166) (6,0.170) (7,0.172) (8,0.167) (9,0.163) (10,0.159) (11,0.160) };

\end{axis}
\end{tikzpicture}
\begin{tikzpicture}
\begin{axis}[
height=3cm, width=5cm,
title={\textbf{F1@20 / MF}},
xtick={1,2,3,4,5,6,7,8,9,10,11},
xticklabels={1,5,10,15,20,25,30,35,40,45,50},
ylabel=Performance,xlabel=The number of epochs,
y tick label style={/pgf/number format/.cd,fixed,fixed zerofill,precision=2,/tikz/.cd},]
\addplot[color=black,mark=square*,mark size=1.5pt,line width=0.6pt]
coordinates {(1,0.053) (2,0.055) (3,0.056) (4,0.056) (5,0.055) (6,0.055) (7,0.054) (8,0.054) (9,0.053) (10,0.053) (11,0.053) };
\addplot[color=blue,mark=*,mark size=1.5pt,line width=0.6pt]
coordinates {(1,0.038) (2,0.040) (3,0.042) (4,0.042) (5,0.043) (6,0.043) (7,0.042) (8,0.042) (9,0.042) (10,0.042) (11,0.041) };

\end{axis}
\end{tikzpicture}
\begin{tikzpicture}
\begin{axis}[
height=3cm, width=5cm,
title={\textbf{NDCG@20 / NCF}},
xtick={1,2,3,4,5,6,7,8,9,10,11},
xticklabels={1,5,10,15,20,25,30,35,40,45,50},
ylabel=Performance,xlabel=The number of epochs,
y tick label style={/pgf/number format/.cd,fixed,fixed zerofill,precision=2,/tikz/.cd},]
\addplot[color=black,mark=square*,mark size=1.5pt,line width=0.6pt]
coordinates {(1,0.193) (2,0.194) (3,0.195) (4,0.195) (5,0.196) (6,0.196) (7,0.197) (8,0.198) (9,0.199) (10,0.199) (11,0.200) };
\addplot[color=blue,mark=*,mark size=1.5pt,line width=0.6pt]
coordinates {(1,0.150) (2,0.151) (3,0.152) (4,0.154) (5,0.155) (6,0.158) (7,0.158) (8,0.158) (9,0.159) (10,0.160) (11,0.161) };

\end{axis}
\end{tikzpicture}
\begin{tikzpicture}
\begin{axis}[
height=3cm, width=5cm,
title={\textbf{F1@20 / NCF}},
xtick={1,2,3,4,5,6,7,8,9,10,11},
xticklabels={1,5,10,15,20,25,30,35,40,45,50},
ylabel=Performance,xlabel=The number of epochs,
y tick label style={/pgf/number format/.cd,fixed,fixed zerofill,precision=2,/tikz/.cd},]
\addplot[color=black,mark=square*,mark size=1.5pt,line width=0.6pt]
coordinates {(1,0.051) (2,0.051) (3,0.052) (4,0.052) (5,0.052) (6,0.053) (7,0.053) (8,0.053) (9,0.053) (10,0.054) (11,0.054) };
\addplot[color=blue,mark=*,mark size=1.5pt,line width=0.6pt]
coordinates {(1,0.038) (2,0.038) (3,0.038) (4,0.039) (5,0.039) (6,0.039) (7,0.040) (8,0.040) (9,0.040) (10,0.041) (11,0.041) };

\end{axis}
\end{tikzpicture}
\\
\textbf{(b) ModCloth}\\

\caption{The effect of the number of dynamic update epochs on the recommendation performances of the advantaged and disadvantaged groups in Task-N.
}\label{appenfig:trunc-param-tepoch-taskN}
\vspace{-1mm}
\end{figure*}
\begin{figure*}[t]
\footnotesize
\centering
\begin{tikzpicture}
\begin{customlegend}[legend columns=5,legend style={align=left,draw=none,column sep=1ex},
        legend entries={\textbf{Advantaged group}\text{  }, \textbf{Disadvantaged group}}]
        % \addlegendimage{draw=purple,mark=square, only marks}
        % \addlegendimage{draw=red,mark=o, only marks}
        \addlegendimage{draw=black,mark=square*}   
        \addlegendimage{draw=blue,color=blue,mark=*} 
        \end{customlegend}
\end{tikzpicture}\\
    % \textbf{(a) Task-R}\vspace{-1mm}
\begin{tikzpicture}
\begin{axis}[
height=3cm, width=5cm,
title={\textbf{NDCG@20 / MF}},
xtick={1,2,3,4,5,6,7,8,9,10,11},
xticklabels={0.1,0.5,1.0,2.0,3.0,4.0,5.0},
ylabel=Performance,xlabel=Tau $\tau$,
y tick label style={/pgf/number format/.cd,fixed,fixed zerofill,precision=2,/tikz/.cd},]

\addplot[color=black,mark=square*,mark size=1.5pt,line width=0.6pt]
coordinates {(1,0.235) (2,0.589) (3,0.712) (4,0.731) (5,0.735) (6,0.737) (7,0.738) };
\addplot[color=blue,mark=*,mark size=1.5pt,line width=0.6pt]
coordinates {(1,0.740) (2,0.760) (3,0.756) (4,0.755) (5,0.751) (6,0.750) (7,0.747) };

\end{axis}
\end{tikzpicture}
\begin{tikzpicture}
\begin{axis}[
height=3cm, width=5cm,
title={\textbf{F1@20 / MF}},
xtick={1,2,3,4,5,6,7,8,9,10,11},
xticklabels={0.1,0.5,1.0,2.0,3.0,4.0,5.0},
ylabel=Performance,xlabel=Tau $\tau$,
y tick label style={/pgf/number format/.cd,fixed,fixed zerofill,precision=2,/tikz/.cd},]
\addplot[color=black,mark=square*,mark size=1.5pt,line width=0.6pt]
coordinates {(1,0.049) (2,0.224) (3,0.282) (4,0.294) (5,0.297) (6,0.299) (7,0.300) };
\addplot[color=blue,mark=*,mark size=1.5pt,line width=0.6pt]
coordinates {(1,0.304) (2,0.297) (3,0.295) (4,0.294) (5,0.295) (6,0.296) (7,0.294) };

\end{axis}
\end{tikzpicture}
\begin{tikzpicture}
\begin{axis}[
height=3cm, width=5cm,
title={\textbf{NDCG@20 / NCF}},
xtick={1,2,3,4,5,6,7,8,9,10,11},
xticklabels={0.1,0.5,1.0,2.0,3.0,4.0,5.0},
ylabel=Performance,xlabel=Tau $\tau$,
y tick label style={/pgf/number format/.cd,fixed,fixed zerofill,precision=2,/tikz/.cd},]

\addplot[color=black,mark=square*,mark size=1.5pt,line width=0.6pt]
coordinates {(1,0.719) (2,0.725) (3,0.731) (4,0.735) (5,0.738) (6,0.739) (7,0.740) };
\addplot[color=blue,mark=*,mark size=1.5pt,line width=0.6pt]
coordinates {(1,0.740) (2,0.743) (3,0.745) (4,0.744) (5,0.743) (6,0.742) (7,0.742) };
\end{axis}
\end{tikzpicture}
\begin{tikzpicture}
\begin{axis}[
height=3cm, width=5cm,
title={\textbf{F1@20 / NCF}},
xtick={1,2,3,4,5,6,7,8,9,10,11},
xticklabels={0.1,0.5,1.0,2.0,3.0,4.0,5.0},
ylabel=Performance,xlabel=Tau $\tau$,
y tick label style={/pgf/number format/.cd,fixed,fixed zerofill,precision=2,/tikz/.cd},]

\addplot[color=black,mark=square*,mark size=1.5pt,line width=0.6pt]
coordinates {(1,0.272) (2,0.294) (3,0.297) (4,0.300) (5,0.303) (6,0.303) (7,0.305) };
\addplot[color=blue,mark=*,mark size=1.5pt,line width=0.6pt]
coordinates {(1,0.293) (2,0.302) (3,0.300) (4,0.302) (5,0.300) (6,0.298) (7,0.298) };

\end{axis}
\end{tikzpicture}
\\
\textbf{(a) Movielenz}\\
\begin{tikzpicture}
\begin{axis}[
height=3cm, width=5cm,
title={\textbf{NDCG@20 / MF}},
xtick={1,2,3,4,5,6,7,8,9,10,11},
xticklabels={0.1,0.5,1.0,2.0,3.0,4.0,5.0},
ylabel=Performance,xlabel=Tau $\tau$,
y tick label style={/pgf/number format/.cd,fixed,fixed zerofill,precision=2,/tikz/.cd},]

\addplot[color=black,mark=square*,mark size=1.5pt,line width=0.6pt]
coordinates {(1,0.137) (2,0.192) (3,0.195) (4,0.196) (5,0.199) (6,0.200) (7,0.199) };
\addplot[color=blue,mark=*,mark size=1.5pt,line width=0.6pt]
coordinates {(1,0.159) (2,0.156) (3,0.159) (4,0.162) (5,0.159) (6,0.156) (7,0.158) };

\end{axis}
\end{tikzpicture}
\begin{tikzpicture}
\begin{axis}[
height=3cm, width=5cm,
title={\textbf{F1@20 / MF}},
xtick={1,2,3,4,5,6,7,8,9,10,11},
xticklabels={0.1,0.5,1.0,2.0,3.0,4.0,5.0},
ylabel=Performance,xlabel=Tau $\tau$,
y tick label style={/pgf/number format/.cd,fixed,fixed zerofill,precision=2,/tikz/.cd},]
\addplot[color=black,mark=square*,mark size=1.5pt,line width=0.6pt]
coordinates {(1,0.038) (2,0.053) (3,0.053) (4,0.054) (5,0.056) (6,0.056) (7,0.056) };
\addplot[color=blue,mark=*,mark size=1.5pt,line width=0.6pt]
coordinates {(1,0.042) (2,0.043) (3,0.041) (4,0.041) (5,0.042) (6,0.040) (7,0.041) };
\end{axis}
\end{tikzpicture}
\begin{tikzpicture}
\begin{axis}[
height=3cm, width=5cm,
title={\textbf{NDCG@20 / NCF}},
xtick={1,2,3,4,5,6,7,8,9,10,11},
xticklabels={0.1,0.5,1.0,2.0,3.0,4.0,5.0},
ylabel=Performance,xlabel=Tau $\tau$,
y tick label style={/pgf/number format/.cd,fixed,fixed zerofill,precision=2,/tikz/.cd},]
\addplot[color=black,mark=square*,mark size=1.5pt,line width=0.6pt]
coordinates {(1,0.192) (2,0.191) (3,0.193) (4,0.194) (5,0.195) (6,0.196) (7,0.197) };
\addplot[color=blue,mark=*,mark size=1.5pt,line width=0.6pt]
coordinates {(1,0.156) (2,0.152) (3,0.153) (4,0.152) (5,0.152) (6,0.152) (7,0.153) };

\end{axis}
\end{tikzpicture}
\begin{tikzpicture}
\begin{axis}[
height=3cm, width=5cm,
title={\textbf{F1@20 / NCF}},
xtick={1,2,3,4,5,6,7,8,9,10,11},
xticklabels={0.1,0.5,1.0,2.0,3.0,4.0,5.0},
ylabel=Performance,xlabel=Tau $\tau$,
y tick label style={/pgf/number format/.cd,fixed,fixed zerofill,precision=2,/tikz/.cd},]

\addplot[color=black,mark=square*,mark size=1.5pt,line width=0.6pt]
coordinates {(1,0.051) (2,0.051) (3,0.051) (4,0.052) (5,0.052) (6,0.052) (7,0.052) };
\addplot[color=blue,mark=*,mark size=1.5pt,line width=0.6pt]
coordinates {(1,0.038) (2,0.038) (3,0.038) (4,0.039) (5,0.038) (6,0.039) (7,0.039) };

\end{axis}
\end{tikzpicture}
\\
\textbf{(b) ModCloth}\\

\caption{The effect of the hyperparameter $\tau$ in our Differentiable Hit (DH) on the recommendation performances of the advantaged and disadvantaged groups in Task-N.
}\label{appenfig:trunc-param-tau-taskN}
\vspace{-1mm}
\end{figure*}
\begin{figure*}[t]
\footnotesize
\centering
\begin{tikzpicture}
\begin{customlegend}[legend columns=5,legend style={align=left,draw=none,column sep=1ex},
        legend entries={\textbf{Advantaged group}\text{  }, \textbf{Disadvantaged group}}]
        % \addlegendimage{draw=purple,mark=square, only marks}
        % \addlegendimage{draw=red,mark=o, only marks}
        \addlegendimage{draw=black,mark=square*}   
        \addlegendimage{draw=blue,color=blue,mark=*} 
        \end{customlegend}
\end{tikzpicture}\\
    % \textbf{(a) Task-R}\vspace{-1mm}
\begin{tikzpicture}
\begin{axis}[
height=3cm, width=5cm,
title={\textbf{NDCG@20 / MF}},
xtick={1,2,3,4,5,6,7,8,9,10,11},
xticklabels={4,8,12,16,20,24,28,32,36,40},
ylabel=Performance,xlabel=The number of negative items $\mu$,
y tick label style={/pgf/number format/.cd,fixed,fixed zerofill,precision=2,/tikz/.cd},]
\addplot[color=black,mark=square*,mark size=1.5pt,line width=0.6pt]
coordinates {(1,0.735) (2,0.735) (3,0.739) (4,0.741) (5,0.732) (6,0.732) (7,0.738) (8,0.738) (9,0.737) (10,0.734) };
\addplot[color=blue,mark=*,mark size=1.5pt,line width=0.6pt]
coordinates {(1,0.751) (2,0.743) (3,0.743) (4,0.743) (5,0.743) (6,0.741) (7,0.743) (8,0.747) (9,0.752) (10,0.743) };

\end{axis}
\end{tikzpicture}
\begin{tikzpicture}
\begin{axis}[
height=3cm, width=5cm,
title={\textbf{F1@20 / MF}},
xtick={1,2,3,4,5,6,7,8,9,10,11},
xticklabels={4,8,12,16,20,24,28,32,36,40},
ylabel=Performance,xlabel=The number of negative items $\mu$,
y tick label style={/pgf/number format/.cd,fixed,fixed zerofill,precision=2,/tikz/.cd},]

\addplot[color=black,mark=square*,mark size=1.5pt,line width=0.6pt]
coordinates {(1,0.297) (2,0.297) (3,0.297) (4,0.298) (5,0.295) (6,0.294) (7,0.297) (8,0.295) (9,0.297) (10,0.295) };
\addplot[color=blue,mark=*,mark size=1.5pt,line width=0.6pt]
coordinates {(1,0.295) (2,0.294) (3,0.295) (4,0.293) (5,0.293) (6,0.294) (7,0.294) (8,0.291) (9,0.298) (10,0.294) };

\end{axis}
\end{tikzpicture}
\begin{tikzpicture}
\begin{axis}[
height=3cm, width=5cm,
title={\textbf{NDCG@20 / NCF}},
xtick={1,2,3,4,5,6,7,8,9,10,11},
xticklabels={4,8,12,16,20,24,28,32,36,40},
ylabel=Performance,xlabel=The number of negative items $\mu$,
y tick label style={/pgf/number format/.cd,fixed,fixed zerofill,precision=2,/tikz/.cd},]
\addplot[color=black,mark=square*,mark size=1.5pt,line width=0.6pt]
coordinates {(1,0.738) (2,0.735) (3,0.742) (4,0.738) (5,0.739) (6,0.740) (7,0.737) (8,0.737) (9,0.736) (10,0.737) };
\addplot[color=blue,mark=*,mark size=1.5pt,line width=0.6pt]
coordinates {(1,0.743) (2,0.744) (3,0.744) (4,0.744) (5,0.747) (6,0.743) (7,0.746) (8,0.748) (9,0.746) (10,0.742) };
\end{axis}
\end{tikzpicture}
\begin{tikzpicture}
\begin{axis}[
height=3cm, width=5cm,
title={\textbf{F1@20 / NCF}},
xtick={1,2,3,4,5,6,7,8,9,10,11},
xticklabels={4,8,12,16,20,24,28,32,36,40},
ylabel=Performance,xlabel=The number of negative items $\mu$,
y tick label style={/pgf/number format/.cd,fixed,fixed zerofill,precision=2,/tikz/.cd},]

\addplot[color=black,mark=square*,mark size=1.5pt,line width=0.6pt]
coordinates {(1,0.303) (2,0.302) (3,0.302) (4,0.302) (5,0.302) (6,0.301) (7,0.301) (8,0.301) (9,0.301) (10,0.300) };
\addplot[color=blue,mark=*,mark size=1.5pt,line width=0.6pt]
coordinates {(1,0.300) (2,0.300) (3,0.301) (4,0.302) (5,0.303) (6,0.303) (7,0.303) (8,0.302) (9,0.302) (10,0.301) };

\end{axis}
\end{tikzpicture}
\\
\textbf{(a) Movielenz}\\
\begin{tikzpicture}
\begin{axis}[
height=3cm, width=5cm,
title={\textbf{NDCG@20 / MF}},
xtick={1,2,3,4,5,6,7,8,9,10,11},
xticklabels={4,8,12,16,20,24,28,32,36,40},
ylabel=Performance,xlabel=The number of negative items $\mu$,
y tick label style={/pgf/number format/.cd,fixed,fixed zerofill,precision=2,/tikz/.cd},]

\addplot[color=black,mark=square*,mark size=1.5pt,line width=0.6pt]
coordinates {(1,0.199) (2,0.199) (3,0.196) (4,0.200) (5,0.198) (6,0.199) (7,0.202) (8,0.204) (9,0.203) (10,0.198) };
\addplot[color=blue,mark=*,mark size=1.5pt,line width=0.6pt]
coordinates {(1,0.159) (2,0.158) (3,0.162) (4,0.160) (5,0.161) (6,0.167) (7,0.162) (8,0.161) (9,0.166) (10,0.166) };

\end{axis}
\end{tikzpicture}
\begin{tikzpicture}
\begin{axis}[
height=3cm, width=5cm,
title={\textbf{F1@20 / MF}},
xtick={1,2,3,4,5,6,7,8,9,10,11},
xticklabels={4,8,12,16,20,24,28,32,36,40},
ylabel=Performance,xlabel=The number of negative items $\mu$,
y tick label style={/pgf/number format/.cd,fixed,fixed zerofill,precision=2,/tikz/.cd},]
\addplot[color=black,mark=square*,mark size=1.5pt,line width=0.6pt]
coordinates {(1,0.056) (2,0.054) (3,0.054) (4,0.054) (5,0.055) (6,0.054) (7,0.056) (8,0.055) (9,0.057) (10,0.055) };
\addplot[color=blue,mark=*,mark size=1.5pt,line width=0.6pt]
coordinates {(1,0.042) (2,0.042) (3,0.043) (4,0.043) (5,0.043) (6,0.045) (7,0.043) (8,0.043) (9,0.044) (10,0.046) };
\end{axis}
\end{tikzpicture}
\begin{tikzpicture}
\begin{axis}[
height=3cm, width=5cm,
title={\textbf{NDCG@20 / NCF}},
xtick={1,2,3,4,5,6,7,8,9,10,11},
xticklabels={4,8,12,16,20,24,28,32,36,40},
ylabel=Performance,xlabel=The number of negative items $\mu$,
y tick label style={/pgf/number format/.cd,fixed,fixed zerofill,precision=2,/tikz/.cd},]
\addplot[color=black,mark=square*,mark size=1.5pt,line width=0.6pt]
coordinates {(1,0.195) (2,0.194) (3,0.193) (4,0.191) (5,0.192) (6,0.191) (7,0.191) (8,0.191) (9,0.190) (10,0.190) };
\addplot[color=blue,mark=*,mark size=1.5pt,line width=0.6pt]
coordinates {(1,0.152) (2,0.153) (3,0.152) (4,0.152) (5,0.151) (6,0.149) (7,0.149) (8,0.151) (9,0.150) (10,0.151) };

\end{axis}
\end{tikzpicture}
\begin{tikzpicture}
\begin{axis}[
height=3cm, width=5cm,
title={\textbf{F1@20 / NCF}},
xtick={1,2,3,4,5,6,7,8,9,10,11},
xticklabels={4,8,12,16,20,24,28,32,36,40},
ylabel=Performance,xlabel=The number of negative items $\mu$,
y tick label style={/pgf/number format/.cd,fixed,fixed zerofill,precision=2,/tikz/.cd},]

\addplot[color=black,mark=square*,mark size=1.5pt,line width=0.6pt]
coordinates {(1,0.052) (2,0.052) (3,0.051) (4,0.051) (5,0.051) (6,0.051) (7,0.051) (8,0.051) (9,0.051) (10,0.051) };
\addplot[color=blue,mark=*,mark size=1.5pt,line width=0.6pt]
coordinates {(1,0.038) (2,0.039) (3,0.038) (4,0.038) (5,0.038) (6,0.037) (7,0.038) (8,0.038) (9,0.038) (10,0.038) };

\end{axis}
\end{tikzpicture}
\\
\textbf{(b) ModCloth}\\

\caption{The effect of the number of negative items for each user in our fairness loss on the recommendation performances of the advantaged and disadvantaged groups in Task-N.
}\label{appenfig:trunc-param-numneg-taskN}
\vspace{-1mm}
\end{figure*}

\end{document}